%% file: main.tex
\documentclass[11pt,letterpaper]{article}
\usepackage[margin=1in]{geometry}
\usepackage{setspace}
\usepackage{dsfont}
\usepackage{amsmath}
\usepackage{amsthm}
\usepackage{amssymb}
\usepackage{amsfonts}
\usepackage{url}
\usepackage{enumerate}
\usepackage{subcaption}
\usepackage{bm}
\usepackage{graphicx}
\usepackage[utf8]{inputenc} 
\usepackage[T1]{fontenc}    
\usepackage{hyperref}       
\usepackage{booktabs}       
\usepackage{nicefrac}       
\usepackage{microtype}      
\usepackage{xcolor}         
\usepackage{algorithm}
\usepackage[noend]{algorithmic}
\usepackage{cleveref}
\usepackage{natbib}
\usepackage{tikz}
\input{notations}

\allowdisplaybreaks

\title{Privacy Filters are Captured by Residues: A Characterization of Free Natural Filters and the Cost of Adaptivity\thanks{Authors BH, EL, PM, and PT in alphabetical order.}}

\author{
    Matthew Regehr\thanks{University of Waterloo. \texttt{matt.regehr@uwaterloo.ca}. Supported by an NSERC CGS-D scholarship.} \and
    Bingshan Hu\thanks{University of British Columbia. \texttt{bingsha1@cs.ubc.ca}} \and
    Ethan Leeman\thanks{Google Research. \texttt{ethanleeman@google.com}} \and
    Pasin Manurangsi\thanks{Google Research. \texttt{pasin@google.com}} \and
    Pierre Tholoniat\thanks{Google. \texttt{pierre@cs.columbia.edu}. Work done while at Columbia University.} \and
    Mathias L\'ecuyer\thanks{University of British Columbia. \texttt{mathias.lecuyer@ubc.ca}}
}

\begin{document}
    \maketitle

    \begin{abstract}
        We study privacy filters, which enable privacy accounting for differentially private (DP) mechanisms with adaptively chosen privacy characteristics. We show that the worst-case privacy loss of interacting with any privacy filter is characterized by a Bellman equation. We apply this theory to develop residue filters, which unify existing privacy filters. We develop the Gaussian DP (GDP) residue filter, which strictly improves upon the na\"ive GDP filter. We show that, in the average case, this new filter asymptotically permits as many pure DP queries as in the non-adaptive case and a factor $\pi/2$ more pure DP queries than the na\"ive GDP filter. We also show that residue filters capture the natural filter, which promises greater utility by leveraging exact privacy accounting techniques. Earlier privacy filters consider only simple privacy parameters such as R\'enyi-DP or GDP parameters. Natural filters account for the entire privacy profile of every query, promising more efficient use of a given privacy budget. We show that, contrary to other forms of DP, natural privacy filters are not free in general. We present a characterization of when a family of private queries admits free natural filters for a given budget. In particular, only families of privacy mechanisms that are totally-ordered when composed admit free natural privacy filters with respect to an arbitrary privacy budget. Finally, we show that, while the natural approximate-DP filter can fail in the presence of an adaptive adversary, it cannot fail too badly: the output remains approximate-DP with parameters at most poly-logarithmically worse than the intended privacy parameters.
    \end{abstract}

    \input{sections/intro}
    \input{sections/overview}
    \input{sections/prelims}
    \input{sections/constrained_adversaries}
    \input{sections/residue_filter}
    \input{sections/natural_filter}
    \input{sections/implications}
    \input{sections/upper_bound}
    \input{sections/total_ordering_proof}
    \input{sections/conclusion}

    \bibliographystyle{alpha}
    \bibliography{refs/zotero_export,refs/manual_refs,refs/biblio}

    \appendix
    \input{sections/app_prelims}
    \input{sections/app_constrained_adversaries}
    \input{sections/app_gdp_residue}
    \input{sections/app_implications}
    \input{sections/app_symmetry}
\end{document}

%% file: notations.tex
\newcommand{\conv}{\operatorname{conv}}

\newcommand{\pld}{\operatorname{PLD}}
\newcommand{\PLD}{\operatorname{PLD}}
\newcommand{\Adv}{\operatorname{Adv}}

\newcommand{\Nat}{\operatorname{Nat}}
\newcommand{\Update}{\operatorname{Update}}
\renewcommand{\epsilon}{\varepsilon}

\newcommand{\algcomment}[1]{\hfill{$\triangleright$~{\em #1}}}
\newcommand{\YIELD}{\STATE \textbf{yield} }
\newcommand{\BREAK}{\STATE \textbf{break} }

\newtheorem{lemma}{Lemma}
\newtheorem{theorem}{Theorem}
\newtheorem{corollary}{Corollary}
\newtheorem{definition}{Definition}
\newtheorem{proposition}{Proposition}
\newtheorem{remark}[theorem]{Remark}
\crefname{remark}{Remark}{Remarks}
\newtheorem{observation}[theorem]{Observation}
\crefname{observation}{Observation}{Observations}

\newcommand{\E}{\mathop{\mathbb{E}}}
\newcommand{\mP}{\mathop{\mathbb{P}}}
\newcommand{\BC}{\mathop{BC}}

\newcommand{\bN}{\mathbb{N}}
\newcommand{\bP}{\mathbb{P}}
\newcommand{\bR}{\mathbb{R}}
\newcommand{\bZ}{\mathbb{Z}}

\newcommand{\cA}{\mathcal{A}}
\newcommand{\cB}{\mathcal{B}}

\newcommand{\cF}{\mathcal{F}}
\newcommand{\cL}{\mathcal{L}}
\newcommand{\cLall}{\cL_{\mathrm{all}}}
\newcommand{\cLsym}{\cL_{\mathrm{sym}}}
\newcommand{\cLpure}{\cL_{\mathrm{pure}}}
\newcommand{\cLapprox}{\cL_{\mathrm{approx}}}
\newcommand{\cLgdp}{\cL_{\mathrm{gdp}}}

\newcommand{\cY}{\mathcal{Y}}

\newtheoremstyle{TheoremNum}
{\topsep}{\topsep}              
{\itshape}                      
{}                              
{\bfseries}                     
{.}                             
{ }                             
{\thmname{#1}\thmnote{ \bfseries #3}}
\theoremstyle{TheoremNum}

\DeclareMathOperator{\supp}{supp}

\newcommand{\eps}{\epsilon}
\newcommand{\bp}{\mathbf{p}}
\newcommand{\beps}{\boldsymbol{\eps}}
\newcommand{\bell}{\boldsymbol{\ell}}
\newcommand{\R}{\mathbb{R}}
\newcommand{\N}{\mathbb{N}}
\newcommand{\cD}{\mathcal{D}}
\newcommand{\cE}{\mathcal{E}}
\newcommand{\RRd}{\RR^{\downarrow}}

\newcommand{\mech}{\mathcal{M}}

\newcommand{\Iddist}{\text{Id}}

\newcommand{\bj}{\mathbf{j}}
\newcommand{\bzero}{\mathbf{0}}
\newcommand{\bone}{\mathbf{1}}

\newcommand{\Ber}{\mathrm{Ber}}
\newcommand{\Var}{\mathrm{Var}}
\def \Paren#1{{\left({#1}\right)}}

\newcommand{\bb}{\mathbf{b}}
\newcommand{\bu}{\mathbf{u}}
\newcommand{\bv}{\mathbf{v}}
\newcommand{\RR}{\mathrm{RR}}

\newcommand{\teps}{\tilde{\eps}}
\newcommand{\disc}{\mathrm{disc}}
\newcommand{\AddFilter}{\operatorname{Discrete}}
\newcommand{\ViewAddFilter}{M^{\mathrm{discrete}}_{\cA}}

%% file: sections/intro.tex
\section{Introduction}

Differential Privacy (DP) \cite{dwork2006calibrating} is a rigorous notion of data privacy that protects against a host of attacks including membership inference \cite{homer2008resolving,dwork2015robust,carlini2022membership} and data reconstruction \cite{dinur2003revealing,dwork2007price,carlini2021extracting,balle2022reconstructing,dick2023confidence,nasr2023scalable}.
As such, there is a sustained push to integrate and deploy DP in a large number of applications \cite{opendp-registry}, from machine learning (ML) and artificial intelligence (AI) models \cite{kairouz2021practical,sinha2025vaultgemma}, to analytics collection \cite{rogers2020linkedin,apple-emoji,ding2017collecting}, to data releases \cite{adeleye2023publishing,abowd20222020}, as well as in general libraries and systems supporting those applications \cite{gaboardi2020programming,wilson2020differentially,jax-privacy2022github,tholoniat2024cookie,ghazi2025differential}.
This practical DP push also motivates important theoretical progress.

A prominent direction of theoretical improvement comes from alternative quantifications of privacy, especially those that offer tighter or even lossless accounting of composition. The classical $(\epsilon, \delta)$-DP \cite{dwork2006our} (or approximate DP) was first refined by forms of privacy derived from R\'enyi divergence including zCDP \cite{DworkR16,bun2016concentrated} and RDP \cite{mironov2017renyi}.
Later, exact forms of privacy accounting were developed including privacy profiles, privacy loss distributions (PLDs) \cite{sommer2019privacy}, and $f$-DP \cite{DongRS19}.
These exact forms of accounting are practically and theoretically appealing as they tightly capture composition as well as resilience to membership inference attacks framed as hypothesis tests \cite{WassermanZ10,DongRS19}.
Exact composition, while more complex, can be computed efficiently to arbitrary precision by numerical methods \cite{KoskelaJH20,gopi2021numerical}.

Theoretical improvements enabling new forms of composition are also necessary for practical DP usage.
Concurrent composition \cite{vadhan2021concurrent} supports composition of DP mechanisms running over multiple rounds of parallel interactions with the user---such as the sparse vector technique \cite{dwork2009complexity}.
This is required in systems that handle queries from multiple parties simultaneously where it is not possible to enforce a sequential ordering of the queries. Individual users may also run DP computations concurrently with a long-running, multi-step mechanism (see e.g. \cite{kostopoulou2023turbo}).

A natural question is whether developments in privacy accounting are compatible with new forms of composition.
Ideally, more flexible forms of composition come for free, i.e. existing results for standard composition can be applied to flexibly composed mechanisms without additional privacy costs.
Free flexible composition also enables reuse of existing results including tight analyses for specific mechanisms and empirical computational results.
As it turns out, concurrent composition is indeed free for all common forms of privacy \cite{vadhan2023concurrent,lyu2022composition}.
Specifically, Vadhan and Zhang \cite{vadhan2023concurrent} show that for both $f$-DP (including PLD and approximate DP) and RDP (including zCDP) notions of privacy, concurrent composition can be seen as a post-processing of standard composition and thus has the same cost.

Another form of flexible composition and the subject of this work is fully adaptive composition, which concerns the composition of DP mechanisms with interactively chosen privacy costs. Standard composition on the other hand requires the privacy costs to be fixed in advance.
This is insufficient for long-running applications in which analysts interactively re-assess the privacy budget allocated to queries in order to meet operational needs \cite{lecuyer2019privacy,kostopoulou2023turbo,kuchler2024cohere}.
Fully adaptive analysis can also lead to substantial privacy loss savings through advanced versions of parallel composition \cite{lecuyer2019privacy} and individual (personalized) DP \cite{ebadi2015differential}, which enables data dependent privacy accounting in ML model training \cite{FeldmanZ21,yu2022individual} and measurement systems \cite{tholoniat2024cookie,ghazi2025differential}.
Fully adaptive composition is formalized in terms of \emph{privacy filters} \cite{RogersRUV16}.
A privacy filter is an algorithm that monitors an interaction between an analyst and a data curator and permits queries as long as a privacy budget condition is maintained.
The stopping condition should be as permissive as possible such that the resulting interaction provably satisfies a meaningful privacy guarantee. We will informally say that a filter is free if, for a given budget, the filter's stopping condition translates directly into a bound on the privacy cost of the entire filter-analyst interaction with no extra costs. We will further call a filter universally free if this is true for any choice of privacy budget.

As with concurrent composition, a key research question is whether advanced privacy accounting techniques can be used to construct free privacy filters for fully adaptive composition.
Haney et al. \cite{haney2023concurrent} show that concurrent, adaptive composition reduces to adaptive composition.
Therefore, the full expressivity of both concurrency and adaptivity hinges on the existence of free privacy filters for different notions of DP.
However, privacy filters are not as well understood as concurrent composition.
We know that adaptive composition under privacy filters is free for $\epsilon$-DP \cite{RogersRUV16} and R\'enyi-based DP \cite{FeldmanZ21,lecuyer2021practical} (zCDP, RDP). The GDP filter, which applies exact ($f$-DP, PLD) accounting to the Gaussian mechanism, is also free \cite{ST22,KoskelaTH23}. This filter was also recently extended to support approximate GDP budgets \cite{tran2026f}. However, for $(\epsilon, \delta)$-DP filters, the best known filter by Whitehouse et al. \cite{WRRW23} extends the RDP filter of Feldman and Zrnic~\cite{FeldmanZ21} to approximate RDP and zCDP, enabling a filter based on lossy translation from $(\epsilon, \delta)$-DP to approximate zCDP \cite{bun2016concentrated}.
This yields an $(\epsilon, \delta)$-DP privacy filter that is asymptotically free but with worse constants than non-adaptive composition due to translation loss.
Even less is known about filters under exact forms of privacy.
We study the \emph{natural filter} in which the realized privacy costs of each query are composed by exact accounting and the realized total cost is compared to a fixed budget.

We make the following important contributions towards fully understanding privacy filters:
\begin{itemize}
    \item First, in~\S\ref{sec:adversary_theory}, we present a general-purpose theory of adaptive interaction. We show that the exact privacy cost of interacting with any filter is characterized by a Bellman equation.
    \item In \S\ref{sec:residue_filters}, we show that existing privacy filters are subsumed by \emph{residue filters}, which perform exact accounting under a residual update condition. We develop a new GDP residue filter and show that, in the average case, it permits asymptotically a factor $\pi/2$ more pure DP queries than the na\"ive GDP filter and as many pure DP queries as if they were issued non-adaptively.
    \item In \S\ref{sec:natural_filter} we apply our theory of constrained adversaries to calculate the exact privacy cost of an interaction with the natural filter. We show that the natural filter is universally free for any budget only for queries that are closed under composition and totally-ordered, in which case the natural filter is captured by a residue filter.
    \item In \S\ref{sec:natural-apx}, we consider the widely used natural $(\eps, \delta)$-DP filter, i.e. the filter based on exact composition and $(\eps, \delta)$-DP budget. We show that, while this filter can fail, it does not fail too badly: the interaction still satisfies $(\chi \cdot \eps, \chi \cdot \delta)$-DP where $\chi$ is only poly-logarithmic in $1/\eps,1/\delta$ and the number of rounds of adaptivity.
\end{itemize}

Our work also includes some minor contributions including that PLDs can be assigned a natural ordering satisfying a completeness property enabling one to take a supremum of PLDs (see \Cref{prop:sup-conv}). As far as we are aware this property has not been documented in the privacy accounting literature. We also record in \S\ref{sec:implications-special-case} a number of families of queries and budgets that do and do not admit free natural filters. Our key result showing the necessity of total-ordering for universally free natural filters is given in \S\ref{sec:total_ordering_proof} and involves a privacy gap amplification technique that may be of general interest.

We note that a recent preprint independent of and simultaneous to our work also studies the natural filter \cite{tran2026f} and contains a similar total-ordering result under the stronger assumption that some query has unbounded support. We identify a gap in their argument, which we discuss further in~\S\ref{sec:natural_filter}. Their argument can be corrected by further strengthening the assumption to require that some permissible query has continuous support. Our result requires only a very mild non-degeneracy condition. In particular, our result applies to discretized finite-support queries that arise in numerical implementations of privacy accounting. Our more general result relies on topological techniques that we develop in~\S\ref{sec:total_ordering_proof}.

%% file: sections/overview.tex
\section{Overview of Results and Techniques}
\label{sec:overview}

In this section, we informally develop our main research agenda, results, and techniques. Careful definitions are delayed until the preliminaries in~\Cref{sec:preliminaries}.

\subsection{Background}

Differential privacy is the study of randomized algorithms (called mechanisms) $M : \cD \to \cY$ where $\cD$ denotes a space of datasets containing sensitive information. Differential privacy requires that an adversary should not gain much statistical advantage for inferring whether the mechanism was run on dataset $D_1 \in \cD$ or a second dataset $D_2 \in \cD$, when the datasets are close to each other in some sense. That is, the distributions of $M(D_1)$ and $M(D_2)$ should be difficult to distinguish from one another in the hypothesis testing sense, and this should hold uniformly for any close datasets. The measure of closeness for datasets is determined by a neighbouring relation but this turns out not to be relevant to our work. In fact, our conclusions hold for an arbitrary pair of datasets. For this reason, we fix throughout a single arbitrary pair of datasets $(D_1, D_2) \in \cD \times \cD$ and view the privacy characteristics of a mechanism $M$ in terms of the hypothesis testing problem $M(D_1)$ vs. $M(D_2)$.

As a starting point for our work, consider arguably the simplest form of privacy, namely pure DP. A mechanism $M$ is called $\eps$-DP if $\lvert \log\left( \frac{dM(D_1)}{dM(D_2)}(y) \right) \rvert \leq \eps$ for all $y \in \cY$. This form of privacy simplifies the full privacy characteristics of the mechanism $M$ down to its absolute maximum privacy loss, measured as a log-likelihood ratio. While this may lose considerable information, it is simple to understand and composes well. Given a sequence of mechanisms $M_1, M_2, \dots, M_k$ where each query $M_i(\cdot; y_1, \dots, y_{i - 1})$ may depend on the result of previous queries $y^i := (y_1, \dots, y_{i - 1})$ but satisfies $\eps_i$-DP independently of $y^i$, then the adaptive composition $M_1 \otimes \dots \otimes M_k$ given by running these queries in sequence is known to satisfy $(\eps_1 + \dots + \eps_k)$-DP \cite{dwork2014algorithmic}. We call this form of adaptivity partial adaptivity because each round of interaction with the data may depend on previously observed randomized outcomes whereas the privacy cost of each round must still be fixed in advance.

It is reasonable to ask whether pure DP composition remains valid in the fully adaptive setting where the pure DP privacy cost of each query $\eps_i$ may depend on the history $y^i$. In this case, it is necessary to introduce a \textit{privacy filter}, an intermediary that monitors the interaction and halts it in order to maintain a predetermined privacy budget for the entire process.
The pure DP filter is initialized with a fixed $\eps$-DP budget and allows a fully adaptive sequence of $\eps_1, \eps_2, \dots, \eps_t$-DP queries as long as $\eps_1 + \dots + \eps_t \leq \eps$. Rogers et al. \cite{RogersRUV16} show that any interaction with the pure DP filter has a privacy cost satisfying $\eps$-DP.

While pure DP is intuitive and composable, it captures only a very rough sense of a mechanism's privacy characteristics and many practical mechanisms do not satisfy pure DP at all. A number of statistical problems are known to exhibit a significant gap between pure DP solutions and so-called \textit{approximate DP} solutions \cite{SteinkeU17a}.

A more advanced notion of privacy is Gaussian DP (GDP). Informally, we say that a mechanism $M$ satisfies $\mu$-GDP if it is at least as difficult to distinguish $M(D_1)$ from $M(D_2)$ as it is to distinguish $N(0, 1)$ from $N(\mu, 1)$. The Gaussian mechanism, which adds normally distributed noise to a statistic, is a workhorse of DP and has a privacy cost that is captured exactly by its GDP parameter $\mu$. Moreover, due to the central limit theorem, the privacy cost of composing many small queries tends toward that of the Gaussian mechanism. Therefore GDP occupies a central role in DP. Accordingly, GDP composes well. For a partially adaptive sequence of $\mu_1, \dots, \mu_t$-GDP queries, the composition satisfies $\sqrt{\mu_1^2 + \dots \mu_t^2}$-GDP. Furthermore, Koskela et al. as well as Smith and Thakurta \cite{KoskelaTH23,ST22} show that this remains the case under full adaptivity. That is, for a fixed $\mu$, any interaction with the GDP filter that verifies $\mu_1^2 + \dots + \mu_t^2 \leq \mu^2$ satisfies $\mu$-GDP, even when $\mu_i = \mu_i(y_1, \dots, y_{i - 1})$ depends on previous results in the interaction.

Unfortunately, while scalar measures of privacy such as $\varepsilon$-DP and $\mu$-GDP are simple and compose well, these are only an approximation of the true privacy characteristics of a typical mechanism. To maintain provable privacy guarantees, scalar privacy measures must generally be taken as an overestimate.
For example, consider the following simplified form of the Poisson subsampled Gaussian mechanism, which is a building block of private machine learning:
\begin{align*}
    M : [0, 1] \to \bR, M(x) = x \cdot \textrm{Bern}(p) + N(0, \sigma^2).
\end{align*}
For a very small $p \in (0, 1)$, one should expect significant \textit{privacy amplification by subsampling}, meaning in this case a gain in privacy due to the high probability that the mechanism ignores its input. However, it is well-known that GDP among other scalar privacy parameters cannot capture this \cite{BunDRS18}. Indeed, by combining the standard analysis of Poisson subsampling (see e.g. \cite{DongRS19}) with the GDP characterization given in \cite[Theorem 3.3]{liu2022identification}, it is straightforward to show that, for any $p \in (0, 1)$, the minimal $\mu$ so that $M$ satisfies $\mu$-GDP is exactly $\mu = \frac{1}{\sigma}$. In general, scalar measures of privacy may offer a poor approximation.
Consequently, privacy filters built on scalar privacy composition theorems may lead to unnecessarily early stopping times, limiting utility.

An important development in privacy accounting is the privacy loss distribution (PLD) \cite{DworkR16}. The PLD of a pair of distributions, denoted $\pld(P \parallel Q)$, is the distribution $\log\left( \frac{dP}{dQ}(\omega) \right), \ \omega \sim P$ of the log-likelihood ratio. Due to the Neyman--Pearson lemma, $\pld(M) := \pld(M(D_1) \parallel M(D_2))$ contains complete information about the hypothesis testing problem $(M(D_1), M(D_2))$ and thus encodes the exact privacy characteristics of $M$. Scalar privacy parameters including the pure DP parameter and the GDP parameter of a mechanism can all be recovered efficiently from the PLD. More importantly, PLDs enable an exact composition theorem for partial adaptivity. That is, for any adaptive sequence $M_1, \dots, M_k$ such that $\pld(M_i(\cdot; y^i)) = L_i$ regardless of the history $y^i$, we have:
\begin{align*}
    \pld(M_1 \otimes \dots \otimes M_k) = L_1 \oplus \dots \oplus L_k.
\end{align*}
This partially adaptive composition can be computed efficiently by numerical techniques including the fast Fourier transform \cite{KoskelaJH20}.

Our overarching research direction in this paper is to understand and mitigate any additional privacy costs under full adaptivity that do not occur under partial adaptivity. We have seen that full adaptivity comes for free under lossy pure DP and GDP accounting, so it is natural to investigate whether full adaptivity is also free under exact PLD accounting. While PLDs are not scalars, they can still be assigned a canonical information-theoretic partial order known as the \textit{Blackwell order}. With the Blackwell order $\preceq$ in hand, the natural extension of scalar composition filters to the exact PLD setting, which we call the \textit{natural filter}, is described in \Cref{alg:natural_filter}. If an interaction with the natural filter under a PLD privacy constraint $B$ has a privacy cost bounded by $B$, we will say that the natural filter is free. If this is true for any choice of $B$, we can reasonably conclude that full adaptivity has no additional privacy costs compared to partial adaptivity.

\begin{algorithm}[tb]
    \caption{NaturalFilter}
    \label{alg:natural_filter}
    \begin{algorithmic}
        \REQUIRE Budget PLD $B$, analyst $\cA$, legal query set $\cL$, query capacity $k \in \bZ^+ \cup \{\infty\}$, dataset $D$
        \FOR{$t \gets 1, \dots, k$}
            \STATE $\cA(Y_1, \dots, Y_{t - 1})$ gives mechanism $M_t$ that has PLD $L_t \in \cL$
            \IF{$L_1 \oplus \dots \oplus L_t \preceq B$}
                \YIELD $Y_t \sim M_t(D)$
            \ELSE
                \BREAK
            \ENDIF
        \ENDFOR
    \end{algorithmic}
\end{algorithm}

\subsection{Full Adaptivity is Not Free in General}

A detailed analysis of the privacy characteristics of the natural filter is presented in~\Cref{sec:natural_filter} using new theoretical tools described in~\Cref{sec:adversary_theory}. We enumerate the consequences of applying the natural filter in a range of concrete settings in~\Cref{sec:implications-special-case}. Unfortunately, the key result is that the natural filter only provides the intended privacy guarantee when the set of permissible queries is closed under composition and totally-ordered. This is a highly restrictive condition and rules out nearly all realistic choices of queries.

We write $\cL^{< \infty} := \{L_1 \oplus \dots \oplus L_k : L_1, \dots, L_k \in \cL\}$ for the closure of $\cL$ under convolution. For technical reasons, we also call a PLD $L$ degenerate if it has only a single finite point in its support.
It is common, for example, to consider only symmetric queries and all symmetric queries are non-degenerate. See~\Cref{app:symmetry} for more discussion of symmetry.
An asymmetric example with a degenerate PLD is the mechanism $M(D_1) = \textrm{Uniform}(\{0\}), M(D_2) = \textrm{Uniform}(\{0, 1\})$. Such mechanisms are not used in practice, so the following assumption is not a meaningful restriction. 

\begin{theorem}[Informal, see~\Cref{thm:total_ordering}]
    Let $\cL$ be a family of allowable query PLDs with no degenerate members other than the identity $\mathbf{0}$. Then, the following are equivalent.
    \begin{enumerate}[(i)]
        \item For all privacy budgets $B$, any interaction with the natural filter permitting queries in $\cL$ under privacy constraint $B$ has a privacy cost dominated by $B$; and
        \item $\cL^{< \infty}$ is totally-ordered.
    \end{enumerate}
\end{theorem}

An immediate consequence is that full adaptivity is not generally free under exact privacy accounting. For instance, if $\cL$ is taken to be the set of PLDs of pure $\eps$-DP queries, then, although $\cL$ itself is totally-ordered with respect to $\eps$, its closure $\cL^{< \infty}$ is not (see~\Cref{cor:no_free_filter_for_pure_dp}). Therefore, there exists a privacy budget $B$ and a legal interaction with the natural filter involving only pure DP queries so that the overall privacy cost of the interaction is not bounded by $B$.

While the requirement of total-ordering may seem surprising, we can explain it by taking a dynamic-programming view of full adaptivity. In~\Cref{sec:adversary_theory}, we develop a new set of theoretical tools for analyzing fully adaptive interaction under privacy constraints. There, we model a fully adaptive interaction as a game between an adversarial privacy analyst and a data curator. At the beginning of each round, the curator has already permitted queries with exact (PLD) privacy costs $L_1, \dots, L_{t - 1}$ and only accepts queries from the analyst with a privacy cost satisfying a rule $L_t \in \Gamma(L_1, \dots, L_{t - 1})$. Notice that any interaction with a privacy filter can be modelled by such a rule $\Gamma$.

Our analysis also relies on the fact that, like real-valued scalars, PLDs are also analytically complete (see~\Cref{prop:sup-conv}). That is, for any non-empty set $\cL$ of PLDs, there exists a unique PLD $\sup{\cL}$ that dominates $\cL$ but is dominated by any upper bound for $\cL$. Analogously, one can take infima of PLDs. To the best of our knowledge, this is a new observation.

With the supremum in hand, we show that the privacy cost of any fully adaptive interaction with a privacy rule satisfies a Bellman equation. This enables an exact PLD-based privacy analysis of any privacy filter as well as fully adaptive interaction in general. As a result, this tool may be of independent theoretical and practical interest.

\begin{lemma}[Informal, see~\Cref{lemma:restricted_adversary}]
    Let $\Gamma$ be any privacy rule and let $\Gamma|_{L_1}$ denote the updated rule after the curator accepts a legal move $L_1$ from the analyst. Furthermore, let $\pld(\Gamma)$ denote the worst-case (supremal) privacy cost of any legal interaction with a privacy rule $\Gamma$. Then
    \begin{align*}
        \pld(\Gamma) = \sup\{L_1 \oplus \pld(\Gamma|_{L_1}) : L_1 \text{ is legal under } \Gamma\}.
    \end{align*}
\end{lemma}

As a simple example, consider a two-round interaction with the natural filter. Suppose the adversary has only two queries available with exact costs $L_1$ and $L_2$. In round one, the adversary issues the query costing $L_1$ and, in round two, it adversarially chooses one of $\{L_1, L_2\}$. Assume that each possible path of the adversary $L_1 \oplus L_1, L_1 \oplus L_2 \preceq B$ satisfies the required privacy cost separately so that the natural filter will accept any choice in round two. On the other hand, by adversarially choosing the next query from $\{L_1, L_2\}$ based on the observed outcome of the first round, the Bellman equation ensures that the worst-case cost of the interaction can be increased to $L_1 \oplus \sup\{L_1, L_2\}$, which may no longer be bounded by $B$. That is, the natural filter can only enforce branchwise validity $\sup\{L_1 \oplus L_1, L_1 \oplus L_2\} \preceq B$ when we really need $L_1 \oplus \sup\{L_1, L_2\} \preceq B$. These are equivalent for filters based on real-valued scalars but not necessarily for exact PLD costs.

More generally, using the Bellman equation we show that the natural filter is universally free for any privacy constraint $B$ for precisely those query families in which the supremum distributes over convolution, i.e.
\begin{align*}
    \forall L \in \cL, \cL' \subseteq \cL^{< \infty} \quad \sup\{L \oplus \cL'\} = L \oplus \sup\{\cL'\}.
\end{align*}
It is relatively straightforward to prove that distributivity holds when $\cL^{< \infty}$ is totally-ordered. Proving the reverse is substantially more technically challenging and we devote all of~\Cref{sec:total_ordering_proof} to its proof. Our main objective there is to show that, for a pair of non-degenerate queries $L_1, L_2$ that are not totally-ordered, we can exploit the topology of the queries in order to amplify the gaps between them under composition and conclude that, without loss of generality,
\begin{align*}
    L_1 \oplus \sup\{L_1, L_2\} \succ \sup\{L_1 \oplus L_1, L_1 \oplus L_2\}.
\end{align*}
Our techniques in~\Cref{sec:total_ordering_proof} show in general that privacy gaps can be amplified by composition, which may be of independent interest. To the best of our knowledge, this is also the first application of topological methods to privacy analysis.

While it is disappointing that the natural filter can fail even for very well-behaved families of queries, we have only shown the existence of one privacy budget $B$ triggering this failure. It is reasonable to ask whether full adaptivity with exact privacy accounting is still free for well-structured privacy constraints. To understand this situation, we may apply the Bellman equation to obtain a more general characterization of the privacy cost of interacting with the natural filter under specific combinations of query and privacy budget constraints. Key results are summarized in~\Cref{tab:natural-filter-examples}. Overall, imposing generous structural assumptions on both queries and privacy budgets is not enough to enable free full adaptivity for practical situations. We refer to \Cref{sec:implications-special-case} for a more detailed list of consequences of our characterization of the natural filter's privacy cost, including plots demonstrating examples of privacy failure.

\begin{table}[H]
    \centering
    \renewcommand{\arraystretch}{1.15}
    \begin{tabular}{llll}
        \toprule
        \textbf{Queries} & \textbf{Budget} & \textbf{Free?} & \textbf{Explanation} \\
        \midrule
        $(\varepsilon,\delta)$-DP
            & $(\varepsilon,\delta)$-DP
            & No
            & Budget selected from query class insufficient \\
        Arbitrary
            & $\mu$-GDP
            & No
            & Well-behaved budget insufficient \\
        $\mu$-GDP
            & Arbitrary
            & Yes
            & Closed and totally-ordered \\
        $(0,\delta)$-DP
            & Arbitrary
            & Yes
            & Closed and totally-ordered \\
        \bottomrule
    \end{tabular}
    \caption{Examples illustrating when natural filters are free.}
    \label{tab:natural-filter-examples}
\end{table}

\subsection{How Badly Can the Natural Filter Fail?}

Although we have seen that full adaptivity does not generally come for free, it is so far unclear whether it is possible to quantify the additional privacy cost of full adaptivity. In~\Cref{sec:natural-apx}, we partially answer this question by providing an upper bound on the privacy guarantee of an interaction with the natural filter under an $(\eps, \delta)$-DP privacy constraint.

For technical reasons, we restrict our analysis to so-called \textit{symmetric} queries. We call a query symmetric if its PLD $L = \pld(P \parallel Q)$ is such that $\pld(P \parallel Q) = \pld(Q \parallel P)$, i.e. $P$ is as difficult to distinguish from $Q$ as $Q$ is from $P$. Many typical mechanisms have symmetric privacy costs and more details on symmetry can be found in~\Cref{app:symmetry}.

\begin{theorem}[Informal, see~\Cref{thm:main-ub-apx}]
    Let $\eps > 0, \delta \in (0, 1/4)$. Then, any interaction with the natural filter allowing up to $k$ symmetric queries under an $(\eps, \delta)$-DP privacy budget satisfies $\Paren{\chi \cdot (\eps + 1), \chi \cdot \delta}$-DP, where
    \begin{align*}
        \chi = O_{\eps, \delta}(\log{k}).
    \end{align*}
\end{theorem}

Our upper bound analysis of the natural filter for symmetric queries reduces to a much simpler discrete privacy accounting problem. First, we discretize the privacy loss into an appropriate range of scales $\eps_1 \leq \eps_2 \leq \dots \leq \eps_n$ ($n = O_{\eps, \delta}(\log{k})$). Using symmetry, we show that a symmetric query can be well-approximated by a mixture of randomized response (RR) queries at each scale $\eps_i$. RR is the canonical discrete mechanism that outputs a single randomized bit by flipping a biased coin. After discretization, we can replace the natural filter by a more permissive additive filter that maintains a separate scalar budget for each scale $\eps_i$. Any interaction accepted by the natural filter is also accepted by the additive filter applied to discrete approximations of the original queries, so it remains to bound the privacy cost of the simpler additive filter. To this end, we show that, with high probability, an interaction with the additive filter can be viewed as a post-processing of a non-adaptive sequence of randomized response queries tuned according to the budget allocated at each privacy scale $\eps_i$. Therefore, up to a small failure probability, our problem reduces to controlling the privacy loss of the composition of the non-adaptive sequence of discrete queries.

Essentially, while the natural filter can fail in general, it cannot fail too badly. The additional privacy cost of full adaptivity with exact privacy accounting is at most a $O_{\eps, \delta}(\log{k})$ factor increase in the privacy cost, where $k$ is the number of rounds of interaction.
It is not clear if the $\log{k}$ factor increase in privacy cost is tight or not and it is an interesting direction for future work to determine whether this gap can be reduced.

\subsection{Privacy Filters are Captured by Residues}

We have shown that full adaptivity does not come for free, even under generous structural assumptions, though the additional cost is at most logarithmic in the number of rounds of interaction.

Our next goal is to understand whether it is possible at all to construct a valid privacy filter based on exact PLD accounting. Indeed, in~\Cref{sec:residue_filters} we develop a general purpose algorithm design framework that will allow us to construct nearly free privacy filters for important applications.

Our starting point is to notice that existing scalar composition filters can be reframed as maintaining a \textit{remaining} privacy budget rather than maintaining an estimate of the privacy budget consumed thus far. This can immediately be extended to the setting of exact privacy accounting with PLDs. To this end, fix a class of budget $\cB$ and a class of legal queries $\cL$ and consider the updating budget filter described in \Cref{alg:update_filter}.

\begin{algorithm}[tb]
    \caption{UpdateFilter}
    \label{alg:update_filter}
    \begin{algorithmic}
        \REQUIRE Analyst $\cA$, update rule $U : \cB \times \cL \to \cB$, budget $B \in \cB$, dataset $D$
        \FOR{$i \gets 1, 2, \dots$}
            \STATE $\cA(Y_1, \dots, Y_{i - 1})$ gives mechanism $M_i$ that has PLD $L \in \cL$
            \IF{$L \preceq B$}
                \YIELD $Y_i \sim M_i(D)$
                \STATE $B \gets U(B, L)$
            \ELSE
                \BREAK
            \ENDIF
        \ENDFOR
    \end{algorithmic}
\end{algorithm}

It turns out that this kind of filter produces the expected privacy guarantee as long as the budget update $U$ satisfies an intuitive \textit{residue} condition: For any query costing $L$, the sum of the updated budget $U(B, L)$ and the query cost $L$ must be dominated by the budget $B$. In this case, we call the updated budget $U(B, L)$ a residue of $B$ with respect to $L$.

\begin{theorem}[Informal, see~\Cref{thm:residue_filter}]
    Fix an update rule $U : \mathcal{B} \times \mathcal{L} \to \mathcal{B}$ and suppose that, for every budget $B \in \mathcal{B}$ and every $L \in \mathcal{L}$ with $L \preceq B$, we have
    \begin{equation*}
        \label{eq:residue_condition_overview}
        U(B, L) \oplus L \preceq B. \tag{$\star$}
    \end{equation*}
    Then, for any initial budget $B \in \mathcal{B}$, any interaction with the update filter \Cref{alg:update_filter} has privacy cost at most $B$. In this case, we call the filter a residue filter.
\end{theorem}

It turns out that residue filters already capture all existing privacy filters based on scalar composition including the pure DP filter, the GDP filter, and the R\'enyi filter. We describe these reductions in more detail in~\Cref{sec:residue_filters}. In fact, we show that the natural filter itself is a residue filter exactly when it is free. Essentially, we can interpret the natural filter, when it is free, as the residue filter that considers all legal residual budget updates and simply takes the supremum over all of them. Unfortunately, as we have seen, taking the \textit{best} residual budget by way of supremum may paradoxically lead to a budget that no longer satisfies the residue condition.

\begin{theorem}[Informal, see~\Cref{thm:universal_free_natural_filter}]
    Fix a query class $\cL$ and consider the update rule
    \begin{align*}
        U(B, L) := \sup\{L' \in \cL^{< \infty} : L \oplus L' \preceq B\}.
    \end{align*}
    Then $U$ satisfies the residue condition if and only if the natural filter is free for every budget $B$, in which case the natural filter and the residue filter with update $U$ are identical.
\end{theorem}

\subsection{Residues Can Recover Asymptotically Free Full Adaptivity}

We now have an algorithm-design tool that unifies all existing privacy filters. While we have seen that full adaptivity cannot come completely for free under exact accounting, our final goal is to understand whether the residue filter framework can be used to recover asymptotically free full adaptivity in practical special cases. More formally, we consider the regime in which the size of the queries tends to zero and the overall budget size tends to infinity. In this regime, we would like a filter that permits asymptotically as many queries as the natural filter, which performs exact privacy accounting. To this end, we introduce a new filter that improves upon the na\"ive GDP filter.

\begin{theorem}[Informal, see~\Cref{thm:gdp_res_filter}]
    Let $G_\mu$ denote the privacy cost of a $\mu$-GDP query and consider the update rule
    \begin{align*}
        U(G_\mu, L) := G_{\mu'}, \quad \mu' \geq 0 \text{ maximal such that } G_{\mu'} \oplus L \preceq G_\mu
    \end{align*}
    for GDP budgets $G_\mu$ and arbitrary queries $L \preceq G_\mu$. This rule satisfies the residue condition \eqref{eq:residue_condition_overview} and hence the update filter derived from it is a residue filter. We call this filter the GDP residue filter.
\end{theorem}

Where the original (na\"ive) GDP filter \cite{ST22,KoskelaTH23} computes the GDP parameter of the query before subtracting it from the budget, the GDP residue filter directly selects the largest Gaussian residual available at each budget update. It is clear that this filter is at least as permissive as the na\"ive GDP filter, what is perhaps less clear is that the GDP residue filter can behave differently. Indeed, in the case of pure DP queries, the square GDP cost of na\"ive compared to residue budget updates differs by a factor of $\frac{\pi}{2}$, so long as we are in the regime where the query cost is not too large relative to the overall budget. In the following, we denote by $o(1)$ an error term tending to zero in this regime.

\begin{theorem}[Informal, see~\Cref{thm:gdp_residue_pure_query}]
    For a budget $\mu > 0$ and $\eps > 0$, let $\mu_{\textrm{na\"ive}}$ denote the na\"ive GDP filter update applied to budget $G_\mu$ and given an $\eps$-DP query. Likewise, let $\mu_{\textrm{res}}$ denote GDP residue. Then
    \begin{enumerate}
        \item $\mu^2 - \mu_{\textrm{na\"ive}}^2 = 4\Phi^{-1}(1/(1 + e^{-\eps}))^2 = \frac{\pi}{2}\eps^2 + O(\eps^4)$; and
        \item $\mu^2 - \mu_{\textrm{res}}^2 = 8\ln(\cosh(\eps/2)) + o(1) = \eps^2 + O(\eps^4) + o(1)$.
    \end{enumerate}
\end{theorem}

As a consequence, the GDP residue filter permits asymptotically and on average a factor $\frac{\pi}{2}$ (i.e. about 50\%) more pure DP queries than the na\"ive GDP filter. Perhaps more notably, we show that this is asymptotically optimal in the sense that no more pure DP queries can be permitted, even without any adaptivity at all (see~\Cref{thm:gdp_residue_many_pure_queries}). We present empirical evidence in~\Cref{sec:residue_filters} closely matching this theoretical prediction in a realistic parameter regime.

In particular, we have shown that, although full adaptivity comes with a cost in general, this cost can be driven to zero asymptotically by a careful choice of residual budget updates. We expect it is possible to construct many more kinds of residue filters that enable asymptotically free full adaptivity for a broad range of other applications, including private machine learning.

%% file: sections/prelims.tex
\section{Preliminaries}
\label{sec:preliminaries}

In our work we reason about the composition of adaptive mechanisms, which are permitted to inspect outcomes of any previously run mechanism. It is enough to define adaptivity for just two rounds of interaction as we can define more complex adaptively composed mechanisms inductively.

\begin{definition}
    Let $M_1 : \cD \to \cY_1$ be a mechanism and let $M_2 : \cD \times \cY_1 \to \cY_2$ be an adaptive mechanism. The adaptive composition of $M_1$ and $M_2$ is the mechanism
    \begin{align*}
        (M_1 \otimes M_2)(D) := (Y_1, Y_2), Y_1 \sim M_1(D), Y_2 \sim M_2(D; Y_1).
    \end{align*}
\end{definition}

It is natural in general to view the privacy characteristics of any private mechanism as a hypothesis testing problem \cite{WassermanZ10,DongRS19}. As such it is desirable to furnish hypothesis testing problems $(P, Q)$ with some kind of ordering so that we can compare the privacy guarantees offered by various mechanisms. We also require an ordering so that we can make sense of natural privacy filters. Fortunately testing pairs can be equipped with a very natural order, namely the Blackwell order. While the Blackwell order is related to data-processing inequalities, its independent use for comparing privacy guarantees originates with \cite{DongRS19}.

\begin{definition}
    \label{def:dominating}
    For a pair of distributions $(P, Q)$ defined on a common probability space $\Omega$ and a second pair $(P', Q')$ defined on $\Omega'$, we say that $(P', Q')$ dominates $(P, Q)$ in the Blackwell order, written $(P, Q) \preceq (P', Q')$, if we can find a Markov kernel $\kappa$ from $\Omega'$ to $\Omega$ such that $P = \kappa P'$ and $Q = \kappa Q'$. We write $(P, Q) \prec (P', Q')$ if $(P, Q) \preceq (P', Q')$ but not vice-versa.
    If $(P, Q) \preceq (Q, P) \preceq (P, Q)$ the pair is called symmetric.
\end{definition}

Informally, the Blackwell order says that more information is available for distinguishing $P'$ and $Q'$ compared to $P$ vs $Q$. For mechanisms, we can interpret this as $(M(D_1), M(D_2)) \preceq (M'(D_1), M'(D_2))$ when $M$ can be reconstructed from $M'$ by introducing additional randomness through the kernel $\kappa$. Note that the Blackwell order is a partial order.

Instead of carrying around a pair of distributions to represent the privacy characteristics of a mechanism, we can consolidate them into a privacy loss distribution (PLD) \cite{DworkR16}. As we will see, the PLD possesses very desirable composition properties and is a practical basis for privacy loss computations.
To formalize the PLD, we require likelihood ratios. Recall that, for a distribution $P$ absolutely continuous with respect to a distribution $Q$ ($P \ll Q$) on a probability domain $\Omega$, their likelihood ratio is the Radon-Nikodym derivative $\frac{dP}{dQ}$. However, in some cases, we would like to consider private mechanisms that violate absolute continuity and still construct the PLD. In this case, we consider the Lebesgue decomposition of $P = P_\ll + P_\bot$ into absolutely continuous and singular components. For points in the support of the singular measure $P_\bot$, we assign likelihood $\infty$ and, for points outside, we may safely assign the likelihood $dP_\ll/dQ$. For simplicity, we denote this extended likelihood ratio as $dP/dQ$.

\begin{definition}
    Let $P$ and $Q$ be distributions on $\Omega$. The PLD of $(P, Q)$, denoted $\pld(P \parallel Q)$, is the distribution over $\overline{\bR} := \bR \cup \{\pm\infty\}$ of $\log(\frac{dP}{dQ}(\omega)), \omega \sim P$.
\end{definition}

For shorthand, we denote by $\pld(M) := \pld(M(D_1) \parallel M(D_2))$ the PLD of a mechanism $M$. Note that a distribution can be easily recognized as a PLD as follows.

\begin{proposition}
    A distribution $L$ on $\overline{\bR}$ is the PLD of some $(P, Q$) if and only if it satisfies $\E_{Z \sim L}[e^{-Z}] \leq 1$. In this case, one such pair is $(L, L')$ where $L'$ is the Esscher tilt of $L$, namely $dL'(z) := e^{-z} dL(z)$ for $z \in \bR$ and $L'(\{-\infty\}) := 1 - \E_{Z \sim L}[e^{-Z}]$.
\end{proposition}

PLDs also inherit the Blackwell order: For PLDs $L$ and $L'$, we say that $L \preceq L'$ if there exist pairs $(P, Q)$ and $(P', Q')$ such that $L := \pld(P \parallel Q)$, $L' = \pld(P' \parallel Q')$, and $(P, Q) \preceq (P', Q')$. It is straightforward to show that this order is well-defined regardless of the underlying representations $(P, Q)$ and $(P', Q')$. When equipped with the Blackwell order, PLDs also form a partially-ordered commutative monoid under convolution.

\begin{proposition}
    \label{prop:pld_convolution_properties}
    For any PLDs $L_1, L_2, L_3, L'_1, L'_2$,
    \begin{enumerate}
        \item The convolution $L_1 \oplus L_2$ is also a PLD;
        \item The identically zero distribution $\Iddist$ and the identically $\infty$ distribution, also denoted $\infty$, are PLDs such that $\Iddist \preceq L_1 \preceq \infty$, $L_1 \oplus \Iddist = L_1$, and $L_1 \oplus \infty = \infty$;
        \item $(L_1 \oplus L_2) \oplus L_3 = L_1 \oplus (L_2 \oplus L_3)$;
    \item $L_1 \oplus L_2 = L_2 \oplus L_1$; and
        \item If $L_2 \preceq L'_2$, then $L_1 \oplus L_2 \preceq L_1 \oplus L'_2$.
    \end{enumerate}
\end{proposition}

A highly prized property of PLDs is that the privacy characteristics of an adaptive mechanism correspond to convolution of PLDs, provided that each component fixes a privacy budget in advance. This is important because the convolution of a long sequence of distributions can be computed very efficiently by applying the fast Fourier transform \cite{KoskelaJH20}. Note that we will require more than this to understand fully adaptive composition.

\begin{proposition}
    \label{prop:composition_convolution}
    Let $M_1 : \cD \to \cY_1$, $M'_1 : \cD \to \cY'_1$, $M_2 : \cD \times \cY_1 \to \cY_2$, and $M'_2 : \cD \times \cY'_1 \to \cY'_2$ be adaptive mechanisms such that $\pld(M_1) \preceq \pld(M'_1)$ and $\pld(M_2(\cdot; y_1)) \preceq \pld(M'_2(\cdot; y_1))$ for every $y_1 \in \cY_1$. Then $\pld(M_1 \otimes M_2) \preceq \pld(M'_1 \otimes M'_2)$. Moreover, if $\pld(M_1) = L_1$ and $\pld(M_2(\cdot; y_1)) = L_2$ for every $y_1 \in \cY_1$, then $\pld(M_1 \otimes M_2) = L_1 \oplus L_2$.
\end{proposition}

We also define a few important classes of PLDs that will feature later in our discussion.

\begin{definition}
    \label{def:pld_class}
    We denote by $\cLall$ the set of all PLDs and by $\cLsym := \{\pld(P \parallel Q) : (P, Q) \text{ symmetric}\}$ the set of all symmetric PLDs.
\end{definition}

\begin{definition}
    \label{def:approx_pure_dp}
    For $\eps \geq 0$, $\delta \in [0, 1]$, let
    \begin{align*}
        R_{\eps, \delta} := \pld(\delta 1_{\bot} + (1 - \delta)\textrm{Bern}(1/(1 + e^\eps)) \parallel \delta 1_{\top} + (1 - \delta)\textrm{Bern}(1/(1 + e^{-\eps})))
    \end{align*}
    We say that $L \in \cLall$ satisfies $(\eps, \delta)$-DP if $L \preceq R_{\eps, \delta}$ or $\eps$-DP if $L \preceq R_\eps := R_{\eps, 0}$. We denote $\cLpure := \{R_\eps : \eps \geq 0\}$ and $\cLapprox := \{R_{\eps, \delta} : \eps \geq 0, \delta \in [0, 1]\}$.
\end{definition}

\begin{definition}
    \label{def:gdp}
    For $\mu \geq 0$, we define $G_\mu := \pld(N(0, 1) \parallel N(\mu, 1))$ and say that $L \in \cLall$ satisfies $\mu$-GDP when $L \preceq G_\mu$. We denote $\cLgdp := \{G_\mu : \mu \geq 0\}$.
\end{definition}

\begin{definition}
    \label{def:rdp}
    Lastly, we say that $L \in \cLall$ satisfies $(\alpha, \rho$)-RDP if $\E_{Z \sim L}[e^{(\alpha - 1)Z}] \leq e^{(\alpha - 1)\rho}$.
\end{definition}

All three forms of privacy offer simple composition guarantees: For pure DP, $R_{\eps_1} \oplus R_{\eps_2} \preceq R_{\eps_1 + \eps_2}$ \cite{dwork2006calibrating}, for GDP $G_{\mu_1} \oplus G_{\mu_2} = G_{\sqrt{\mu_1^2 + \mu_2^2}}$ \cite{DongRS19}, and given $L_1$ satisfying $(\alpha, \rho_1)$-RDP and $L_2$ satisfying $(\alpha, \rho_2)$-RDP, $L_1 \oplus L_2$ satisfies $(\alpha, \rho_1 + \rho_2)$-RDP \cite{mironov2017renyi}. More discussion of symmetry can be found in \Cref{app:symmetry}.

Closely related to the PLD is the hockey-stick curve: a convex reparameterization of the privacy profile, which is in turn a generalization of $(\eps, \delta)$-DP.

\begin{definition}
    For a pair of distributions $(P, Q)$ over $\Omega$, their hockey-stick divergence at order $x \in \bR^\times := (0, \infty)$ is
    \begin{align*}
        H_x(P \parallel Q) := \sup_{E \subseteq \Omega} P(E) - xQ(E).
    \end{align*}
\end{definition}

For convenience, we write $H_x(M) := H_x(M(D_1) \parallel M(D_2))$ and sometimes $h_M(x) := H_x(M)$ for the hockey-stick curve of a mechanism $M$.
The hockey-stick curve is very closely related to the classical $(\eps, \delta)$ form of differential privacy: a mechanism $M$ satisfies $(\eps, \delta)$-DP exactly when $H_{e^\eps}(M) \leq \delta$. In general, hockey-stick curves are characterized by a few simple conditions.

\begin{proposition}[\cite{ZhuDW22} Lemma 9]
    \label{prop:hs_characterization}
    A curve $h : \bR^\times \to [0, 1]$ is the hockey-stick curve of some $(P, Q)$ iff $h$ is convex and decreasing with $\lim_{x \to 0}h(x) = 1$ and $h(x) \geq 1 - x$ for all $x \in \bR^\times$.
\end{proposition}

Hockey-stick curves are also closely related to PLDs and can be fully recovered from the PLD via the following formula. The proof is in \Cref{app:prelim_proofs}. Note that we denote $(t)_+ := \max\{t, 0\}$.

\begin{proposition}
    \label{prop:pld_to_hs}
    For any pair $(P, Q)$ with likelihood $\ell$ and $x \in \bR^\times$,
    \begin{align*}
        H_x(P \parallel Q)
            = \E_{Y \sim P}[(1 - x/\ell(Y))_+]
            = \E_{Z \sim \pld(P \parallel Q)}[(1 - xe^{-Z})_+].
    \end{align*}
\end{proposition}

For convenience, we will sometimes write $h_L$ to denote the unique hockey-stick curve associated to a PLD $L$, namely $h_L(x) := \E_{Z \sim L}[(1 - xe^{-Z})_+]$. Moreover, like PLDs, hockey-stick curves also capture adaptive composition in a natural way. The proof is also in \Cref{app:prelim_proofs}.

\begin{proposition}
    \label{prop:hockey_stick_composition}
    Let $M_1 : \cD \to \cY_1$ and $M_2 : \cD \times \cY_1 \to \cY_2$ be adaptive mechanisms and let $\ell^1 := \frac{dM_1(D_1)}{dM_1(D_2)}$ denote the likelihood ratio for $M_1$. Then
    \begin{align*}
        h_{M_1 \otimes M_2}(x) = \E_{Y_1 \sim M_1(D_1)}[h_{M_2(\cdot; Y_1)}(x/\ell^1(Y_1))].
    \end{align*}
\end{proposition}

Our last characterization of privacy is given the Type I/Type II error tradeoff curve, also known as the $f$-DP framework of privacy \cite{DongRS19}.

\begin{definition}
    Let $(P, Q)$ be distributions on $\Omega$. For a randomized hypothesis test $\phi : \Omega \to \Delta(\{P, Q\})$, we call $\alpha_\phi := \bP_{\omega \sim P, \mu \sim \phi(\omega)}[\mu = Q]$ its Type I error and $\beta_\phi := \bP_{\omega \sim Q, \mu \sim \phi(\omega)}[\mu = P]$ its Type II error. The tradeoff curve for $(P, Q)$ is defined as
    \begin{align*}
        T_\alpha(P \parallel Q) := \inf_{\phi} \{\beta_\phi : \alpha_\phi \leq \alpha\}.
    \end{align*}
\end{definition}

In general, we can characterize valid tradeoff curves as follows.

\begin{proposition}[\cite{DongRS19} Proposition 2.2]
    \label{prop:tradeoff_characterization}
    A curve $\tau : [0, 1] \to [0, 1]$ is the tradeoff curve for some $(P, Q)$ iff $\tau$ is continuous, convex, decreasing, and $\tau(\alpha) \leq 1 - \alpha$ for all $\alpha \in [0, 1]$.
\end{proposition}

As with hockey-stick curves, the PLD determines the tradeoff curve (see \Cref{prop:pld_to_tradeoff}). For convenience, we denote the tradeoff curve corresponding to a PLD $L$ as $\tau_L$. It is common to represent the privacy characteristics of a mechanism by either its hockey-stick curve or its tradeoff curve. We show that there is a natural link between these representations via inversion and convex conjugacy. For a tradeoff function $\tau$ of a testing problem $(P, Q)$, we write $\tau^{-1}(\beta) := \inf\{\alpha \in [0, 1] : \tau(\alpha) \leq \beta\}$ for its inverse, namely the tradeoff curve for the testing problem $(Q, P)$. For a convex function $g$, we denote its convex conjugate by $g^*(y) := \sup_{x \in \mathbb{R}} xy - g(x)$. Note that we can extend tradeoff functions to the real-line by setting them to $\infty$ outside their support $[0, 1]$. Note that a special case of the following formula appears in \cite{DongRS19} (Prop 2.12) for symmetric tradeoff functions but this generalization is new to the best of our knowledge. It also appears independently as Lemma 10 in \cite{glasserman2026privacy}. We give a proof in \Cref{app:prelim_proofs}.

\begin{proposition}
    \label{prop:tradeoff_to_hs}
    Let $(P, Q)$ be a pair of distributions with tradeoff function $\tau(\alpha) := T_\alpha(P \parallel Q)$ and hockey-stick curve $h(x) := H_x(P \parallel Q)$. Then
    \begin{align*}
        h(x) = 1 + (\tau^{-1})^*(-x)
    \end{align*}
\end{proposition}

It follows as a natural consequence that the Blackwell order on pairs, the ordering induced by
the tradeoff function, as well as the hockey-stick induced ordering are all equivalent.

\begin{proposition}
    \label{prop:f-hs-equivalence}
    Let $(P, Q)$ and $(P', Q')$ be pairs of distributions with tradeoff functions $\tau$ and $\tau'$ respectively as well as hockey-stick curves $h$ and $h'$ respectively. Then
    \begin{align*}
        (P, Q) \preceq (P', Q') \iff \tau \succeq \tau' \iff h \preceq h'.
    \end{align*}
\end{proposition}

The equivalence of the first and second ordering is given by a celebrated theorem of Blackwell \cite[Theorem 10]{Bla51}. As for the equivalence of the second and third ordering, this now follows by applying Fenchel–Moreau duality to \Cref{prop:tradeoff_to_hs} and recalling that convex conjugation is an order-reversing operation.

As a consequence, it turns out that PLDs endowed with the Blackwell order possess the same remarkable completeness property that enables analysis in the real numbers.

\begin{proposition}\label{prop:sup-conv}
    For any non-empty family of PLDs $\mathcal{L}$, there exists a unique PLD $\sup \mathcal{L}$ that dominates $\mathcal{L}$ but is dominated by every upper bound for $\mathcal{L}$. In this case,
    \begin{enumerate}[(i)]
        \item $h_{\sup \mathcal{L}}(x) = \sup\{h_L(x) : L \in \mathcal{L}\}$
        \item $\tau_{\sup \mathcal{L}} = \conv\{\tau_L : L \in \mathcal{L}\}$
    \end{enumerate}
    where $(\conv{\cF})(\alpha) := \sup\{f(\alpha) : f \text{ is convex}, f \preceq \cF\}$ denotes the lower convex envelope.
\end{proposition}

This result follows immediately from \Cref{prop:f-hs-equivalence} and from noticing that the properties characterizing hockey-stick curves (see \Cref{prop:hs_characterization}) are all closed under pointwise suprema. Part (ii) follows from noticing that all of the properties of tradeoff curves (see \Cref{prop:tradeoff_characterization}) are also closed under the lower convex envelope operation. Notice that it may be the case that $\sup{\cL} = \infty$. Also notice that we may construct infima of privacy losses in an analogous manner, namely by taking lower convex envelopes of hockey-stick curves or by taking pointwise suprema of tradeoff curves. We will not require infima for our work but we do rely on the supremum extensively in order to reason about natural privacy filters.
As far as we are aware, the completeness property of PLDs has not actually been documented in the literature.

%% file: sections/constrained_adversaries.tex
\section{Query-Restricted Adversaries}
\label{sec:adversary_theory}

In this section we develop a general theory of the privacy characteristics of an adversarial adaptive mechanism when the queries the adversary may issue are restricted. As we will see, this will lead directly to a theory of residue and natural filters, in addition to new filters. For simplicity, we will assume throughout that an adversary may always ``pass'' its turn by issuing a query with no privacy loss. We will consider adversaries that issue finitely many queries as well as adversaries that issue a countable sequence of queries. We will abuse notation somewhat by writing such things as $M_1 \otimes \dots \otimes M_\infty$ to denote the countable extension of the prefixes $M_1 \otimes \dots \otimes M_t$ as $t \to \infty$.

\begin{definition}
    Let $k \in \bZ^+ \cup \{\infty\}$. We define a privacy rule of length $k$ as a map $\Gamma$ that takes a sequence of PLDs $(L_1, \dots, L_t)$ ($0 \leq t < k$) and outputs some collection of PLDs including $\Iddist$. We define a $\Gamma$-adversary as an adaptively composed mechanism $M = M_1 \otimes \dots \otimes M_k$ such that
    \begin{align*}
        L_t := \pld(M_t(\cdot; y_1, \dots, y_{t - 1})) \in \Gamma(L_1, \dots, L_{t - 1})
    \end{align*}
    for every $(y_1, \dots, y_{k - 1}) \in \cY_1 \times \dots \times \cY_{k - 1}$.
    We will write $\Adv(\Gamma)$ to denote the set of $\Gamma$-adversaries and we will denote by
    \begin{align*}
        \pld(\Gamma) := \sup\{\pld(M) : M \in \Adv(\Gamma)\}
    \end{align*}
    the worst-case privacy loss of arbitrary $\Gamma$-adversaries.
\end{definition}

Crucially, notice that no single adversary fully realizes the privacy cost of the rule $\Gamma$. Recalling the $f$-DP interpretation of the supremum (see \Cref{prop:sup-conv}), the privacy cost of the rule $\Gamma$ is in fact realized by a family of adversaries that minimize the false negative rate for every given confidence level. As the following core lemma reveals, it turns out that the maximum privacy cost of the rule $\Gamma$ for a given confidence level occurs when the adversary first plays the best move for this confidence level, observes the outcome, then depending on the likelihood of having observed this outcome recursively chooses its remaining strategy tuned to maximize the overall privacy cost. Note that by $\lambda$ we mean the empty sequence. The following result is proved in \Cref{app:adversary_theory}.

\begin{lemma}
    \label{lemma:restricted_adversary}
    Let $\Gamma$ be a rule of length $k \in \bZ^+ \cup \{\infty\}$. If $k = 1$, then $\pld(\Gamma) = \sup{\Gamma(\lambda)}$ and, for $k > 1$, we have
    \begin{align*}
        \pld(\Gamma) = \sup\{L \oplus \pld(\Gamma|_L) : L \in \Gamma(\lambda)\}
    \end{align*}
    where $\Gamma|_{L_0}(L_1, \dots, L_t) := \Gamma(L_0, L_1, \dots L_t)$ denotes the rule of length $k - 1$ that fixes the first move played against $\Gamma$ to $L_0$.
\end{lemma}

A final tool we will require for analyzing adversaries respecting a privacy rule is that privacy rules of infinite length can be understood by passing to limits of finite privacy rules. The proof is given in \Cref{app:adversary_theory}.

\begin{lemma}
    \label{lemma:infinite_rule_limit}
    Let $\Gamma$ be a privacy rule of infinite length. For $k \in \bZ^+$, let $\Gamma^k$ be the privacy rule of length $k$ denoting the restriction of $\Gamma$ to query sequences of length less than $k$. Then
    \begin{align*}
        \pld(\Gamma) = \sup\{\pld(\Gamma^k) : k \in \bZ^+\}.
    \end{align*}
\end{lemma}

%% file: sections/residue_filter.tex
\section{Residue Filters}
\label{sec:residue_filters}

In this section, we develop a new class of privacy filters capable of leveraging exact privacy accounting: the residue filters. We show that residue filters capture all of the filters available in the literature including the pure DP filter, the R\'enyi filter, the GDP filter, as well as the natural filter. In general, we say that a filter with a given privacy budget is free if the privacy cost of any interaction with the filter is bounded by the budget. We call a filter universally free if it is free for any budget.


A key piece of structure that appears to enable existing privacy filters is the notion of a remaining privacy budget. Our first main contribution of this section is to lift this filter structure into the general exact accounting framework of PLD budgets and queries under Blackwell ordering. For this, we require some simple ingredients: Let $\mathcal{B}$ be a set of PLDs representing legal budgets, let $\mathcal{L}$ be a family of allowed queries, and let $U : \mathcal{B} \times \mathcal{L} \to \mathcal{B}$ be a map taking a legal budget and query and producing an updated budget. For simplicity, we assume $\Iddist \in \mathcal{B}, \mathcal{L}$. \Cref{alg:update_filter} formalizes the interaction between analyst and the update filter with initial budget $B$. Now, notice that \Cref{alg:update_filter} can be captured succinctly by the recursive privacy rule
\begin{align*}
    \Update_{U, B}(\lambda) &:= \{L \in \mathcal{L} : L \preceq B\} \\
    \Update_{U, B}|_{L_1} &:= \textrm{Update}_{U, U(B, L_1)}.
\end{align*}

Our key result is that privacy filters derived from budget updates are universally free when the updates satisfy a natural residue condition. Informally, it is enough that every budget encountered during the interaction dominates the privacy cost of the updated budget plus the cost of the query. In this case, we call the updated budget a residue of the budget with respect to the query.

\begin{theorem}
    \label{thm:residue_filter}
    Fix an update rule $U : \mathcal{B} \times \mathcal{L} \to \mathcal{B}$ and suppose that, for every budget $B \in \mathcal{B}$ and every $L \in \mathcal{L}$ with $L \preceq B$, we have
    \begin{equation*}
        \label{eq:residue_condition}
        U(B, L) \oplus L \preceq B. \tag{$\star$}
    \end{equation*}
    Then, for any initial budget $B \in \mathcal{B}$, the update filter \Cref{alg:update_filter} is free, i.e. $\pld(\Update_{U, B}) \preceq B$. In this case, we call the filter a residue filter.
\end{theorem}

\begin{proof}
    By \Cref{lemma:infinite_rule_limit}, it suffices to prove that the finite rule $\Update_{U, B}^k$ is free for all $B \in \mathcal{B}$ and $k \in \bZ^+$. We proceed by induction. For $k = 1$, we have
    $\pld(\Update_{U, B}^k) = \sup\{L \in \mathcal{L} : L \preceq B\} \preceq B$ for any $B \in \mathcal{B}$.
    Moreover, if this is true for $k \in \bZ^+$, then by \Cref{lemma:restricted_adversary} and by assumption,
    \begin{align*}
        \pld(\Update_{U, B}^{k + 1})
            & = \sup\{L \oplus \pld(\Update_{U, B}^{k + 1}|_L) : L \in \mathcal{L}, L \preceq B\} \\
            & = \sup\{L \oplus \pld(\Update_{U, U(B, L)}^k) : L \in \mathcal{L}, L \preceq B\} \\
            & \preceq \sup\{L \oplus U(B, L) : L \in \mathcal{L}, L \preceq B\} \preceq B
    \end{align*}
    for any $B \in \mathcal{B}$ as well.
\end{proof}


It turns out that residue filters already capture the existing privacy filters in addition to the natural filter, as we will see.

\begin{remark}
    The GDP filter \cite{KoskelaTH23,ST22}, which permits a sequence of queries that are $\mu_1, \dots, \mu_t$-GDP as long as $\mu_1^2 + \dots + \mu_t^2 \leq \mu^2$ for a fixed budget $\mu$, is exactly the residue filter with initial budget $B = G_\mu$ and update rule $U(G_\mu, L) := G_{\sqrt{\mu^2 - \nu^2}}$ where $\nu \geq 0$ is minimal such that $L \preceq G_\nu$. Because
    \begin{align*}
        U(G_\mu, L) \oplus L \preceq G_{\sqrt{\mu^2 - \nu^2}} \oplus G_\nu = G_\mu,
    \end{align*}
    \Cref{thm:residue_filter} immediately yields that the GDP filter is free.
\end{remark}

\begin{remark}
    The pure DP filter \cite{RogersRUV16}, which permits a sequence of queries that are $\eps_1, \dots, \eps_t$-DP as long as $\eps_1 + \dots + \eps_t \leq \eps$ for a fixed budget $\eps$, is exactly the residue filter with initial budget $B = R_\eps$ and update rule $U(R_\eps, L) := R_{\eps - \eps_L}$ where $\eps_L \geq 0$ is minimal such that $L \preceq R_{\eps_L}$. Because
    \begin{align*}
        U(R_\eps, L) \oplus L \preceq R_{\eps - \eps_L} \oplus R_{\eps_L} \preceq R_\eps,
    \end{align*}
    \Cref{thm:residue_filter} immediately yields that the pure DP filter is free.
\end{remark}

Unlike the GDP and pure DP forms of differential privacy, R\'enyi DP \cite{FeldmanZ21} cannot be certified exactly by domination by a canonical PLD, so the R\'enyi filter cannot directly be expressed as a residue filter. Nonetheless:

\begin{remark}
    The R\'enyi filter \cite{FeldmanZ21}, which permits a sequence of queries that are $(\alpha, \rho_1), \dots, (\alpha, \rho_t)$-RDP as long as $\rho_1 + \dots + \rho_t \leq \rho$ for a fixed budget $\rho$ and order $\alpha$, can be expressed as a budget maintaining filter satisfying a residue condition. The maintained budget is initially $\rho$ and, upon receiving an $(\alpha, \rho_i)$-RDP query, the filter accepts the query as long as $\rho_i \leq \rho$ and updates the budget to $\rho - \rho_i$. By linearity properties of RDP, this update satisfies an analogous residue condition to \eqref{eq:residue_condition} and hence the same techniques can be used to show that this filter is free.
\end{remark}

In particular, derivative filters including the zCDP and the advanced composition filter seen in e.g. \cite{WRRW23} are also a form of residue filter.

Perhaps more notable is that the natural filter, whenever it is free, can also be understood as a residue filter. We will discuss this in detail in \S\ref{sec:natural_filter}. The key takeaway is that residue filters either directly or in essence subsume previous privacy filters as well as the natural filter itself. This is strongly suggestive evidence that filters based on update rules satisfying a residue-like condition should be viewed as a standard technique for constructing new privacy filters. We strengthen this claim by constructing a new filter for GDP budgets that strictly improves upon the GDP filter.

\begin{theorem}
    \label{thm:gdp_res_filter}
    Consider the update rule
    \begin{align*}
        U(G_\mu, L) := G_{\mu'}, \quad \mu' \geq 0 \text{ maximal such that } G_{\mu'} \oplus L \preceq G_\mu
    \end{align*}
    for GDP budgets $G_\mu$ and arbitrary queries $L \preceq G_\mu$. This rule satisfies the residue condition \eqref{eq:residue_condition} and hence the update filter derived from it is a residue filter. We call this filter the GDP residue filter.
\end{theorem}

It is immediate that this update rule satisfies the residue condition\footnote{One can show carefully that the maximum exists by noticing that $\cLgdp$ is totally-ordered and applying \Cref{lemma:total_ordering_implies_commutativity}, which we develop in the next section.}. The update can be computed efficiently by binary search over $\mu' \in [0, \mu]$. It is also clear that the GDP residue filter cannot produce an updated budget smaller than the standard GDP filter. What is perhaps less obvious is that the GDP residue filter is different from the GDP filter. We analyze the performance of these two filters for pure DP queries and show that, asymptotically, the GDP filter deducts over 50\% more square budget per query than the GDP residue filter.

\begin{theorem}
    \label{thm:gdp_residue_pure_query}
    For a budget $\mu > 0$ and $\eps > 0$, let $\mu_{\textrm{na\"ive}}$ denote the na\"ive GDP filter update applied to budget $G_\mu$ and query $R_\eps$. Likewise, let $\mu_{\textrm{res}}$ denote GDP residue. Then
    \begin{enumerate}
        \item $\mu^2 - \mu_{\textrm{na\"ive}}^2 = \underbrace{4\Phi^{-1}(1/(1 + e^{-\eps}))^2}_{=: C_{\textrm{na\"ive}}(\eps)}$; and
        \item $\mu^2 - \mu_{\textrm{res}}^2 = \underbrace{8\ln(\cosh(\eps/2))}_{=: C_{\textrm{res}}(\eps)} + u, 0 \leq u \leq \frac{4\eps^2}{\mu_{\textrm{res}}^2}$.
    \end{enumerate}
\end{theorem}

A proof is provided in \Cref{app:gdp_residue}. To help interpret these bounds, we can compose standard power series expansions to obtain $C_{\textrm{na\"ive}}(\eps) = \frac{\pi}{2}\eps^2 + O(\eps^4)$ and $C_{\textrm{res}}(\eps) = \eps^2 + O(\eps^4)$. Therefore, for sufficiently small $\eps$ and large residue $\mu_{\textrm{res}}$, the error term $\frac{4\eps^2}{\mu_{\textrm{res}}^2}$ vanishes and hence the square GDP cost of the na\"ive update is asymptotically a factor $\frac{\pi}{2} \approx 1.57$ greater than the square GDP cost of the GDP residue. \Cref{fig:residue-to-GDP-filter} shows empirically that this holds for large and small queries.

\begin{figure}[htbp]
    \centering
    \begin{subfigure}[t]{0.32\textwidth}
        \centering
        \includegraphics[width=\linewidth]{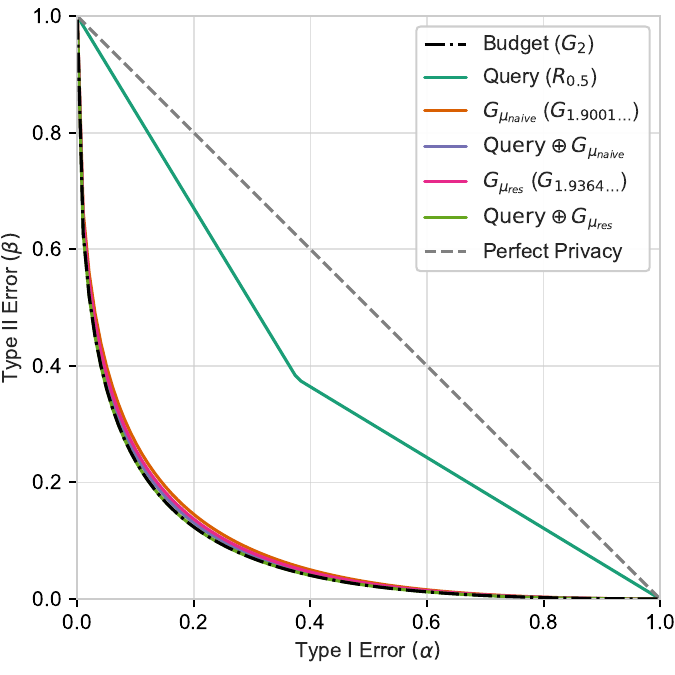}
        \caption{Small query, large budget: square GDP cost ratio ${\sim}1.56$.}
        \label{subfig:gdp_res_good_query}
    \end{subfigure}\hfill
    \begin{subfigure}[t]{0.32\textwidth}
        \centering
        \includegraphics[width=\linewidth]{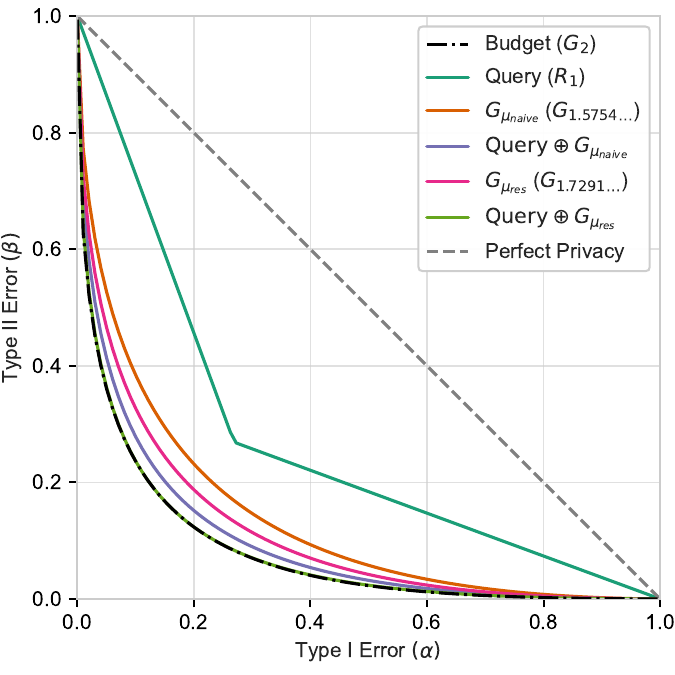}
        \caption{Large query, large budget: square GDP cost ratio ${\sim}1.50$.}
        \label{subfig:gdp_res_bad_query}
    \end{subfigure}\hfill
    \begin{subfigure}[t]{0.32\textwidth}
        \centering
        \includegraphics[width=\linewidth]{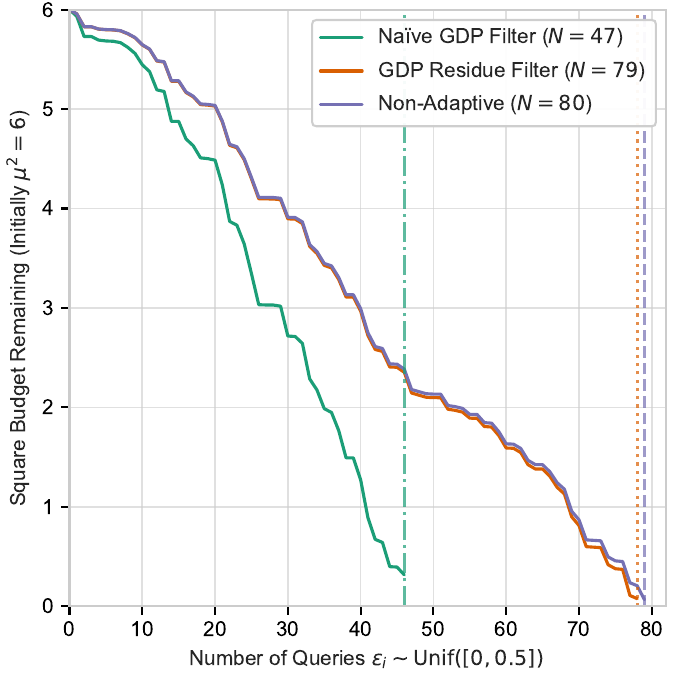}
        \caption{Sampled residue vs na\"ive filter trajectory, non-adaptive baseline.}
        \label{subfig:res_naive_trajectory}
    \end{subfigure}
    \caption{The na\"ive GDP filter directly approximates the query by GDP and determines a new budget $\mu_{naive}$. When composing $G_{\mu_{naive}}$ with the query (purple), a gap with the original budget shows that budget was wasted.
    The GDP residue filter instead searches for the largest valid residual GDP budget $\mu_{res}$. Composing $G_{\mu_{res}}$ with the query almost matches the budget: very little is wasted.}
    \label{fig:residue-to-GDP-filter}
\end{figure}

As a result, we may expect that the GDP residue filter permits approximately 50\% more queries than the na\"ive GDP filter. Indeed, the following result characterizes the number of queries permitted by each filter. The na\"ive filter is driven by the costs $C_{\textrm{na\"ive}}(\eps)$ and the residue filter is essentially driven by the smaller costs $C_{\textrm{res}}(\eps)$. Consequently, for an average query sequence, the residue filter asymptotically permits the same number of queries as if the queries were issued non-adaptively, which is in turn a factor $\pi/2$ more queries than permitted by the na\"ive GDP filter. The following result is proved in \Cref{app:gdp_residue}. \Cref{fig:residue-to-GDP-filter} shows a sample trajectory of both filters compared to the non-adaptive baseline closely matching our theoretical prediction.

\begin{theorem}
    \label{thm:gdp_residue_many_pure_queries}
    Fix $0 \leq \eps_1, \eps_2, \dots \leq E$. Let $N_{\textrm{na\"ive}}(\mu)$ denote the number of $\eps_i$-DP queries accepted by the na\"ive GDP filter with budget $\mu$. Let $N_{\textrm{res}}(\mu)$ be the number of queries accepted by the GDP residue filter with budget $\mu$. Finally, let $N_{\textrm{max}}(\mu) := \max\{n : R_{\eps_1} \oplus \dots \oplus R_{\eps_n} \preceq G_\mu\}$ denote the number of queries that could be accepted if they were issued non-adaptively. Then
    \begin{enumerate}
        \item $N_{\textrm{na\"ive}}(\mu) = \max\{n : \sum_{i = 1}^n C_{\textrm{na\"ive}}(\eps_i) \leq \mu^2\}$; and
        \item $\max\{n : \sum_{i = 1}^n C_{\textrm{res}}(\eps_i) \leq \mu^2 - O_E(\mu)\} \leq N_{\textrm{res}}(\mu) \leq N_{\textrm{max}}(\mu) \leq \max\{n : \sum_{i = 1}^n C_{\textrm{res}}(\eps_i) \leq \mu^2\}$.
    \end{enumerate}
    In particular, for i.i.d. $\eps_i \sim \textrm{Uniform}([0, E])$, $\mu = \omega(1)$, and $E = o(1)$, we have
    \begin{align*}
        \E[N_{\textrm{max}}(\mu)] \sim \E[N_{\textrm{res}}(\mu)] \sim \frac{\pi}{2}\E[N_{\textrm{na\"ive}}(\mu)].
    \end{align*}
\end{theorem}

%% file: sections/natural_filter.tex
\section{Characterizing the Natural Filter}
\label{sec:natural_filter}

We now turn toward the natural filter \Cref{alg:natural_filter} with some budget $B$, queries restricted to some set $\cL$, and, optionally, a query capacity $k \in \bZ^+ \cup \{\infty\}$. For simplicity, we will assume that $\Iddist \in \cL$.
Surprisingly, as we will see, ensuring that the query sequences $L_1 \oplus \dots \oplus L_t$ are ``branchwise'' bounded by $B$ does not ensure that the overall interaction between the filter and an adversarial analyst will also have privacy cost bounded by $B$.

Our goal is to understand when this filter is free, i.e. comes at no additional privacy cost. In this section, we give a dynamic programming characterization of the privacy cost of the natural filter (\Cref{thm:free_natural_filter}), we show that the natural filter is universally free exactly when it can be captured by a residue filter (\Cref{thm:universal_free_natural_filter}), and finally we show that this can only happen in general for query families that are closed under composition and totally-ordered (\Cref{thm:total_ordering}).


Now, notice that the natural filter \Cref{alg:natural_filter} with budget $B$, query set $\cL$, and capacity $k$ is free exactly when $\pld(\Nat_{\cL, B}^k) \preceq B$ where
\begin{align*}
    \Nat_{\cL, B}(L_1, \dots, L_t)
        := \{L_{t + 1} \in \cL : L_1 \oplus \dots \oplus L_{t + 1} \preceq B\}.
\end{align*}
In particular, by applying \Cref{lemma:restricted_adversary} to the rule $\Nat_{\cL, B}^k$, we immediately get the following characterization of the cost of interacting with the natural filter.

\begin{theorem}
    \label{thm:free_natural_filter}

    Let $\cL \subseteq \cLall$, let $B \in \cLall$, and consider, for $L' \preceq B$, the PLD operator
    \begin{align*}
        V_{k + 1}(L') := \sup\{L \oplus V_k(L' \oplus L) : L \in \cL, L' \oplus L \preceq B\}
    \end{align*}
    where $V_1(L') := \sup\{L \in \cL : L' \oplus L \preceq B\}$. Then, for any $k \in \bZ^+$ and legal $L_1, \dots, L_t \in \cL$, we have $\pld(\Nat_{\cL, B}^k\mid_{L_1, \dots, L_t}) = V_{k - t}(L_1 \oplus \dots \oplus L_t)$ and, in particular, $\pld(\Nat_{\cL, B}^k) = V_k(\Iddist)$.
\end{theorem}

Informally, $V_{k - t}(L_1 \oplus \dots \oplus L_t)$ is the maximal remaining privacy cost of an adversary that has already played moves $L_1, \dots, L_t$. Naturally, when only one move is left to select, the effect of adaptivity vanishes and the adversary simply selects the valid move of maximal cost.

It is also of interest to understand when a given family of queries $\cL$ admits universally free natural filters, i.e. when the natural filter is free for an arbitrary budget $B$. In the case of the natural filter, we show that it is universally free exactly when it can be realized as a residue filter.
For convenience, given $k \in \bZ^+ \cup \{\infty\}$ we let $\cL^{< k} := \{L_1 \oplus \dots \oplus L_t : L_1, \dots, L_t \in \cL, 0 \leq t < k\}$ denote the closure of $\cL$ under fewer than $k$ convolutions and we let $L \oplus \cL' := \{L \oplus L' : L' \in \cL'\}$ for a set of PLDs $\cL'$.

\begin{theorem}
    \label{thm:universal_free_natural_filter}
    Fix $\cL \subseteq \cLall$, $k \in \bZ^+ \cup \{\infty\}$, and consider the update rule $U : \cLall \times \cL \to \cLall$ given by $U(B, L) := \sup\{L' \in \cL^{< k} : L \oplus L' \preceq B\}$. The following are equivalent:
    \begin{enumerate}[(i)]
        \item For all $B \in \cLall$, $\pld(\Nat_{\cL, B}^k) \preceq B$;
        \item For all $L \in \cL$ and $\cL' \subseteq \cL^{< k}$, $L \oplus \sup{\cL'} = \sup(L \oplus \cL')$;
        \item For all $B \in \cLall$ and $L \in \cL$, $U(B, L) \oplus L \preceq B$; and
        \item For all $B \in \cLall$, $\pld(\Update_{U, B}^k) \preceq B$.
    \end{enumerate}
    In any of these cases, the natural filter $\Nat_{\cL, B}^k$ and the residue filter $\Update_{U, B}^k$ are identical.
\end{theorem}

\begin{proof}
    Without loss of generality, by \Cref{lemma:infinite_rule_limit}, assume that $k$ is finite.

    First, we show that (i) implies (ii), so suppose that the natural filter is free for every budget $B$ and let $L \in \cL$ and $\cL' \subseteq \cL^{<k}$. For every $L' \in \cL'$, we have $L \oplus L' \preceq L \oplus \sup{\cL'}$ by \Cref{prop:pld_convolution_properties} and hence $\sup(L \oplus \cL') \preceq L \oplus \sup{\cL'}$ for free. As for the reverse inequality, consider the length-$2$ rule $\Gamma(\lambda) := \{L\}, \Gamma(L) := \cL'$ as well as the budget $B := \sup(L \oplus \cL')$. Now, for any $M \in \Adv(\Gamma)$, notice that we can transform $M$ into $M' \in \Adv(\Nat_{\cL, B}^k)$ such that $\pld(M) = \pld(M')$. Indeed, let $M = M_1 \otimes M_2 \in \Adv(\Gamma)$ and fix $y_1 \in \cY_1$ so that $\pld(M_2(\cdot; y_1)) = L_{2,1} \oplus \dots \oplus L_{2,t} \in \cL^{< k}$. For each $L_{2,i}$, we can find a mechanism $M_{2,i}$ with $\pld(M_{2,i}) = L_{2,i}$. Now, set $M'_2(\cdot ; y_1) := M_{2,1} \otimes \dots \otimes M_{2,t}$ so that $\pld(M'_2(\cdot; y_1)) = L_{2,1} \oplus \dots \oplus L_{2,t} = \pld(M_2(\cdot; y_1))$. Clearly, $M' := M_1 \otimes M'_2 \in \Adv(\Nat_{\cL, B}^k)$ and, by \Cref{prop:composition_convolution}, we have $\pld(M) = \pld(M')$, so by \Cref{lemma:restricted_adversary} and since the natural filter is free for $B$ by assumption, we have
    \begin{align*}
        L \oplus \sup{\cL'}
            = \sup_{M \in \Adv(\Gamma)} \pld(M)
            \preceq \sup_{M' \in \Adv(\Nat_{\cL, B}^k)} \pld(M')
            \preceq B
            = \sup(L \oplus \cL').
    \end{align*}

    That (ii) implies (iii) is immediate from the construction of the update rule $U$ by taking $\cL' := \{L' \in \cL^{< k} : L \oplus L' \preceq B\}$; likewise (iii) implies (iv) follows immediately from \Cref{thm:residue_filter}, so what remains to show is that, if the residue filter $\Update_{U, B}$ is universally free, then so is $\Nat_{\cL, B}$, in which case the two filters coincide.

    To this end, first notice that $\Update_{U, B}^k$ is always at least as permissive a filter as $\Nat_{\cL, B}^k$. Indeed, for any $L \in \cL$ and $L' \in \cL^{< k}$ such that $L \oplus L' \preceq B$, it follows immediately from the construction of $U$ that $L' \preceq U(B, L)$. Thus
    for any legal query sequence $L_1, \dots, L_t$,
    \begin{align*}
        L_{t + 1} \in \Nat_{\cL, B}^k(L_1, \dots, L_t)
            & \implies L_1 \oplus \dots \oplus L_{t + 1} \preceq B \\
            & \implies L_2 \oplus \dots \oplus L_{t + 1} \preceq U(B, L_1) \\
            & \implies \dots \\
            & \implies L_{t + 1} \preceq U(\dots U(B, L_1) \dots, L_t) \\
            & \implies L_{t + 1} \in \Update_{U, B}^k(L_1, \dots, L_t).
    \end{align*}
    Therefore, if $\Update_{U, B}^k$ is universally free, then so is $\Nat_{\cL, B}^k$.

    Finally, notice that, as long as the residue condition for $U$ is satisfied, then $\Nat_{\cL, B}^k$ is at least as permissive a filter as $\Update_{U, B}^k$. Indeed, for any $L \in \cL$ and $L' \in \cL^{< k}$ for which $L' \preceq U(B, L)$, \Cref{prop:pld_convolution_properties} and the residue condition ensure that $L \oplus L' \preceq L \oplus U(B, L) \preceq B$. Inductively, just as before, this implies that
    for any legal query sequence $L_1, \dots, L_t$,
    we have that $L_{t + 1} \in \Update(U, B)^k(L_1, \dots, L_t) \implies L_{t + 1} \in \Nat_{\cL, B}^k(L_1, \dots, L_t)$ and so the two filters are indeed equivalent.
\end{proof}

Our key result is that these conditions can essentially only be satisfied for a family of queries $\cL$ whose closure under composition $\cL^{< k}$ is totally-ordered. For a distribution $A$ over $\overline{\bR}$, we denote by $\supp(A) := \{t \in \mathbb{R} : \forall \text{open } U \ni t, A(U) > 0\}$ the support of $A$, namely the finite points with locally positive probability mass. We say that $A$ is non-degenerate if $|\supp(A)| > 1$, i.e. $A$ is supported on at least two finite points.

\begin{theorem}
    \label{thm:total_ordering}
    Let $\cL \subseteq \cLall$ and $k \in \bZ^+ \cup \{\infty\}$. If $\cL^{< k}$ is totally-ordered, then $\pld(\Nat_{\cL, B}^k) \preceq B$ for every $B \in \cLall$. Conversely, if the non-degenerate members of $\cL^{< k}$ are not totally-ordered, then $\pld(\Nat_{\cL, B}^{2k-3}) \npreceq B$ for some choice of $B \in \cLall$.
\end{theorem}

Proving the forward direction is straightforward. We show in general that the supremum commutes with convolution for totally-ordered families of queries.

\begin{lemma}
    \label{lemma:total_ordering_implies_commutativity}
    Let $L$ be a PLD and let $\cL$ be a totally-ordered set of PLDs. Then $L \oplus \sup{\cL} = \sup(L \oplus \cL)$.
\end{lemma}

\begin{proof}
    As in the proof of \Cref{thm:universal_free_natural_filter}, we already have $\sup(L \oplus \cL) \preceq L \oplus \sup{\cL}$ for free, so due to \Cref{prop:sup-conv}, we just need to show that $h_{L \oplus \sup{\cL}} \preceq h_{\sup(L \oplus \cL)}$. By total-ordering of $\cL$ and by \Cref{prop:sup-conv}, we can find a sequence $L_1 \preceq L_2 \preceq \dots \in \cL$ so that $h_{L_t} \to h_{\sup{\cL}}$ pointwise. For any $x \in \bR^\times$, monotone convergence yields
    \begin{align*}
        h_{L \oplus \sup{\cL}}(x)
            & = \E_{Z \sim L}[\E_{Z' \sim \sup{\cL}}[(1 - xe^{-(Z + Z')})_+]] \\
            & = \E_{Z \sim L}[h_{\sup{\cL}}(xe^{-Z})] \\
            & = \E_{Z \sim L}[\lim_{t \to \infty} h_{L_t}(xe^{-Z})] \\
            & = \lim_{t \to \infty}\E_{Z \sim L}[h_{L_t}(xe^{-Z})] \\
            & = \lim_{t \to \infty}\E_{Z \sim L}[\E_{Z' \sim L_t}[(1 - xe^{-(Z + Z')})_+]] \\
            & = \lim_{t \to \infty} h_{L \oplus L_t}(x) \\
            & \leq h_{\sup(L \oplus \cL)}(x).
    \end{align*}
\end{proof}

Proving the other direction is substantially more challenging. The key lemma is that, given a pair of non-degenerate queries that are not totally-ordered, composing one of these queries with their supremum can amplify the gaps between them, leading to a violation of the commutativity condition. The proof of this technical result involves analyzing the topological structure of the hockey-stick curves of unordered queries. The details are given in~\Cref{sec:total_ordering_proof}. Note that non-degeneracy is key as the following result fails for degenerate queries such as $L_1 := 1/2 \cdot 1_{-\ln{2}} + 1/2 \cdot 1_\infty$ and $L_2 := 2/3 \cdot 1_0 + 1/3 \cdot 1_\infty$.

\begin{lemma}
    \label{lemma:commutativity_implies_total_ordering}
    Let $L_1$ and $L_2$ be non-degenerate PLDs such that $L_1 \npreceq L_2$ and $L_2 \npreceq L_1$, i.e. $\{L_1, L_2\}$ is not totally-ordered. Then, for either $i = 1$ or $i = 2$, we have
    \begin{align*}
        L_i \oplus \sup\{L_1, L_2\} \succ \sup\{L_i \oplus L_1, L_i \oplus L_2\}.
    \end{align*}
\end{lemma}

We note that the proof of the total-ordering result in the concurrent work~\cite{tran2026f} contains an error here. Their proof of Theorem 3.6 in Section D.2 incorrectly applies their assumption of unbounded support in order to justify the topological conditions required for privacy gap amplification of composition. That form of the result can be corrected by strengthening the assumption in~\Cref{lemma:commutativity_implies_total_ordering} to require that one of the two queries have continuous support. In this case, the proof techniques in~\Cref{sec:total_ordering_proof} can be substantially simplified at the expense of significantly limiting the negative implications of~\Cref{thm:total_ordering}, which we detail in~\Cref{sec:implications-special-case}.

\begin{proof}[Proof of \Cref{thm:total_ordering}]
    The forward direction now follows directly from \Cref{lemma:total_ordering_implies_commutativity} and \Cref{thm:universal_free_natural_filter}. As for the reverse direction, suppose that the subfamily of non-degenerate queries in $\cL^{< k}$ is not totally-ordered, i.e. there exist $L_1, L_2 \in \cL^{< k}$ that are not ordered. By \Cref{lemma:commutativity_implies_total_ordering}, we have that $L_i \oplus \sup\{L_1, L_2\} \succ \sup\{L_i \oplus L_1, L_i \oplus L_2\}$. On the other hand, if the commutativity condition of \Cref{thm:universal_free_natural_filter} were satisfied for $\cL^{< 2k-3}$, then, letting $L_i = L_{i,1} \oplus \dots \oplus L_{i,t}$ ($t < k$), we would get
    \begin{align*}
        L_i \oplus \sup\{L_1, L_2\}
            & = L_{i,1} \oplus \dots \oplus L_{i,t} \oplus \sup\{L_1, L_2\} \\
            & = L_{i,1} \oplus \dots \oplus L_{i,t-1} \oplus \sup\{L_{i,t} \oplus L_1, L_{i,t} \oplus L_2\} \\
            & = \dots \\
            & = \sup\{L_{i,1} \oplus \dots \oplus L_{i,t} \oplus L_1, L_{i,1} \oplus \dots \oplus L_{i,t} \oplus L_2\} \\
            & = \sup\{L_i \oplus L_1, L_i \oplus L_2\},
    \end{align*}
    which is a contradiction.
\end{proof}

%% file: sections/implications.tex
\section{Notable Implications}
\label{sec:implications-special-case}

\subsection{Natural Filters Are NOT Free in General}

In this section, we demonstrate examples of queries and budgets that do and do not pass the characterizing conditions of the previous section. For visualization purposes, it will be convenient to pass to tradeoff curves. Recall that the analogous operation for the PLD supremum is the lower convex envelope (cf. \Cref{prop:sup-conv}). Recall also that the order of domination for PLDs is reversed for their tradeoff curves (see \Cref{prop:f-hs-equivalence}). In this section, we will generally discuss piecewise linear tradeoff curves, which can be composed exactly as follows. A proof is given in \Cref{app:implications-special-case}.

\begin{proposition}[Exact composition of piecewise-linear tradeoff curves]
    \label{prop:piecewise-linear-composition}
    Consider two piecewise linear tradeoff curves $f_1$ and $f_2$ corresponding to PLDs $L_1$ and $L_2$.
    Let $(w_{i, j})_{j = 1}^{n_i}$ be the lengths of the intervals over which $f_i$ is linear, let
    $(f'_{i,j})_{j = 1}^{n_i}$ be the slope of $f_i$ over those intervals, and let $\delta_i := 1 - f_i(0)$.
    Now, let $(j_k, \ell_k)_{k = 1}^{n_1 n_2}$ be an enumeration of $[n_1] \times [n_2]$ such that the product of the slopes $f'_{1, j_k} \cdot f'_{2, \ell_k}$ is descending in $k$, i.e. $f_{1, j_1} \cdot f_{2, \ell_1}$ is the product of the steepest slopes of $f_1$ and $f_2$. Then, the composed tradeoff curve $f$ corresponding to $L_1 \oplus L_2$ is also a piecewise linear tradeoff curve over intervals of length $w_k := w_{1,j_k}w_{2,\ell_k}$ with slopes $f'_k := -f'_{1,j_k} \cdot f'_{2,\ell_k}$ such that $1 - f(0) = \delta_1 + \delta_2 - \delta_1\delta_2$.
\end{proposition}

An immediate implication of \Cref{thm:total_ordering} is that the natural filter is not universally free for general queries, even for a single adaptive round. This can be seen by noticing that $\cLapprox$ is not totally-ordered (recall \Cref{def:pld_class}). \Cref{subfig:counter-example-crossing} demonstrates this visually by passing to the tradeoff curves of $\cLapprox$.

\begin{corollary}
    \label{corollary:no-free-filters}
    The natural filter for $\geq 2$ adaptive $(\eps, \delta)$-DP queries is not universally free, i.e. there exists a budget $B \in \cLall$ for which $\pld(\Nat_{\cLapprox, B}^2) \npreceq B$.
\end{corollary}

\begin{figure}[htbp]
    \centering
    \begin{subfigure}[t]{0.48\textwidth}
        \centering
        \includegraphics[width=\linewidth]{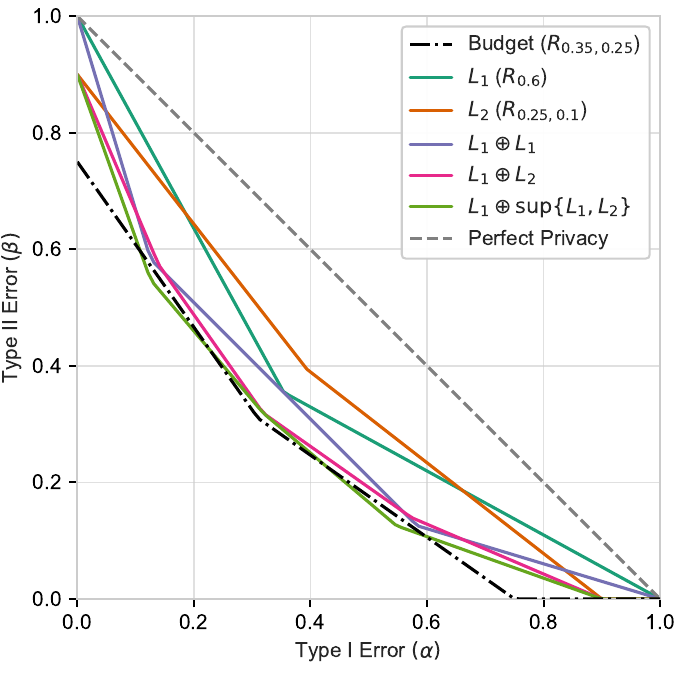}
        \caption{Queries $L_1, L_2 \in \cLapprox$ such that $\{L_1, L_2\}$ is not totally-ordered and a budget $B \in \cLapprox$ such that $L_1 \oplus L_1, L_1 \oplus L_2 \preceq B$, yet $L_1 \oplus \sup\{L_1, L_2\} \npreceq B$.}
        \label{subfig:counter-example-crossing}
    \end{subfigure}\hfill
    \begin{subfigure}[t]{0.48\textwidth}
        \centering
        \includegraphics[width=\linewidth]{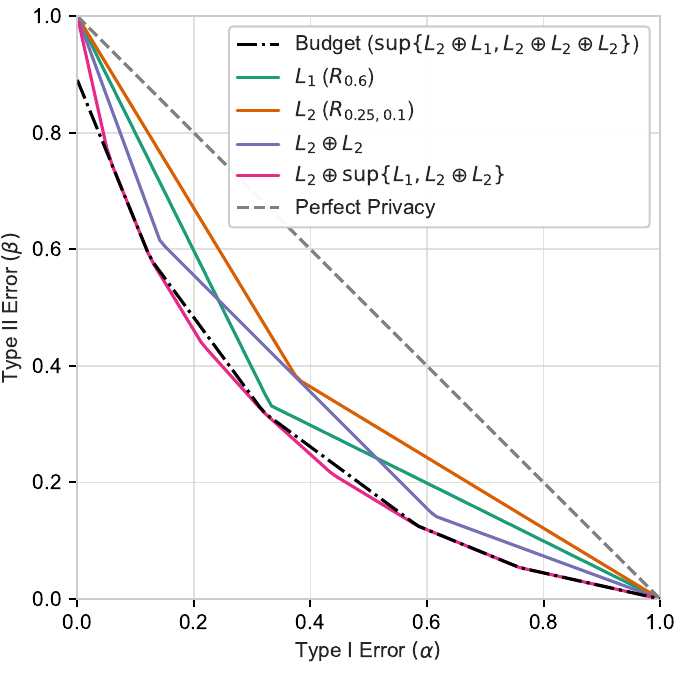}
        \caption{Queries $L_1, L_2 \in \cLpure$ such that $\{L_1, L_2 \oplus L_2\}$ is not totally-ordered.}
        \label{subfig:pure-dp-counter-example}
    \end{subfigure}
    \caption{Counterexamples showing that natural filters fail for $(\eps, \delta)$-DP queries and $(\eps, \delta)$-DP budgets as well as $\eps$-DP queries and general budgets.}
    \label{fig:counter-example}
\end{figure}

Another interesting case is that of the $\epsilon$-DP queries $\cLpure$. Indeed, this family is totally-ordered and known to admit privacy filters under basic composition \cite{RogersRUV16}. However, this family of queries is not closed under composition and $\cLpure^{<3}$ is not totally-ordered anymore. Consequently, \Cref{thm:total_ordering} yields the following result, illustrated visually in \Cref{subfig:pure-dp-counter-example}.

\begin{corollary}
    \label{cor:no_free_filter_for_pure_dp}
    The natural filter for queries $\cLpure$ is universally free for $k = 2$ queries but not $k \geq 3$ queries, i.e. $\pld(\Nat_{\cLpure, B}^2) \preceq B$ for all $B \in \cLall$ whereas $\pld(\Nat_{\cLpure, B}^3) \npreceq B$ for some $B \in \cLall$.
\end{corollary}

We have shown that natural filters are not universally free, even when restricted to $(\eps, \delta)$-queries.
These results are surprising, as all previous filter results were that filters are free (RDP/zCDP \cite{FeldmanZ21,lecuyer2021practical}, GDP \cite{KoskelaTH23,ST22}) or asymptotically free ($(\epsilon,\delta)$-DP \cite{RogersRUV16, WRRW23}).
One might hope that restricting the budget would be sufficient to ensure free natural filters. Two natural candidates come to mind. First, restricting $B \in \cL$---a practical condition of choosing privacy bounds from the legal query set itself. Second, we could restrict $B$ to a GDP bound $G_\mu$ or an approximate GDP bound $R_{0, \delta} \oplus G_\mu$ to support approximate $(\epsilon, \delta)$-DP queries. These are good candidates, as the only known free natural filters \cite{KoskelaTH23,ST22} are for $\cLgdp$ queries and budgets. Unfortunately, neither restriction is sufficient.

The first point, that restricting the budget to $\cL$ does not ensure the freeness of the natural filter, can be seen by applying \Cref{thm:free_natural_filter}. \Cref{subfig:counter-example-crossing} shows an example of queries $L_1, L_2 \in \cLapprox$ and a budget $B \in \cLapprox$ such that $L_1 \oplus L_1, L_1 \oplus L_2 \preceq B$ yet $L_1 \oplus \sup\{L_1, L_2\} \npreceq B$. In this case, \Cref{thm:free_natural_filter} implies $\pld(\Nat_{\cLapprox, B}^2) \npreceq B$.

\begin{corollary}
    \label{cor:no-free-filter-for-f-in-F}
    The natural filter is not free for $(\eps, \delta)$-DP queries, even for an $(\eps, \delta)$-DP budget, i.e. for some choice of $B \in \cLapprox$, we have that $\pld(\Nat_{\cLapprox, B}^2) \npreceq B$.
\end{corollary}

\begin{figure}[htbp]
    \centering
    \begin{subfigure}[t]{0.48\textwidth}
        \centering
        \includegraphics[width=\linewidth]{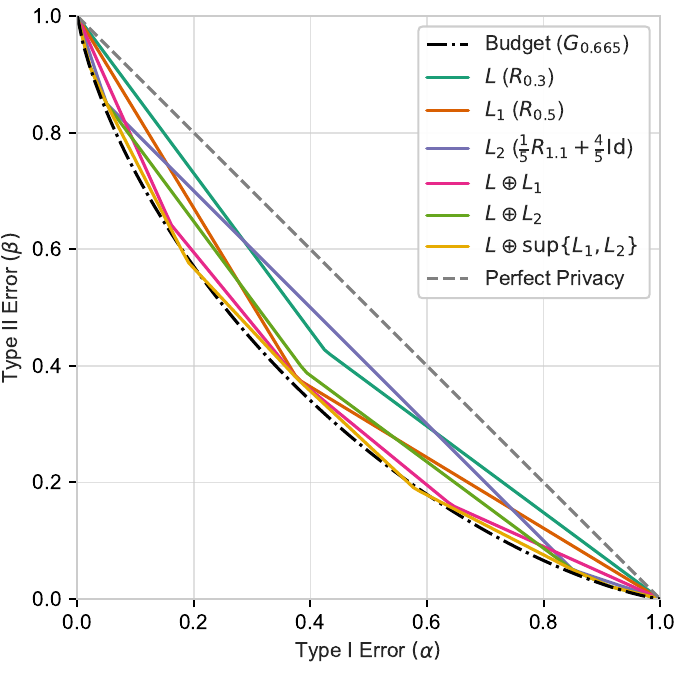}
        \caption{An example showing queries $L, L_1, L_2 \in \cLall$ and a budget $G_\mu \in \cLgdp$ such that $L \oplus L_1, L \oplus L_2 \preceq G_\mu$, yet $L \oplus \sup\{L_1, L_2\} \npreceq G_\mu$.}
        \label{subfig:counter-example-pure-gdp}
    \end{subfigure}\hfill
    \begin{subfigure}[t]{0.48\textwidth}
        \centering
        \includegraphics[width=\linewidth]{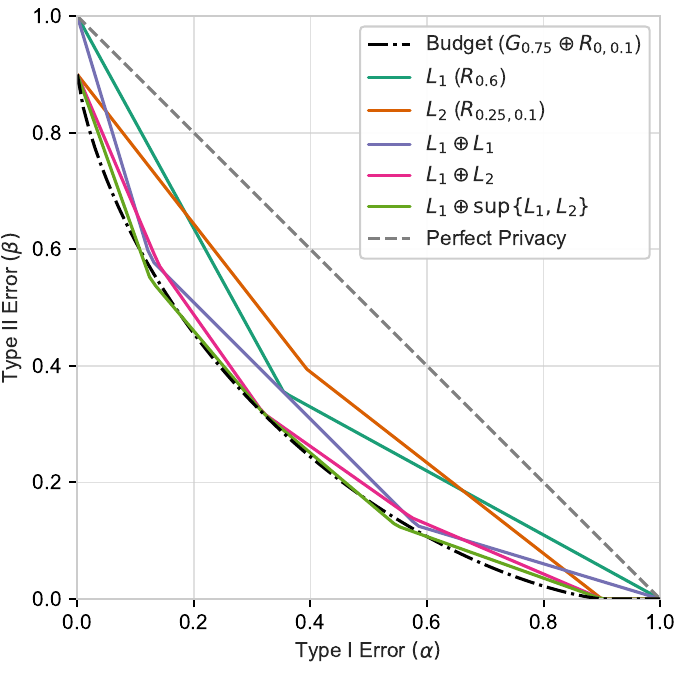}
        \caption{An example showing queries $L_1, L_2 \in \cLapprox$ and a budget $(\mu, \delta)$ such that $L_1 \oplus L_1, L_1 \oplus L_2 \preceq G_\mu \oplus R_{0, \delta}$, yet $L_1 \oplus \sup\{L_1, L_2\} \npreceq G_\mu \oplus R_{0, \delta}$.}
        \label{subfig:counter-example-approx-gdp}
    \end{subfigure}
    \caption{Failure of natural filters for GDP and approximate GDP budgets.}
    \label{fig:counter-example-GDP}
\end{figure}

Similarly, restricting the budget to a well-behaved class does not generally lead to free filters either. That GDP budgets do not suffice for general queries can again be shown by appealing to \Cref{thm:free_natural_filter} in the same manner, shown visually in \Cref{subfig:counter-example-pure-gdp}. Likewise, approximate GDP budgets do not give free natural filters for $(\eps, \delta)$-DP queries, as illustrated in \Cref{subfig:counter-example-approx-gdp}.

\begin{corollary}
    \label{cor:no-free-filter-for-f-gaussian}
    The natural filter is not free for general queries, even for a GDP budget, i.e. for some $\mu \geq 0$, $\pld(\Nat_{\cLall, G_\mu}^2) \npreceq G_\mu$.

    Furthermore, the natural filter is not free for $(\eps, \delta)$-DP queries, even if the budget is chosen to be approximate GDP, i.e. for some $\mu \geq 0$ and $\delta \in [0, 1]$, $\pld(\Nat_{\cLapprox, G_\mu \oplus R_{0, \delta}}^2) \npreceq G_\mu \oplus R_{0, \delta}$.
\end{corollary}

\subsection{Families of Queries With Universally Free Natural Filters}

\Cref{thm:total_ordering} ensures that query families that are closed under composition and totally-ordered admit universally free natural filters.

In particular, since $G_{\mu_1} \oplus G_{\mu_2} = G_{\sqrt{\mu_1^2 + \mu_2^2}}$ and since $G_{\mu_1} \preceq G_{\mu_2}$ whenever $0 \leq \mu_1 \leq \mu_2$, we have:

\begin{corollary}\label{cor:gaussian-filter}
    The natural filter for GDP queries is universally free, i.e. $\pld(\Nat_{\cLgdp, B}) \preceq B$ for all $B \in \cLall$.
\end{corollary}

Likewise, as $R_{0, \delta_1} \oplus R_{0, \delta_2} = R_{0, \delta_1 + \delta_2 - \delta_1\delta_2}$ and as $R_{0, \delta_1} \preceq R_{0, \delta_2}$ for $0 \leq \delta_1 \leq \delta_2 \leq 1$, we have:

\begin{corollary}\label{cor:pure-delta-filter}
    The natural filter for $(0, \delta)$-DP queries is universally free, i.e. for every $B \in \cLall$ we have $\pld(\Nat_{\{R_{0, \delta : \delta \in [0, 1]}\}, B}) \preceq B$.
\end{corollary}

%% file: sections/upper_bound.tex
\section{Upper Bound for the Natural $(\eps, \delta)$-DP PLD Filter}
\label{sec:natural-apx}

We consider the natural $(\eps, \delta)$-DP PLD filter, which is simply an instantiation of \Cref{alg:natural_filter} with the $(\eps, \delta)$-DP budget $B = R_{\eps, \delta}$ where, recalling \Cref{def:pld_class},
\begin{align*}
    R_{\eps, \delta} =
    \begin{cases}
        \infty & \text{w.p. } \delta \\
        \eps & \text{w.p. } \frac{1 - \delta}{e^{-\eps} + 1} \\
        -\eps & \text{w.p. } \frac{1 - \delta}{e^{\eps} + 1}
    \end{cases}.
\end{align*}
It is well known that a PLD $L$ is dominated by $R_{\eps, \delta}$ iff $h_L(e^\eps) \leq \delta$. In other words, this filter uses the ``natural'' condition that checks whether $L_1 \oplus \cdots \oplus L_t$ satisfies $(\eps, \delta)$-DP.

Our main result is that, while the above natural $(\eps, \delta)$-DP PLD filter can fail, it cannot fail \emph{too badly}; in the sense that we still get a DP guarantee with a blow-up in parameters that is only polylogarithmic in $k, 1/\delta$, as stated below\footnote{Recall the notation $\Nat_{\cL, B}^k$ from \Cref{sec:natural_filter}.}. We remark that our upper bound holds for the family $\cLsym$ of all \emph{symmetric} PLDs, which is a standard practical setting (see \Cref{app:symmetry}). 

\begin{theorem} \label{thm:main-ub-apx}
    Let $\eps > 0, \delta \in (0, 1/4)$. Then, the natural $(\eps, \delta)$-DP PLD filter for symmetric queries satisfies $\Paren{\chi \cdot (\eps + 1), \chi \cdot \delta}$-DP; in other words,
    \begin{align*}
        \pld\Paren{\Nat_{\cLsym, R_{\eps, \delta}}^k} \preceq R_{\chi \cdot (\eps + 1), \chi \cdot \delta},
    \end{align*}
    where $$\chi = O\Paren{\Paren{\log k + \log\Paren{1 + \frac{1}{\eps}}} \cdot \log(1/\delta)}.$$ 
\end{theorem}

When $\eps = \Omega(1)$, the above guarantee is simply $(\chi' \cdot \eps, \chi' \cdot \delta)$-DP where $\chi' = O(\log k \cdot \log(1/\delta))$. 

\begin{remark}
    In our current proof, the parameters are quite flexible and we can tradeoff the blow-up factors of the two parameters $\eps, \delta$. For example, by changing $\eps_i$ in \Cref{subsec:disc} to $\frac{2\eps}{2^{n - 1 - i}}$., we can prove a guarantee of $\left(\chi \cdot \eps, \chi \cdot \frac{\delta}{\eps} \right)$-DP for $\chi = O(\log k \cdot \log(1/\delta))$. However, we choose to keep only a single setting of parameters in \Cref{thm:main-ub-apx} for simplicity.

    For the same reason, we do not attempt to optimize $\chi$. In particular, it is plausible that our techniques can achieve $\chi = \tilde{O}(\log k + \log(1/\delta))$ with a more refined argument. Nevertheless, our proof cannot eliminate the dependency of $k$ or $\delta$ in $\chi$, and this remains an interesting open question.
\end{remark}

\paragraph{Proof Overview.}
We start with an informal overview of our proof, which consists of three steps:
\begin{itemize}
\item \textbf{Privacy Loss Discretization.} Before we analyze the privacy filter, we start by providing a discretization argument for PLDs. Namely, we pick discretization points $\eps_0, \dots, \eps_n$. Then, we ``round'' the PLD $L$ up or down to be these values (or their negations). We show that this process results in $L$ being ``sandwiched'' between two instances of ``generalized randomized response'', which has the following form: Randomly select index $i$ from a certain probability distribution and then output $i$ together with the $\eps_i$-DP randomized response. We show that $L$ is dominated by the ``rounded up'' (aka ``pessimistic'') version of such randomized response while it dominates the ``rounded down'' (aka ``optimistic'') version. We remark that our discretization argument is different than previous work (e.g. \cite{MM18,KoskelaJH20,DoroshenkoGKKM22,gopi2021numerical,GhaziK0M22}) which do \emph{not} produce a generalized randomized response instance (especially for the pessimistic version). Indeed, we exploit the symmetry of $L$ to achieve our discretization.
\item \textbf{Translation to Additive Filter.} We consider the following ``additive'' filter, in which there is a non-negative real value budget associated with each $\eps_i$ and each time we simply subtract from the budget; the filter rejects if we run out of budget for any $\eps_i$. With the above discretization, we can show that the $(\eps, \delta)$-PLD filter is ``stronger'' than such an additive filter (where the budget spent in each step is from the discretization). In other words, if one passes the PLD filter, then one also passes the additive filter. This means that, for the purpose of privacy guarantee, we may analyze the latter.
\item \textbf{Analysis of Additive Filter via Post-Processing.} In the final step, we provide the privacy analysis of the additive filter. Roughly speaking, we show that the adversary's view of the additive filter can be shown to be a post-processing of a certain pre-specified number of randomized response (``with high probability''). Due to the post-processing property of DP, we can thus conclude that this is as private as composition of randomized response, which we then show to be $(\eps', \delta')$-DP. 
\end{itemize}

The above overview skips multiple subtle details, and the remainder of this section is devoted to carrying out the above steps in full.

\input{sections/upper_bound_proof}

%% file: sections/upper_bound_proof.tex
\subsection{Discretization of PLD}
\label{subsec:disc}

Our discretization is based on a variant of randomized response where with a different prescribed probability, we run a randomized response with a different value of $\eps$. We formalize this below.

\begin{definition}[Randomized Response]
We use $\RR_{\eps}$ to denote the $\eps$-DP binary randomized response. That is, for $x \in \{0, 1\}$, \begin{align*}
\RR_{\eps}(x) =
\begin{cases}
x &\text{ with probability } \frac{1}{e^{-\eps} + 1}, \\
1 - x &\text{ with probability } \frac{1}{e^{\eps} + 1}.
\end{cases}
\end{align*}
Note that for $\eps = \infty$, $\RR_{\eps}(x) = x$ with probability one.
\end{definition}

Let $\Delta_{n-1} := \{\bp = (p_1, \dots,p_n) \in \mathbb{R}_{\geq 0} \mid p_1 + \cdots + p_n = 1\}$ be the $(n - 1)$-simplex, and let $\cE_n := \{\beps = (\eps_0, \dots, \eps_n) \in \R_{\geq 0}^{n + 1} \mid 0 = \eps_0 \leq \eps_1 \leq \cdots \leq \eps_n = \infty\}$.

\begin{definition}[Generalized Randomized Response]
%
For any $\bp \in \Delta_{n - 1}$ and $\beps \in \cE_n$,
we write $\RR_{\bp, \beps}$ to denote the following mechanism (where $x$ denotes the input):
\begin{itemize}
\item Sample index $i \in [n]$ according to $\bp$,
\item Sample $y \sim \RR_{\eps_i}(x)$, 
\item Output $(i, y)$.
\end{itemize}
We also define the following ``rounded down'' version $\RRd_{\bp, \beps}$ of the above mechanism as follows where we use $\eps_0$ to denote $0$ throughout:
\begin{itemize}
\item Sample index $i \in [n]$ according to $\bp$,
\item Sample $y \sim \RR_{\eps_{i-1}}(x)$, 
\item Output $(i, y)$.
\end{itemize}
\end{definition}
Note that $\RR_{\bp, \beps}$ and $\RRd_{\bp, \beps}$ only differ on the second step, where the former invokes $\RR_{\eps_i}$ but the latter invokes $\RR_{\eps_{i-1}}$.


In addition to Generalized Randomized Response, another notion we will use for the discretization step is the notion of \emph{discretized privacy loss profile}, which is a vector $\bp = (p_1, \dots, p_n)$ where $p_i$ is the probability that the privacy loss falls into $[\eps_{i-1}, \eps_i)$ or $(-\eps_i, -\eps_{i-1}]$. This is formalized below.

\begin{definition}[Discretized Privacy Loss Profile]
For any PLD $L$, the discretized privacy loss profile of $L$ and $\beps \in \cE_n$, denoted by $\disc_{\beps}(L)$, is the vector $\bp = (p_1, \dots, p_n) \in \Delta_{n - 1}$ where
\begin{align*}
p_i = 
\begin{cases}
\Pr_{Z \sim L}\left[\eps_{i-1} \leq |Z| < \eps_i\right] & \text{ if } i \ne n, \\
\Pr_{Z \sim L}\left[\eps_{n-1} \leq |Z| \right] &\text{ if } i = n.
\end{cases}
\end{align*}
\end{definition}

The ``optimistic'' and ``pessimistic'' discretization of $L$ is exactly based on $\RRd_{\bp, \beps}$ and $\RR_{\bp, \beps}$, and we have the following theorems which give the aforementioned ``sandwich'' around the pair $L$.

\begin{theorem}[Pessimistic Estimate] \label{thm:pessimistic}
Suppose that $L$ is a symmetric PLD, and let $\bp$ be its discretized privacy loss profile. Then, $\pld(\RR_{\bp, \beps}(0), \RR_{\bp, \beps}(1)) \succeq L$.
\end{theorem}

\begin{theorem}[Optimistic Estimate] \label{thm:optimistic}
Suppose that $L$ is a symmetric PLD, and let $\bp$ be its discretized privacy loss profile. Then, $\pld(\RRd_{\bp, \beps}(0), \RRd_{\bp, \beps}(1)) \preceq L$. 
\end{theorem}

The remainder of this subsection are devoted to the proofs of these two theorems.

\subsubsection{Useful Lemmata}

Before we can prove \Cref{thm:pessimistic,thm:optimistic}, we will first prove a couple of technical lemmata that would be useful. First is the symmetry of $\RR$ and $\RRd$, which are simple to prove.

\begin{lemma} \label{lem:mult-rr-sym}
For any $\bp, \beps$, the pairs $(\RR_{\bp, \beps}(0), \RR_{\bp, \beps}(1))$ and $(\RRd_{\bp, \beps}(0), \RRd_{\bp, \beps}(1))$ are symmetric. 
\end{lemma}

\begin{proof}
Consider the Markov kernel $\kappa$ that takes in $(i, y)$ and outputs $(i, 1 - y)$. It is obvious to see that $\kappa(\RR_{\bp, \beps}(0)) = \kappa(\RR_{\bp, \beps}(1))$ and $\kappa(\RR_{\bp, \beps}(1)) = \kappa(\RR_{\bp, \beps}(0))$. Thus, $(\RR_{\bp, \beps}(0), \RR_{\bp, \beps}(1))$ dominates $(\RR_{\bp, \beps}(1), \RR_{\bp, \beps}(0))$ and vice versa. In other words, $(\RR_{\bp, \beps}(0), \RR_{\bp, \beps}(1))$ is symmetric.

The proof for $(\RRd_{\bp, \beps}(0), \RRd_{\bp, \beps}(1))$ is analogous.
\end{proof}

For the proofs below, it will be more convenient to deal with pairs rather than privacy loss. For this, we give an alternate definition of discretized privacy loss profile below.

\begin{definition}
Let $(P, Q)$ be any pair of distributions.
For any $0 \leq \eps < \eps'$, let
\begin{align*}
I_{\eps, \eps'} :=
\begin{cases}
[\eps, \eps') &\text{ if } \eps' \ne \infty, \\
[\eps, \eps') \cup \{\infty\} &\text{ if } \eps' = \infty.
\end{cases}
\end{align*}
Moreover, let $$\Omega^{P, Q}_{\eps,\eps'} = \left\{\omega~\middle|~\log\left(\frac{dP}{dQ}(\omega)\right) \in I_{\eps, \eps'}\right\},$$
and similarly
$$\Omega^{Q, P}_{\eps,\eps'} = \left\{\omega~\middle|~\log\left(\frac{dQ}{dP}(\omega)\right) \in I_{\eps, \eps'}\right\}.$$
\end{definition}

The following observation follows directly from the definition of PLD.

\begin{observation}
For any $L = \pld(P, Q)$ and $\beps \in \cE_n$, the discretized privacy loss profile of $L$ is equal to $\bp = (p_1, \dots, p_n)$ where 
$$p_i = P\Paren{\Omega^{P, Q}_{\eps_{i - 1},\eps_i} \cup \Omega^{Q, P}_{\eps_{i - 1},\eps_i}}.$$
\end{observation}

The next lemma is a type of symmetry with regards to $\Omega^{P, Q}_{\eps,\eps'}$ and $\Omega^{P, Q}_{\eps',\eps}$:

\begin{lemma} \label{lem:inverse-cdf-same}
For any symmetric pair $(P, Q)$ and $0 \leq \eps < \eps'$, 
\begin{align*}
P\Paren{\Omega^{P, Q}_{\eps,\eps'}} = Q\Paren{\Omega^{Q, P}_{\eps,\eps'}} &&\text{ and} &&P\Paren{\Omega^{Q, P}_{\eps,\eps'}} = Q\Paren{\Omega^{P, Q}_{\eps,\eps'}}.
\end{align*}
\end{lemma}

\begin{proof}
We only prove the first equality as the second is analogous.

Let $F_{P, Q}$ and $F_{Q, P}$ denote the CDFs of $\PLD(P, Q)$ and $\PLD(Q, P)$, respectively.
Since the hockey-stick curves $H_{x}(P \parallel Q)$ and $H_{x}(Q \parallel P)$ are equal, \cite[Lemma 23]{ZhuDW22} implies that $F_{P, Q}$ and $F_{Q, P}$ are exactly the same. Thus, for $\eps' \ne \infty$ we have
\begin{align*}
P\Paren{\Omega^{P, Q}_{\eps,\eps'}} = \lim_{\teps \to (\eps)^-, \teps' \to (\eps')^-} F_{P, Q}(\teps') - F_{P, Q}(\teps) 
= \lim_{\teps \to (\eps)^-, \teps' \to (\eps')^-} F_{Q, P}(\teps') - F_{Q, P}(\teps) = Q\Paren{\Omega^{Q, P}_{\eps,\eps'}}.
\end{align*}
For $\eps' = \infty$, we similarly have
\begin{align*}
P\Paren{\Omega^{P, Q}_{\eps,\eps'}} = \lim_{\teps \to (\eps)^-} 1 - F_{P, Q}(\teps) 
= \lim_{\teps \to (\eps)^-} 1 - F_{Q, P}(\teps) = Q\Paren{\Omega^{Q, P}_{\eps,\eps'}}. &\qedhere
\end{align*}
\end{proof}

\subsubsection{Pessimistic Estimate Proof}

We are now ready to prove \Cref{thm:pessimistic}.

\begin{proof}[Proof of \Cref{thm:pessimistic}]
Suppose that $L = \pld(P, Q)$.
Since both $(P, Q)$ and $(\RR_{\bp, \beps}(0), \RR_{\bp, \beps}(1))$ are symmetric (\Cref{lem:mult-rr-sym}), by \Cref{lem:sym-check-only-large-alpha}, it suffices to show that $H_x(P \parallel Q) \leq H_x(\RR_{\bp, \beps}(0) \parallel \RR_{\bp, \beps}(1))$ for any $x \geq 1$.

For $i \in [n]$, let $P_i$ (resp. $Q_i$) denote the distribution of $P$ (resp. $Q$) conditioned on the value being in $\Omega^{P, Q}_{\eps_{i-1},\eps_i} \cup \Omega^{Q, P}_{\eps_{i-1},\eps_i}$. From \Cref{lem:inverse-cdf-same}, we have that $P$ and $Q$ can be written as mixture distributions $\sum_{i \in [n]} p_i \cdot P_i$ and $\sum_{i \in [n]} p_i \cdot Q_i$, respectively. Since $P_i$'s support are disjoint and $Q_i$'s support are disjoint (and the supports of $P_i, Q_j$ are disjoint for all $i \ne j$), we simply have
\begin{align*}
H_x(P \parallel Q) = \sum_{i \in [n]} p_i \cdot H_x(P_i \parallel Q_i).
\end{align*}

Similarly, we have
\begin{align*}
H_x(\RR_{\bp, \beps}(0) \parallel \RR_{\bp, \eps}(1)) = \sum_{i \in [n]} p_i \cdot H_x(\RR_{\eps_i}(0) \parallel \RR_{\eps_i}(1)).
\end{align*}
Thus, it suffices to show that\footnote{This essentially follows from \cite{KairouzOV15} but we repeat the full proof here for completeness.} $H_x(P_i \parallel Q_i) \leq H_x(\RR_{\eps_i}(0) \parallel \RR_{\eps_i}(1))$ for all $i \in [n]$ and $x \geq 1$. To see this is true, consider three cases:
\begin{itemize}
\item \textbf{Case I: $x \geq e^{\eps_i}$.} Then, $H_x(P_i \parallel Q_i) = 0 = H_x(\RR_{\eps_i}(0) \parallel \RR_{\eps_i}(1))$
\item \textbf{Case II: $1 \leq x \leq e^{\eps_{i - 1}}$.} In this case, we have $H_x(\RR_{\eps_i}(0) \parallel \RR_{\eps_i}(1)) = \frac{1 - x \cdot e^{-\eps_i}}{1 + e^{-\eps_i}}.$ Meanwhile, we have
\begin{align*}
H_x(P_i \parallel Q_i) &= P_i\Paren{\Omega^{P, Q}_{\eps_{i - 1},\eps_i}} - x \cdot Q_i\Paren{\Omega^{P, Q}_{\eps_{i - 1},\eps_i}} \\
&= \frac{1}{p_i} \cdot \left(P\Paren{\Omega^{P, Q}_{\eps_{i - 1},\eps_i}} - x \cdot Q\Paren{\Omega^{P, Q}_{\eps_{i - 1},\eps_i}}\right) \\
&= \frac{P\Paren{\Omega^{P, Q}_{\eps_{i - 1},\eps_i}} -x \cdot Q\Paren{\Omega^{P, Q}_{\eps_{i - 1},\eps_i}}}{P\Paren{\Omega^{P, Q}_{\eps_{i - 1},\eps_i}} + P\Paren{\Omega^{Q, P}_{\eps_{i - 1},\eps_i}}} \\
(\text{\Cref{lem:inverse-cdf-same}}) &= \frac{P\Paren{\Omega^{P, Q}_{\eps_{i - 1},\eps_i}} - x \cdot Q\Paren{\Omega^{P, Q}_{\eps_{i - 1},\eps_i}}}{P\Paren{\Omega^{P, Q}_{\eps_{i - 1},\eps_i}} + Q\Paren{\Omega^{P, Q}_{\eps_{i - 1},\eps_i}}} \\
&= \frac{x + 1}{1 + \frac{Q\Paren{\Omega^{P, Q}_{\eps_{i - 1},\eps_i}}}{P\Paren{\Omega^{P, Q}_{\eps_{i - 1},\eps_i}}}} - x \\
&\leq \frac{x + 1}{1 + e^{-\eps_i}} - x \\
&= \frac{1 - x \cdot e^{-\eps_i}}{1 + e^{-\eps_i}},
\end{align*}
where the inequality is due to the definition of $\Omega^{P, Q}_{\eps_{i - 1},\eps_i}$.
\item \textbf{Case III: $x \in \Paren{e^{\eps_{i - 1}}, e^{\eps_{i}}}$.} Since $H_{x}$ is convex in $x$, we have
\begin{align*}
H_x(P_i \parallel Q_i) 
&\leq \frac{x - e^{\eps_{i - 1}}}{e^{\eps_{i}} - e^{\eps_{i - 1}}} \cdot H_{e^{\eps_{i}}}(P_i \parallel Q_i) + \frac{e^{\eps_i} - x}{e^{\eps_{i}} - e^{\eps_{i - 1}}} \cdot H_{e^{\eps_{i-1}}}(P_i \parallel Q_i) \\
&\leq \frac{x - e^{\eps_{i - 1}}}{e^{\eps_{i}} - e^{\eps_{i - 1}}} \cdot H_{e^{\eps_{i}}}(\RR_{\eps_i}(0) \parallel \RR_{\eps_i}(1)) + \frac{e^{\eps_i} - x}{e^{\eps_{i}} - e^{\eps_{i - 1}}} \cdot H_{e^{\eps_{i-1}}}(\RR_{\eps_i}(0) \parallel \RR_{\eps_i}(1)) \\
&= H_x(\RR_{\eps_i}(0) \parallel \RR_{\eps_i}(1)),
\end{align*}
where the second inequality follows from Cases I and II above, and the last equality is due to linearity of $H_x(\RR_{\eps_i}(0) \parallel \RR_{\eps_i}(1))$ for $x \in I_{e^{\eps_{i - 1}}, e^{\eps_{i}}}$. \qedhere
\end{itemize}
\end{proof}

\subsubsection{Optimistic Estimate Proof}

We next prove \Cref{thm:optimistic}. This proof is in fact simpler than the previous one and only two cases are required.

\begin{proof}[Proof of \Cref{thm:optimistic}]
Similar to the proof of \Cref{thm:pessimistic}, it suffices to show that $$H_x(P_i \parallel Q_i) \geq H_x(\RR_{\eps_{i-1}}(0) \parallel \RR_{\eps_{i-1}}(1))$$ for all $i \in [n]$ and $x \geq 1$. To see this is true, consider two cases:
\begin{itemize}
\item \textbf{Case I: $x \geq e^{\eps_{i - 1}}$.} In this case, we simply have $H_x(\RR_{\eps_{i-1}}(0) \parallel \RR_{\eps_{i-1}}(1)) = 0 \leq H_x(P_i \parallel Q_i)$.
\item \textbf{Case II: $x \in [1, e^{\eps_{i - 1}})$.} In this case, we have $H_x(\RR_{\eps_{i-1}}(0) \parallel \RR_{\eps_{i-1}}(1)) = \frac{1 - x \cdot e^{-\eps_{i-1}}}{1 + e^{-\eps_{i-1}}}.$ Meanwhile, similar to Case II in the proof of \Cref{thm:pessimistic}, we have
\begin{align*}
H_x(P_i \parallel Q_i) = \frac{x + 1}{1 + \frac{Q\Paren{\Omega^{P, Q}_{\eps_{i - 1},\eps_i}}}{P\Paren{\Omega^{P, Q}_{\eps_{i - 1},\eps_i}}}} - x 
\geq \frac{x + 1}{1 + e^{-\eps_{i-1}}} - x 
= \frac{1 - x \cdot e^{-\eps_{i-1}}}{1 + e^{-\eps_{i-1}}},
\end{align*}
where the inequality is due to the definition of $\Omega^{P, Q}_{\eps_{i - 1},\eps_i}$. \qedhere
\end{itemize}
\end{proof}

\subsection{From PLD Filter to Additive Filter}

Having established fundamental results on discretization, we can now state our ``additive'' privacy filter. In this filter, we have a budget $\bell = (\ell_1, \dots, \ell_n) \in \R^n_{\geq 0}$ where informally $\ell_i$ corresponds to the budget for $\RR_{\eps_i}$. When a mechanism comes, subtract its discretized privacy profile $\bp$ from the budget. If the budget remains non-negative in all coordinates (denoted by\footnote{Here we write $\bzero$ to denote the all-zero vector, and $\bu \geq \bv$ for vectors $\bu, \bv$ if the inequality holds coordinate-wise.} $\bell \geq \bzero$), we subtract it from the budget and return the response to the adversary. This is formalized in \Cref{alg:additive-profile-filter}.

\begin{algorithm}[tb]
    \caption{Additive Discretized PL Profile Filter}
    \label{alg:additive-profile-filter}
    \begin{algorithmic}
        \REQUIRE Discretization points $\beps \in \cE_n$, budgets $\bell \in \R^n_{\geq 0}$, analyst $\cA$, query capacity $k$, dataset $D$
        \FOR{$t \gets 1, \dots, k$}
            \STATE $\cA(Y_1, \dots, Y_{t - 1})$ gives mechanism $M_t$ that has PLD $L_t \in \cLsym$
            \STATE $\bp^t \gets $ Discretized privacy loss profile of $L_t$ w.r.t. $\beps$
            \STATE $\bell \gets \bell - \bp^t$
      \IF{$\bell \geq \bzero$}
        \YIELD $Y_t \sim M_t(D)$
      \ELSE
        \BREAK
      \ENDIF
        \ENDFOR
    \end{algorithmic}
\end{algorithm}

\paragraph{Setting of Parameters.} For the remainder of this section, we set the parameter $\beps$ as follows:
\begin{itemize}
\item Let $n = 1 + \left\lceil \log_2\Paren{\frac{k(\eps + 2)}{\eps}} \right\rceil$.
\item Let $\eps_n = \infty$
\item For all $i \in [n - 1]$, let $\eps_i = \frac{\eps + 2}{2^{n - 1 - i}}$.
\end{itemize}

Furthermore, let us define $\bell^* = (\ell^*_1, \dots, \ell^*_n)$ as follows:
\begin{align} \label{eq:additive-budget-values}
\ell^*_i =
\begin{cases}
k & \text{ if } i = 1 \\
8 \cdot \max \left\{\ell \in \N ~\middle\vert~ H_{e^{\eps}}(\RR_{\eps_{i-1}}^{\otimes (\ell - 1)}(0) \parallel \RR_{\eps_{i-1}}^{\otimes (\ell - 1)}(1)) \leq 2\delta \right\} &\text{ if } 1 < i < n, \\
4\delta & \text{ if } i = n
\end{cases}
\end{align}

\paragraph{Additive Profile Filter as Adversary.} Following notation in \Cref{sec:adversary_theory,sec:natural_filter}, the Additive Discretized PL Profile Filter (\Cref{alg:additive-profile-filter}) can be viewed as the rule
\begin{align*}
\AddFilter_{\beps, \bell}^k(L_1, \dots, L_t) = \{L_{t + 1} \in \cLsym : \disc_{\beps}(L_1) + \cdots + \disc_{\beps}(L_{t+1}) \leq \bell\}.
\end{align*}

Our main theorem in this subsection is that the Additive Discretized PL Profile Filter--when initialized with the above discretization points and privacy budgets--is less strict than the natural $(\eps, \delta)$-DP PLD filter: 

\begin{theorem} \label{thm:additive-filter-vs-pld-filter}
For $\bell = \bell^*$ as in \Cref{eq:additive-budget-values} and any $L_1, \dots, L_t \in \cLsym$, we have
\begin{align} \label{eq:filter-less-restricted}
\Nat_{\cLsym, R_{\eps, \delta}}^k(L_1, \dots, L_t) \subseteq \AddFilter_{\beps, \bell}^k(L_1, \dots, L_t).
\end{align}
\end{theorem}

\begin{proof}
Consider any $L_1, \dots, L_{t + 1} \in \cLsym$, and let $\bp^1 = \disc_{\beps}(L_1), \dots, \bp^{t+1} = \disc_{\beps}(L_{t+1})$.
Note that \eqref{eq:filter-less-restricted} is equivalent to showing the following: if $h_{L_1 \oplus \cdots \oplus L_{t+1}}({e^\eps}) \leq \delta$, then we must have $p^1_i + \cdots + p^{t+1}_i \leq \ell^*_i$ for all $i \in [n]$. The case $i = 1$ is obvious since $\ell^*_1 = k \geq t$. We will henceforth only consider $i > 1$. 

We will show the contrapositive: Suppose that $p^1_i + \cdots + p^{t+1}_i > \ell^*_i$ for some $i \in \{2, \dots, n\}$. We will show that $h_{L_1 \oplus \cdots \oplus L_{t+1}}({e^\eps}) > \delta$. Recall (from \Cref{thm:optimistic}) that $L_j$ dominates $(\RRd_{\bp^j, \eps}(0), \RRd_{\bp^j, \eps}(1))$. Thus, it suffices to show that 
\begin{align}
H_{e^\eps}\left(\RRd_{\bp^1, \eps}(0) \times \cdots \times \RRd_{\bp^{t+1}, \eps}(0) ~\middle\|~ \RRd_{\bp^1, \eps}(1) \times \cdots \times \RRd_{\bp^{t+1}, \eps}(1)\right) > \delta. \label{eq:target-proof-additive}
\end{align}

For $X \in \{0, 1\}$, consider post-processing $(y^1, \dots, y^{t+1}) \sim \RRd_{\bp^1, \eps}(X) \times \cdots \times \RRd_{\bp^{t+1}, \eps}(X)$ as follows, where $\tau \in \N$ will be specified below.
\begin{itemize}
\item Recall that each $y^j = (y^j_1, y^j_2)$ is a tuple of two numbers: $y^j_1 \in [n]$ and $y^j_2 \in \{0, 1\}$. Let $j_1 < \dots < j_m$ be indices $j$ such that $y^{j} = i$.
\item If $m < \tau$, then output $\bot$. Otherwise, output $(y^{j_1}, \dots, y^{j_\tau})$.
\end{itemize}
Let $R(X)$ denote the distribution of the output of the above process, and let $q$ denote the probability that the above process does \emph{not} output $\bot$. (Note that $q$ does not depend on $X$.) Observe that, conditioned on \emph{not} outputting $\bot$, the distribution $R(X)$ is exactly the same as $\RR_{\eps_{i-1}}^{\otimes \tau}(X)$. Thus,
\begin{align}
&H_{e^\eps}\left(\RRd_{\bp^1, \eps}(0) \times \cdots \times \RRd_{\bp^{t+1}, \eps}(0) ~\middle\|~ \RRd_{\bp^1, \eps}(1) \times \cdots \times \RRd_{\bp^{t+1}, \eps}(1)\right) \nonumber \\
&\geq q \cdot H_{e^\eps}\left(\RR_{\eps_{i-1}}^{\otimes \tau}(0) ~\middle\|~ \RR_{\eps_{i-1}}^{\otimes \tau}(1)\right). \label{eq:lowerbound-hs-additive}
\end{align}

We will next consider two cases separately: $i = n$ and $1 < i < n$.
\paragraph{Case I: $i = n$.} In this case, set $\tau = 1$. Notice that $H_{e^\eps}(\RR_{\eps_{n-1}}^{\otimes \tau}(0) \parallel \RR_{\eps_{n-1}}^{\otimes \tau}(1)) = \frac{1 - e^{\eps - \eps_{n-1}}}{1 + e^{-\eps_{n-1}}}$. Due to our choice $\eps_{n-1} = \eps + 2$, we have $H_{e^\eps}(\RR_{\eps_{n-1}}^{\otimes \tau}(0) \parallel \RR_{\eps_{n-1}}^{\otimes \tau}(1)) > 0.5$. Observe also that
\begin{align*}
q &= 1 - (1 - p^1_n) \cdots (1  - p^{t+1}_n)
\geq 1 - e^{-p^1_n - \cdots -p^{t+1}_n}
\geq 1 - e^{-4\delta}
\geq 2\delta,
\end{align*}
where the first and last inequalities follow from the well-known\footnote{See e.g. \cite{log-ineq}.} inequalities $\frac{x}{1 + x} \leq \ln(1 + x) \leq x$ for all $x > -1$ (and from $\delta < 1/4$).
Combining the two inequalities and \eqref{eq:lowerbound-hs-additive}, we can conclude that \eqref{eq:target-proof-additive} holds (for $i = n$) as desired.

\paragraph{Case II: $1 < i < n$.} In this case, we set $\tau = \max \left\{\ell \in \N ~\middle\vert~ H_{e^{\eps}}\left(\RR_{\eps_{i-1}}^{\otimes (\ell - 1)}(0) ~\middle\|~ \RR_{\eps_{i-1}}^{\otimes (\ell - 1)}(1)\right) \leq 2\delta \right\}$. By the choice of $\tau$, we have $$H_{e^\eps}(\RR_{\eps_{i-1}}^{\otimes \tau}(0) \parallel \RR_{\eps_{i-1}}^{\otimes \tau}(1)) > 2\delta.$$

We will next bound $q$. Observe that $m$ is distributed exactly as $Z = Z_1 + \cdots + Z_{t + 1}$ where $Z_1 \sim \Ber(p^1_i), \dots, Z_{t + 1} \sim \Ber(p^{t + 1}_i)$. Thus, we have $\E[Z] = p^1_i + \cdots + p^{t + 1}_i > \ell^*_i$ and $\Var[Z] \leq \E[Z]$. Note that we set $\tau < \ell^*_i/2$. Thus, we can apply Chebyshev's inequality to conclude that
\begin{align*}
1 - q = \Pr[Z < \tau] \leq \Pr[Z < \E[Z]/2] \leq \frac{\Var[Z]}{(\E[Z]/2)^2} \leq \frac{4}{\E[Z]} < \frac{4}{\ell^*_i} < \frac{1}{2}.
\end{align*}
Combining the two inequalities and \eqref{eq:lowerbound-hs-additive}, we can conclude that \eqref{eq:target-proof-additive} holds as desired.
\end{proof}

\subsection{Analyzing Additive Discretized PL Profile Filter}

Due to \Cref{thm:additive-filter-vs-pld-filter}, it suffices to directly 
analyze the Additive Discretized PL Profile Filter. In particular, we will show the following:
\begin{theorem} \label{thm:main-ub-additive}
For any $\eps > 0, \delta \in (0, 1/4)$, let $\beps, \bell^*$ be as in \eqref{eq:additive-budget-values}. Then, 
\begin{align*}
\pld\Paren{\AddFilter_{\beps, \bell^*}^k} \preceq R_{\chi \cdot (\eps + 1), \chi \cdot \delta},
\end{align*}
where $$\chi = O\Paren{\Paren{\log k + \log\Paren{1 + \frac{1}{\eps}}} \cdot \log(1/\delta)}.$$ 
\end{theorem}

Before we prove \Cref{thm:main-ub-additive}, let us first note that it immediately implies the main theorem of this section (\Cref{thm:main-ub-apx}).

\begin{proof}[Proof of \Cref{thm:main-ub-apx}]
Combining \Cref{thm:additive-filter-vs-pld-filter,thm:main-ub-additive} respectively, we have
\begin{align*}
\pld\Paren{\Nat_{\cLsym, R_{\eps, \delta}}^k} \preceq \pld\Paren{\AddFilter_{\beps, \bell^*}^k} \preceq R_{\chi \cdot (\eps + 1), \chi \cdot \delta}. &\qedhere
\end{align*}
\end{proof}

The rest of this section is dedicated to the proof of \Cref{thm:main-ub-additive}.
Throughout this subsection, let $\bb^* \in \N^n$ be defined as follows (where $\bell^*$ is set as in \eqref{eq:additive-budget-values}):
\begin{align} \label{eq:sample-budget-setting}
b^*_i = 
\begin{cases}
0 &\text{ if } i = n \\
2\lceil \log(1/\delta) \rceil \cdot \ell^*_i &\text{ if } i \ne n.
\end{cases}
\end{align}

\subsubsection{Useful Lemmata}

We will also use the following version of ``triangle inequality''. Note that the assumption here is stronger than just assumption that the TV distance between $P, P'$ and $Q, Q'$ are smaller than $\gamma$. This stronger assumption allows us to derive an additive bound that is only $\gamma$, rather than a weaker bound $(x + 1) \cdot \gamma$ (if we were to only assume a bound on the TV distance). This is crucial for our proof of \Cref{thm:main-ub-additive}, which requires a large value of $x$.

\begin{lemma} \label{lem:triangle}
    Let $x, \gamma \geq 0$, and let $P, Q$ be distributions on domain $\Omega$. Furthermore, let $P', Q'$ be distributions on domain $\Omega \cup \{\omega^*\}$ such that the following holds: $P(E) \geq P'(E)$ and $Q(E) \geq Q'(E)$ for all $E \subseteq \Omega$, and $P'(\omega^*) \leq \gamma$. Then, $H_x(P \parallel Q) \leq H_x(P' \parallel Q') + \gamma$.
\end{lemma}

\begin{proof}
    We have
    \begin{align*}
        H_x(P \parallel Q) &= \sup_{E \subseteq \Omega} P(E) - x \cdot Q(E) \\
        &=  \sup_{E \subseteq \Omega} 1 - P(\Omega \setminus E) - x \cdot Q(E) \\
        &\leq \sup_{E \subseteq \Omega} 1 - P'(\Omega \setminus E) - x \cdot Q'(E) \\
        &= \sup_{E \subseteq \Omega} P'(\omega^*) + P'(E) - x \cdot Q'(E) \\
        &\leq P'(\omega^*) + H_x(P' \parallel Q') \\
        &\leq \gamma + H_x(P' \parallel Q'). \qedhere
    \end{align*}
\end{proof}

The next lemma asserts that, to check if a symmetric PLD dominates another, it suffices to check the hockey-divergence of order $x \geq 1$. This is unlike the (asymmetric) case where we need to check $x \in (0, 1)$ as well. This will be convenient for our proofs. We defer its (simple) proof to \Cref{app:symmetry}.

\begin{lemma} \label{lem:sym-check-only-large-alpha}
For any symmetric PLDs $L, L'$, $L \succeq L'$ iff $h_x(L) \geq h_x(L')$ 
for all $x \geq 1$.
\end{lemma}

We will use the following ``Multiplicative Azuma's Inequality'' from \cite[Corollary 11]{KQ21}. This is an extension of the classic multiplicative Chernoff bound, with the difference that the distributions $\cD_i$ can depend on the previous outcomes $W_1, \dots, W_{i - 1}$. The only constraint required is that $\sum_{i \in [k]} \E_{W_i \sim \cD_i}[W_i] \leq \mu$ where $\mu$ is a prespecified parameter.

\begin{theorem}[Multiplicative Azuma's Inequality \cite{KQ21}] \label{thm:azuma}
Let $\mu > 0$.
Suppose that an adversary constructs random variables $W_1, \dots, W_k \in [0, 1]$ via the following process: When the outcomes of $W_1, \dots, W_{i - 1}$ are determined, the adversary can select the distribution $\cD_i$ with mean $\mu_i$ for $W_i$ under the constraint $\sum_{i \in [k]} \mu_i \leq \mu$.

If $W = \sum_{i} W_i$, then for any $\zeta > 0$,
\begin{align*}
\Pr[W \geq (1 + \zeta) \mu] \leq \exp\left(-\frac{\zeta^2}{2 + \zeta} \mu\right).
\end{align*}
\end{theorem}

It will also be convenient to state the following version of the tail bound for a similar setting as above but when $\mu$ is very small.

\begin{lemma} \label{lem:adaptive-markov}
Let the setting be exactly the same as in \Cref{thm:azuma}. Then, $\Pr[W \geq \mu / \zeta] \leq \zeta$.
\end{lemma}

\begin{proof}
This is simply an application of Markov's inequality, since $\E[W] \leq \mu$.
\end{proof}

\subsubsection{Privacy Analysis of RR: Comparing Different $\eps$'s}

Before we proceed to the main argument, we need some technical results. The main lemma below states that, if we take $\teps$-DP randomized response and compose it 12 times, then its pair dominates that of $2\teps$-DP randomized response. This result seems interesting on its own and might be useful beyond the context of this work. We remark that we did not attempt to optimize the constant ``12'' here; nevertheless, it is not hard to see that any constant less than 6 does \emph{not} work (by taking $\eps \to 0$ and looking at the hockey-stick divergence of order $x = 1$, i.e. the TV distance).

\begin{lemma} \label{lem:double-eps-dom}
For any $\teps \geq 0$,  $(\RR^{\otimes 12}_{\teps}(0), \RR^{\otimes 12}_{\teps}(1)) \succeq (\RR_{2\teps}(0), \RR_{2\teps}(1))$.
\end{lemma}

\begin{proof}
Let $t = 12$.
Since both $(\RR^{\otimes t}_{\teps}(0), \RR^{\otimes t}_{\teps}(1))$ and $(\RR_{2\teps}(0), \RR_{2\teps}(1))$ are symmetric, by \Cref{lem:sym-check-only-large-alpha}, it suffices to show that $H_x(\RR_{2\teps}(0) \parallel \RR_{2\teps}(1)) \leq H_x(\RR^{\otimes t}_{\teps}(0) \parallel \RR^{\otimes t}_{\teps}(1))$ for any $x \geq 1$. To see this is true, consider three cases:
\begin{itemize}
\item \textbf{Case I: $x \geq e^{2\teps}$.} Then, $H_x(\RR_{2\teps}(0) \parallel \RR_{2\teps}(1)) = 0 \leq H_x(\RR^{\otimes t}_{\teps}(0) \parallel \RR^{\otimes t}_{\teps}(1))$.
\item \textbf{Case II: $x = 1$.} Note that $H_x$ is simply the TV distance. Thus, $H_x(\RR_{2\teps}(0) \parallel \RR_{2\teps}(1)) = \frac{e^{2\teps} - 1}{e^{2\teps} + 1}.$ 

To lower bound $H_x(\RR^{\otimes t}_{\teps}(0) \parallel \RR^{\otimes t}_{\teps}(1))$, let
$S_i$ denote the set of all tuples in $\{0, 1\}^t$ such that the number of ones is exactly $i$. Let $P = \RR^{\otimes t}_{\teps}(0), Q = \RR^{\otimes t}_{\teps}(1)$ We thus have
\begin{align*}
H_x(\RR^{\otimes t}_{\teps}(0) \parallel \RR^{\otimes t}_{\teps}(1)) &\geq 
\sum_{i=0}^{t/2-1} \Paren{P(S_i) - Q(S_i)} \\
&= \sum_{i=0}^{t/2-1} \binom{t}{i} \left(\frac{e^{(t - i) \teps} - e^{i\teps}}{(e^{\teps}+1)^t}\right) \\
&= \sum_{i=0}^{t/2-1} \binom{t}{i}  \left(\frac{e^{i\teps}\Paren{e^{(t - 2i) \teps} - 1}}{(e^{\teps}+1)^t}\right) \\
&\geq \frac{e^{2\teps} - 1}{(e^{\teps}+1)^t} \Paren{\binom{t}{t/2-1}e^{(t/2-1)\teps} + \sum_{i=0}^{t/2-2} \binom{t}{i}  \Paren{e^{i\teps}\Paren{e^{(t - 2i - 2) \teps} + 1}}}.
\end{align*}
Thus, we have
\begin{align*}
&\frac{H_x(\RR^{\otimes t}_{\teps}(0) \parallel \RR^{\otimes t}_{\teps}(1))}{H_x(\RR_{2\teps}(0) \parallel \RR_{2\teps}(1))} \\
&= \frac{e^{2\teps} + 1}{(e^{\teps}+1)^t} \Paren{\binom{t}{t/2-1}e^{(t/2-1)\teps} + \sum_{i=0}^{t/2-2} \binom{t}{i}  \Paren{e^{i\teps}\Paren{e^{(t - 2i - 2) \teps} + 1}}} \\
&= \frac{1}{(e^{\teps}+1)^t} \Paren{\binom{t}{t/2-1}\Paren{e^{(t/2+1)\teps} + e^{(t/2-1)\teps}} + \sum_{i=0}^{t/2-2} \binom{t}{i}  \Paren{e^{(t-i)\teps} + e^{i\teps} + e^{(t-i-2)\teps} + e^{(i + 2)\teps}}} \\
&\geq \frac{1}{(e^{\teps}+1)^t} \Paren{\binom{t}{t/2-1}\Paren{e^{(t/2+1)\teps} + e^{(t/2-1)\teps}} + \sum_{i=0}^{t/2-2} \binom{t}{i}  \Paren{e^{(t-i)\teps} + e^{i\teps}} + 2\binom{t}{t/2-2} e^{(t/2)\teps}} \\
&= \frac{1}{(e^{\teps}+1)^t} \Paren{(e^{\teps}+1)^t + \Paren{2\binom{t}{t/2-2} - \binom{t}{t/2}} e^{(t/2)\teps}} \\
&\geq 1,
\end{align*}
where the last inequality is due to our choice $t = 12$.
\item \textbf{Case III: $x \in \Paren{1, e^{2\teps}}$.} Since the privacy loss values for both $\pld(\RR^{\otimes t}_{\teps}(0), \RR^{\otimes t}_{\teps}(1))$ and $\pld(\RR_{2\teps}(0), \RR_{2\teps}(1))$ are multiples\footnote{Note that for $\pld(\RR^{\otimes t}_{\teps}(0), \RR^{\otimes t}_{\teps}(1))$, this follows from the fact that $t$ is even.} of $2\teps$, the hockey stick curves $H_x$ of both pairs are linear for $x \in \left[1, e^{2\teps}\right]$. Thus, we have
\begin{align*}
&H_x(\RR^{\otimes t}_{\teps}(0) \parallel \RR^{\otimes t}_{\teps}(1)) \\ 
&= \frac{x - 1}{e^{2\teps} - 1} \cdot H_{e^{2\teps}}(\RR^{\otimes t}_{\teps}(0) \parallel \RR^{\otimes t}_{\teps}(1)) + \frac{e^{2\teps} - x}{e^{2\teps} - 1} \cdot H_1(\RR^{\otimes t}_{\teps}(0) \parallel \RR^{\otimes t}_{\teps}(1)) \\
&\geq \frac{x - 1}{e^{2\teps} - 1} \cdot H_{e^{2\teps}}(\RR_{2\teps}(0) \parallel \RR_{2\teps}(1)) + \frac{e^{2\teps} - x}{e^{2\teps} - 1} \cdot H_1(\RR_{2\teps}(0) \parallel \RR_{2\teps}(1)) \\
&= H_{x}(\RR_{2\teps}(0) \parallel \RR_{2\teps}(1)),
\end{align*}
where the inequality follows from Cases I and II above. \qedhere
\end{itemize}
\end{proof}

Using the above lemma together with our setting of parameters, it is simple to derive the following, which will be more convenient in a subsequent usage.

\begin{lemma} \label{lem:each-rr-samplingbudget-dp}
For every $i \in [n]$, $\RR_{\eps_i}^{\otimes b^*_i}$ is $\left(2\lceil \log(1/\delta) \rceil \cdot (9\eps + 2), 32\lceil \log(1/\delta) \rceil \cdot \delta\right)$-DP.
\end{lemma}

\begin{proof}
Consider three cases:
\begin{itemize}
\item \textbf{Case I: $i = n$.} This case is obvious since $b^*_i = 0$. 
\item \textbf{Case II: $i = 1$.} In this case, $\ell^*_i = k$ and $\eps_i \leq \frac{\eps}{k}$. Thus, by our setting of $b^*_i = 2\lceil \log(1/\delta) \rceil \cdot \ell^*_i$ and basic composition, $\RR^{\otimes b_i}_{\eps_i}$ is $\Paren{2\lceil \log(1/\delta) \rceil \cdot \eps}$-DP as desired.
\item \textbf{Case III: $1 < i < n$.} From the setting of $\ell^*_i$ in \eqref{eq:additive-budget-values}, $\RR_{\eps_{i}}^{\otimes \frac{\ell^*_i -1}{8}}$ is $(\eps, 2\delta)$-DP. Since $\eps_i \leq \eps + 2$, by basic composition, $\RR_{\eps_{i}}^{\otimes \ell^*_i}$ is $(9\eps + 2, 16\delta)$-DP. Thus, by our setting of $b^*_i = 2\lceil \log(1/\delta) \rceil \cdot \ell^*_i$ and basic composition, $\RR^{\otimes b_i}_{\eps_i}$ is $\left(2\lceil \log(1/\delta) \rceil \cdot (9\eps + 2), 32\lceil \log(1/\delta) \rceil \cdot \delta\right)$-DP as desired. \qedhere 
\end{itemize}
\end{proof}

\subsubsection{Formulation as Post-Processing of RR}

As alluded to earlier, we will prove \Cref{thm:main-ub-additive} by showing that, with high probability, the filter's output is a post-processing of randomized response answers. To do this, recall from \Cref{sec:preliminaries} that it suffices to consider two neighbouring datasets $D_0, D_1$. Our post-processing procedure is described in full in \Cref{alg:post-process-rr}. It takes as input the randomized response outputs on a bit $h \in \{0, 1\}$, and its output (approximately) matches the output of the additive filter. To formalize this, let $\ViewAddFilter(D)$ denote the output distribution of the Additive Discretized Privacy Loss Filter (\Cref{alg:additive-profile-filter}) on input dataset $D$ with analyst $\cA$ (where other parameters are omitted from the notation for readability).


\begin{algorithm}[tb]
    \caption{Post-processing Algorithm}
    \label{alg:post-process-rr}
    \begin{algorithmic}
        \REQUIRE Discretization points $\beps \in \cE_n$, budgets $\bell \in \R^n_{\geq 0}$, analyst $\cA$, query capacity $k$, neighbouring datasets $D_0, D_1$, independent samples $y_i^j \sim \RR_{\eps_i}(h)$ (where $h \in \{0, 1\}$)
        \STATE $\bj \gets \bzero$ \algcomment{$n$-dimensional all-zero vector; indices for current RR samples}
        \FOR{$t \gets 1, \dots, k$}
            \STATE $\cA(Y_1, \dots, Y_{t - 1})$ gives mechanism $M_t$ that has PLD $L_t \in \cLsym$
            \STATE $\bp^t \gets $ Discretized privacy loss profile of $L_t$ w.r.t. $\beps$
            \STATE $\bell \gets \bell - \bp^t$
            \STATE $\kappa_t \gets$ Markov kernel with $\kappa_t\Paren{\RR_{\bp^t, \beps}(h')} = M_t(D_{h'})$ for $h' \in \{0,1\}$ \algcomment{exists due to \Cref{thm:pessimistic}}
            \IF{$\bell \geq \bzero$}
            \STATE $i_t \gets$ sample according to $\bp^t$
            \STATE $j_{i_t} \gets j_{i_t} + 1$
            \IF{$j_{i_t} > b_{i_t}$}
            \RETURN FAIL \algcomment{Exceed budget; fail and terminate the algorithm}
            \ENDIF
            \STATE Sample $Y_t \sim \kappa_t\Paren{i_t, y^{j_{i_t}}}$
            \ELSE
            \STATE $Y_t \gets \bot$
            \ENDIF
        \ENDFOR
        \RETURN $(Y_1, \ldots, Y_k)$
    \end{algorithmic}
\end{algorithm}

Our post-processing procedure (\Cref{alg:post-process-rr}) has the parameter $\bb = (b_1, \dots, b_n)$ representing the ``sample budgets'' (i.e. how many randomized response samples are allowed to be used). To analyze the algorithm, first note that when there are no sample budgets (i.e. $b_1 = \cdots = b_n = \infty$), the output is exactly the same as the Additive Discretized Privacy Loss Filter:

\begin{observation}
    \label{obs:filter-as-postproc-rr}
    When we set $b_i = \infty$ for all $i \in [n]$, the output of \Cref{alg:post-process-rr} has exactly the same distribution as $\ViewAddFilter(D_h)$.
\end{observation}

\begin{proof}
Since $b_i = \infty$ for all $i \in [n]$, we never output FAIL. Thus, for each $t \in [k]$, $(i_t, y^{j_{i_t}})$ has exactly the same distribution as $\RR_{\bp^t, \beps}(h)$. From the definition of $\kappa_t$, $Y_t \sim \kappa_t(i_t, y^{j_{i_t}})$ has exactly the same distribution as $M_t(D_h)$. Thus, the entire process is exactly the same as $\ViewAddFilter(D_h)$.
\end{proof}

We next show that, even when we set the budget as in \eqref{eq:additive-budget-values}, the algorithm barely fails:

\begin{lemma} \label{lem:post-rr-truncation-failure-prob}
When we set $\ell_i = \ell^*_i$ and $b_i = b^*_i$ for all $i \in [n]$ (where $\bell^*, \bb^*$ are as specified in \eqref{eq:additive-budget-values} and \eqref{eq:sample-budget-setting}, respectively), \Cref{alg:post-process-rr} outputs FAIL with probability at most $4n \cdot \delta$.
\end{lemma}

\begin{proof}
The algorithm outputs FAIL only when $j_i > b_i$ for some $i \in [n]$. Thus, by the union bound, it suffices to show that $\Pr[j_i \geq b_i] \leq 4\delta$ for every $i \in [n]$. To do this, we consider two cases:
\begin{itemize}
\item \textbf{Case I: $i < n$.} Let $\cD_t$ denote the distribution $\Ber(p^t_i)$ and $W_t$ denote $\bone[i_t = i]$. For notational convenience, we also assume that, if the algorithm outputs FAIL in the $t_0$-th iteration, then we simply set $p^t_i = 0, \cD_t = \Ber(0)$ for $t \in \{t_0 + 1, \dots, k\}$.

Under these notations, we have $W_t \sim \cD_t$ and $W = W_1 + \cdots + W_k$ is exactly $j_i$. Finally, due to the constraint of the filter, we have $\E[W_1] + \cdots + \E[W_k] = p^1_i + \cdots + p^k_i \leq \ell^*_i$. Thus, applying \Cref{thm:azuma} with $\zeta = 6\lceil \log(1/\delta) \rceil - 1$ and observing that $\frac{\zeta^2}{2 + \zeta} \geq \frac{\zeta + 1}{6} \geq \lceil \log(1/\delta) \rceil$, we have
\begin{align*}
\Pr\left[j_i > b^*_i\right] = \Pr[W > (1 + \zeta) \cdot \ell^*_i] \leq \exp(-\log(1/\delta)) = \delta.
\end{align*}
\item \textbf{Case II: $i = n$.} Using the same notation as before and recalling that $b_i^* = 0$ and $\ell^*_i = 4\delta$ in this case, we have
\begin{align*}
\Pr\left[j_i > b^*_i\right] = \Pr[j_i \geq 1] = \Pr[W \geq 1] & \leq \ell^*_i = 4\delta,
\end{align*}
where the inequality is due to \Cref{lem:adaptive-markov}. \qedhere
\end{itemize}
\end{proof}

\subsubsection{Putting Things Together: Proof of \Cref{thm:main-ub-additive}}

Combining all the pieces together, we can now prove \Cref{thm:main-ub-additive}.

\begin{proof}[Proof of \Cref{thm:main-ub-additive}]
Let $\eps' = n \cdot 2\lceil \log(1/\delta) \rceil \cdot (9\eps + 2)$ and $\delta' = n \cdot 32\lceil \log(1/\delta) \rceil \cdot \delta + 4n \delta$. We will show that the $\pld\Paren{\AddFilter_{\beps, \bell^*}^k} \preceq R_{\eps',\delta'}$.

Let $D_0, D_1$ be any neighbouring input datasets. From \Cref{obs:filter-as-postproc-rr}, $\ViewAddFilter(D_h)$ is exactly the distribution of the output of \Cref{alg:post-process-rr} with $b_i = \infty$. Consider also the output distribution of \Cref{alg:post-process-rr} with $b_i = b^*_i$ as defined in \eqref{eq:sample-budget-setting}; let us call the resulting distribution $P'(D_h)$. 

By \Cref{lem:post-rr-truncation-failure-prob}, $(P, Q) = (\ViewAddFilter(D_0), \ViewAddFilter(D_1))$ and $(P', Q') = (P'(D_0), P'(D_1))$ satisfies the conditions in \Cref{lem:triangle} with $\omega^* = \text{FAIL}$ and $\gamma = 4n\delta$. Thus, \Cref{lem:triangle} implies that
\begin{align*}
H_{e^{\eps'}}(\ViewAddFilter(D_0) \parallel \ViewAddFilter(D_1)) \leq H_{e^{\eps'}}(P'(D_0) \parallel P'(D_1)) + 4n \delta.
\end{align*}
To bound $H_{e^{\eps'}}(P'(D_0) \parallel P'(D_1))$, notice that, since we set $b_i = b^*_i$ for all $i \in [n]$, $P'(D_h)$ is simply a post-processing of $\RR_{\eps_1}^{\otimes b^*_1}(h) \otimes \cdots \otimes \RR_{\eps_n}^{\otimes b^*_n}(h)$. Thus, we have
\begin{align*}
H_{e^{\eps'}}(P'(D_0) \parallel P'(D_1)) &\leq H_{e^{\eps'}}\left(\RR_{\eps_1}^{\otimes b^*_1}(0) \otimes \cdots \otimes \RR_{\eps_n}^{\otimes b^*_n}(0) ~\middle\|~ \RR_{\eps_1}^{\otimes b^*_1}(1) \otimes \cdots \otimes \RR_{\eps_n}^{\otimes b^*_n}(1)\right).
\end{align*}
Finally, from \Cref{lem:each-rr-samplingbudget-dp} and basic composition, we can conclude that $\RR_{\eps_1}^{\otimes b^*_1} \otimes \cdots \otimes \RR_{\eps_n}^{\otimes b^*_n}$ is $\Paren{\eps', n \cdot 32\lceil \log(1/\delta) \rceil \cdot \delta}$-DP. In other words,
\begin{align*}
H_{e^{\eps'}}\left(\RR_{\eps_1}^{\otimes b^*_1}(0) \otimes \cdots \otimes \RR_{\eps_n}^{\otimes b^*_n}(0) ~\middle\|~ \RR_{\eps_1}^{\otimes b^*_1}(1) \otimes \cdots \otimes \RR_{\eps_n}^{\otimes b^*_n}(1)\right) \leq n \cdot 32\lceil \log(1/\delta) \rceil \cdot \delta.
\end{align*}
Combining the three inequalities, we get
\begin{align*}
H_{e^{\eps'}}(\ViewAddFilter(D_0) \parallel \ViewAddFilter(D_1)) \leq n \cdot 32\lceil \log(1/\delta) \rceil \cdot \delta + 4n \delta  = \delta'.
\end{align*}
Thus, $\pld\Paren{\AddFilter_{\beps, \bell^*}^k} \preceq R_{\eps',\delta'}$ as claimed.
\end{proof}

%% file: sections/total_ordering_proof.tex
\section{Privacy Gap Amplification Without Total Blackwell Ordering}
\label{sec:total_ordering_proof}

In this section, we give a proof of the key technical \Cref{lemma:commutativity_implies_total_ordering} that goes through the topology of hockey-stick curves. That is, we are given non-degenerate PLDs $L_1$ and $L_2$ that are not ordered. Our goal is to argue that, for one of these two PLDs $L_i$, $L_i \oplus \sup\{L_1, L_2\} \succ \sup\{L_i \oplus L_1, L_i \oplus L_2\}$.

Our first step is to show that gaps between PLDs can be amplified in some sense by composition. Recall that the support of a distribution $A$ over $\overline{\bR}$, is $\supp(A) := \{t \in \mathbb{R} : \forall \text{open } U \ni t, A(U) > 0\}$, i.e. the finite points with locally positive probability. It turns out that any PLD whose support can be shifted to hit gaps between other PLDs causes the gaps to align when composed.

\begin{proposition}[PLD Gap Amplification]
    \label{prop:pld_gap_amplification}
    Let $A, L, L_1, L_2$ be PLDs so that $L_i \prec L$
    with $h_{L_i}(e^{\varepsilon_i + u}) < h_L(e^{\varepsilon_i + u})$ for some $\varepsilon_1, \varepsilon_2 \in \supp(A)$ and some $u \in \mathbb{R}$. Then
    \begin{align*}
        A \oplus L \succ \sup\{A \oplus L_1, A \oplus L_2\}.
    \end{align*}
\end{proposition}

\begin{proof}
    First, it is clear that $A \oplus L \succeq \sup\{A \oplus L_1, A \oplus L_2\}$, so we just need to find a single gap.

    To that end, by continuity, we can find a constant $\rho > 0$ s.t $\inf_{\varepsilon \in (\varepsilon_i + u \pm \rho)} h_L(e^\varepsilon) - h_{L_i}(e^\varepsilon) > 0$. Now, set $\varepsilon^\star := \varepsilon_1 + \varepsilon_2 + u$ and notice that
    \begin{align*}
        h_{A \oplus L}(e^{\varepsilon^\star}) - h_{A \oplus L_i}(e^{\varepsilon^\star})
            & = \mathbb{E}_{Z \sim A \oplus L}[(1 - e^{\varepsilon^\star - Z})_+] - \mathbb{E}_{Z \sim A \oplus L_i}[(1 - e^{\varepsilon^\star - Z})_+] \\
            & = \mathbb{E}_{Z_1 \sim A}[\mathbb{E}_{Z_2 \sim L}[(1 - e^{(\varepsilon^\star - Z_1) - Z_2})_+] - \mathbb{E}_{Z_2 \sim L_i}[(1 - e^{(\varepsilon^\star - Z_1) - Z_2})_+]] \\
            & = \mathbb{E}_{Z \sim A}[h_L(e^{\varepsilon^\star - Z}) - h_{L_i}(e^{\varepsilon^\star - Z})] \\
            & \geq \mathbb{P}_{Z \sim A}(Z \in (\varepsilon_{2 - i} \pm \rho)) \cdot \inf_{\varepsilon \in (\varepsilon_{2 - i} \pm \rho)} h_L(e^{\varepsilon_1 + \varepsilon_2 + u - \varepsilon}) - h_{L_i}(e^{\varepsilon_1 + \varepsilon_2 + u - \varepsilon}) \\
            & = \mathbb{P}_{Z \sim A}(Z \in (\varepsilon_{2 - i} \pm \rho)) \cdot \inf_{\varepsilon \in (\varepsilon_i + u \pm \rho)} h_L(e^\varepsilon) - h_{L_i}(e^\varepsilon) \\
            & > 0.
    \end{align*}
    In particular,
    $h_{\sup\{A \oplus L_1, A \oplus L_2\}}(e^{\varepsilon^\star}) = \max\{h_{A \oplus L_1}(e^{\varepsilon^\star}), h_{A \oplus L_2}(e^{\varepsilon^\star})\}
        < h_{A \oplus L}(e^{\varepsilon^\star})$.
\end{proof}

In order to apply gap amplification to PLDs that are not totally-ordered i.e. $L_1, L_2 \prec \sup\{L_1, L_2\}$, we would like to reduce the problem to geometric and topological properties of convex curves over $\bR^\times := (0, \infty)$. To that end, we also notice that the support of a PLD is determined by the topology of its hockey-stick curve. This result follows immediately from \Cref{prop:pld_to_hs}

\begin{proposition}
    \label{prop:support_conversion}
    For a PLD $L$ with hockey-stick curve $h_L$, denote by $\supp(h_L) := \{x \in \bR^\times : \forall \text{open }U \ni x, h_L|_U \text{ is not affine}\}$ the points at which $h_L$ ``bends''. Then for any $e^\varepsilon \in \supp(h_L)$, we have $\varepsilon \in \supp(L)$.
\end{proposition}

Our problem is now purely geometric. Given decreasing, convex, unordered curves $h_1, h_2 : \bR^\times \to [0, 1]$ that meet in the limit $x \to 0$ (recall \Cref{prop:hs_characterization}), we show that one of them can be scaled along the x-axis so that its elbows hit a region where $h_1 < h_2$ and another region where $h_1 > h_2$. Our proof strategy is to argue that the only way this can fail to occur is when $h_1$ and $h_2$ interleave one another geometrically as in~\Cref{fig:unordered_gap_amplification}. This forces the curves $h_1$ and $h_2$ to be unbounded, contradicting the fact that hockey-stick curves are bounded to $[0, 1]$.

\begin{figure}[htbp]
    \centering
    \includegraphics[width=0.75\linewidth]{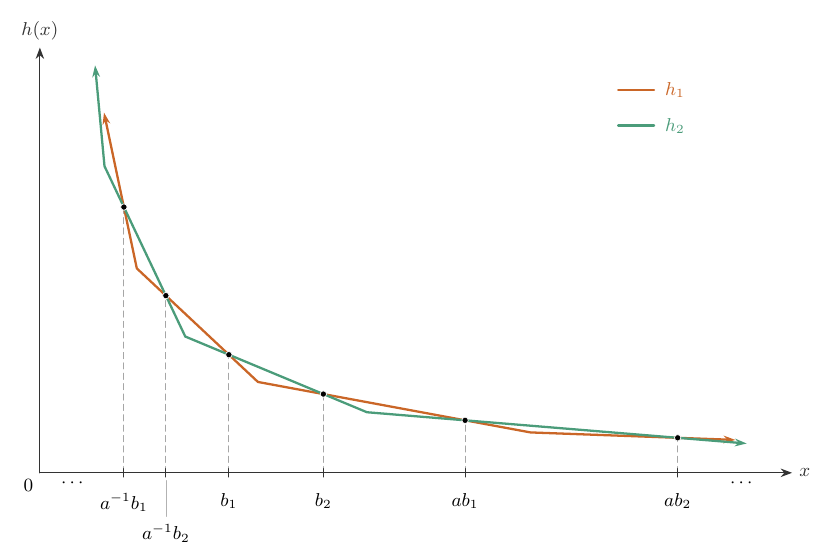}
    \caption{Unbounded curves $h_1$ and $h_2$ that interleave so that neither can be scaled horizontally in such a way that its rescaled elbows intersect the regions $\{h_1 < h_2\}$ and $\{h_1 > h_2\}$ simultaneously.}
    \label{fig:unordered_gap_amplification}
\end{figure}

We proceed by analysis of these regions' topology within the ambient space $\bR^\times$. Given $S \subseteq \bR^\times$, let $\partial S$ denote its boundary, i.e. those points whose neighbourhoods meet both $S$ and $S^C := \bR^\times \setminus S$, let $\mathrm{int}(S) := S \setminus \partial S$ denote its interior, and let $\overline{S} := S \cup \partial S$ denote its closure. Now, let
\begin{align*}
    X(h_1, h_2) := \{x \in \bR^\times : h_1(x) < h_2(x)\},
\end{align*}
let $X(h_2, h_1)$ be defined similarly, and let $E(h_1, h_2) := (X(h_1, h_2) \cup X(h_2, h_1))^C$ be the points where $h_1$ and $h_2$ agree. Note that, convex curves over $\bR^\times$ are also continuous, so $X(h_1, h_2)$ and $X(h_2, h_1)$ are both open and $E(h_1, h_2)$ is closed. Finally, let
\begin{align*}
    C(h_1, h_2) := \{x \in \bR^\times : (x r^{-1}, x) \subseteq X(h_1, h_2), (x, x r) \subseteq X(h_2, h_1) \text{ or vice-versa for some } r > 1\}
\end{align*}
denote the \textit{simple} crossing points of $h_1$ and $h_2$. We would like to show that either $\supp(h_1)$ or $\supp(h_2)$ ``links'' $X(h_1, h_2)$ with $X(h_2, h_1)$ in the following sense.

\begin{definition}
    For subsets $S, T_1, T_2$ of $\bR^\times$, we say that $S$ (multiplicatively) links $T_1$ and $T_2$ if there is $c > 0$ so that $c \cdot S \cap T_1 \neq \varnothing$ and $c \cdot S \cap T_2 \neq \varnothing$. Furthermore, for $S, T \subseteq \bR^\times$, we say that $S$ divides $T$, written $S \mid T$, if $S$ does not link $T$ and $T^C$.
\end{definition}

In some cases, it will be easier to directly link a set with the boundary of another set. The following lemma shows that this will suffice for our purposes.

\begin{lemma}
    \label{lemma:closure_linkage}
    Let $S, U_1, U_2 \subseteq \bR^\times$ with $U_1$ and $U_2$ open. The following are equivalent.
    \begin{enumerate}[(i)]
        \item $S$ links $U_1$ and $U_2$
        \item $S$ links $U_1$ and $\overline{U_2}$
        \item $S$ links $\overline{U_1}$ and $U_2$
    \end{enumerate}
\end{lemma}

\begin{proof}
    Clearly, (i) implies (ii) and (iii). Moreover, (ii) and (iii) are symmetric, so we just need to show (ii) $\implies$ (i). To that end, assume that there is $x_1, x_2 \in S$ and $c > 0$ such that $cx_1 \in U_1$ and $cx_2 \in \overline{U_2}$. Now, if $cx_2 \in U_2$, we are done, so assume $cx_2 \in \partial U_2$. That is, any arbitrarily small neighbourhood of $cx_2$ meets $U_2$, so, since $U_1 \ni cx_1$ is open, we can just choose $c' > 0$ sufficiently close to $c$ so that $c'x_1$ remains in $U_1$ and $c'x_2 \in U_2$.
\end{proof}

Critically, non-trivial closed and divisible subsets of $\bR^\times$ have a highly regular multiplicative structure enclosed by their divisors.

\begin{proposition}
    \label{prop:dense_subgroup}
    Let $S \subseteq \bR^\times$ contain more than one point and let $\varnothing \subsetneq T \subsetneq \bR^\times$ be closed. Then, if $S \mid T$, there is $a > 1$ and $B \subseteq [1, a)$ such that $T = a^\bZ B = \{a^k b : k \in \bZ, b \in B\}$ and such that $a \leq s'/s$ for any $s < s' \in S$.
\end{proposition}

\begin{proof}
    We first claim that $(s'/s)^\bZ T \subseteq T$ for any $s, s' \in S$. Indeed, let $t \in T$ and $s, s' \in S$. Then $t/s \cdot S \ni t/s \cdot s = t$ meets $T$, so, since $S \mid T$, we must have $T \supseteq t/s \cdot S \ni t/s \cdot s' = (s'/s) \cdot t$. By swapping $s$ and $s'$, we get that $(s'/s)^{-1} \cdot t \in T$ as well. Repeating inductively, we get that $(s'/s)^k \cdot t \in T$ for arbitrary $k \in \bZ$ as desired.

    Now, let $a := \inf\{s'/s : s < s' \in S\}$. Then there is a sequence $a_n := s'_n/s_n \to a$ for some $s_n < s'_n \in S$. By the preceding claim, we have $T = \overline{T} \supseteq \overline{\bigcup_{n \in \bN} a_n^\bZ \ T}$. We claim that $a > 1$. Indeed, if $a = 1$, then $a_n^\bZ T$ forms a cover of $\bR^\times$ with maximum multiplicative distance $a_n \to 1$, so $\bigcup_{n = 1}^\infty a_n^\bZ T$ must be dense in $\bR^\times$, which forces $T = \bR^\times$, whereas we assumed otherwise. Therefore $a > 1$ and
    \begin{align*}
        T \supseteq \overline{\bigcup_{n \in \bN} a_n^\bZ T}
          = \overline{\bigcup_{k \in \bZ} \bigcup_{t \in T} \{a_n^k t : n \in \bN\}}
          \supseteq \bigcup_{k \in \bZ} \bigcup_{t \in T}\overline{\{a_n^k t : n \in \bN\}}
          \supseteq \bigcup_{k \in \bZ} \bigcup_{t \in T} \{a^k t\}
          = a^\bZ T.
    \end{align*}
    Finally, let $B := T \cap [1, a)$. For any $t \in T$, we can choose $k \in \bZ$ so that $a^k t \in [1, a)$ and, furthermore, $a^k t \in a^\bZ T \subseteq T$, so in fact $b := a^k t \in B$ and thus $t = a^{-k} a^k t = a^{-k} b \in a^\bZ B$. That is, $T \subseteq a^\bZ B \subseteq a^\bZ T \subseteq T$.
\end{proof}

We also require a slight variation on the usual characterization of convexity: a concave curve that passes above and then below a convex curve must remain below it.

\begin{lemma}
    \label{lemma:convex_vs_affine}
    Let $f, g : (a, b) \to \bR$ so that $f$ is convex and $g$ is concave. Then, if $f(x_1) \geq g(x_1)$ and $f(x_2) \leq g(x_2)$ for some $x_1 < x_2 \in (a, b)$, then we have $f(x) \geq g(x)$ for all $x \in (a, x_1)$.
\end{lemma}

\begin{proof}
    Suppose we could find $x \in (a, x_1)$ with $f(x) < g(x)$. Then, by convexity of $f$ and concavity of $g$ over $[x, x_2]$, we have
    \begin{align*}
        f(x_1)
            \leq \frac{x_2 - x_1}{x_2 - x}\underbrace{f(x)}_{< g(x)} + \frac{x_1 - x}{x_2 - x} \underbrace{f(x_2)}_{\leq g(x_2)}
            < \frac{x_2 - x_1}{x_2 - x} g(x) + \frac{x_1 - x}{x_2 - x} g(x_2)
            \leq g(x_1),
    \end{align*}
    which is a contradiction.
\end{proof}

Now, our first important observation is that $h_i$ is supported at the boundary of its region of dominance, excluding simple crossings. This is essentially because convex and affine curves can only cross at points.

\begin{proposition}
    \label{prop:boundary_support}
    Let $h_1, h_2 : \bR^\times \to \bR$ be convex. Then
    \begin{enumerate}[(i)]
        \item $\partial X(h_1, h_2) \setminus C(h_1, h_2) \subseteq \supp(h_2)$; and
        \item $\partial X(h_2, h_1) \setminus C(h_1, h_2) \subseteq \supp(h_1)$.
    \end{enumerate}
\end{proposition}

\begin{proof}
    We just prove (i) as (ii) follows by symmetry. In this case it is equivalent to show that $\partial X(h_1, h_2) \setminus \supp(h_2) \subseteq C(h_1, h_2)$.
    
    To that end, consider any $x \in \partial X(h_1, h_2) \setminus \supp(h_2)$. Then there must be some interval $(x r^{-1}, x r)$ around $x$ on which $h_2$ is affine. Moreover, since $x \in \partial{X(h_1, h_2)}$, there must be some $x' \in (x r^{-1}, x r)$ for which $h_1(x') < h_2(x')$. On the other hand, both $X(h_1, h_2)$ and $X(h_2, h_1)$ are disjoint open subsets of a connected space, so neither meets $\partial{X(h_1, h_2)}$. In particular, $h_1(x) = h_2(x)$ and $x \neq x'$. Without loss of generality, $x r^{-1} < x' < x$.

    Now, we claim that $(x', x) \subseteq X(h_1, h_2)$ and $(x, x r) \subseteq X(h_2, h_1)$. Indeed, if there were $x'' \in (x', x)$ with $h_1(x'') \geq h_2(x'')$, then by applying Lemma~\ref{lemma:convex_vs_affine} with $x_1 = x''$ and $x_2 = x$, we get $h_1(x') \geq h_2(x')$, which is a contradiction. Likewise, if we could find $x'' \in (x, x r)$ with $h_1(x'') \leq h_2(x'')$, then applying Lemma~\ref{lemma:convex_vs_affine} to $x_1 = x$ and $x_2 = x''$ yields once again $h_1(x') \geq h_2(x')$, which is impossible. Overall, shrinking $r$ a little as needed, we have that $x \in C(h_1, h_2)$.
\end{proof}

The other key observation is that the support of $h_i$ meets every maximal interval in which it is dominated. Recall that every open set can be decomposed into connected components (disjoint open intervals in this case). For convenience, we will write $h(0) := \lim_{x \to 0} h(x)$ when this limit exists.

\begin{proposition}
    \label{prop:dominated_support}
    Let $h_1, h_2 : \bR^\times \to \bR$ be convex, bounded, and decreasing such that $h_1(0) = h_2(0)$. Then for any connected component $I \neq \bR^\times$ of $X(h_1, h_2)$, we have $I \cap \supp(h_1) \neq \varnothing$. Similarly for each connected component $I \neq \bR^\times$ of $X(h_2, h_1)$, we have $I \cap \supp(h_2) \neq \varnothing$.
\end{proposition}

\begin{proof}
    The two claims are symmetric so we just prove the first one. Let $I = (a, b) \neq \bR^\times$ be a connected component of $X(h_1, h_2)$ and assume that $I$ does not meet $\supp(h_1)$, i.e. $h_1$ is affine over the interval $\overline{I}$.
    
    First, we claim that $b < \infty$. Indeed, suppose that we had $b = \infty$. Then $a > 0$ as $I \neq \bR^\times$. Now, since $h_1$ and $h_2$ are convex and, in particular continuous over $\bR^\times$, if it were the case that $h_1(a) < h_2(a)$, then we could extend the interval $(a, \infty)$ to the left to obtain a larger interval in $X(h_1, h_2)$, whereas we assumed that $I$ was a connected component. Thus $h_2(a) \leq h_1(a)$. But $h_1$ is affine and bounded on an unbounded interval $[a, \infty)$, so the only possibility is that $h_1$ is constant on $[a, \infty)$. In particular since $h_2$ is decreasing, $h_2$ must be dominated by $h_1$ over all of $I = (a, \infty) \subseteq X(h_1, h_2)$, which is a contradiction.

    Next, we claim that $h_2$ is dominated by $h_1$ at the endpoints $\{a, b\}$. Indeed, if $a = 0$, then $h_1(a) = h_2(a)$ by assumption and, if $a > 0$, then by the same interval extension argument as before, we have $h_2(a) \leq h_1(a)$. Likewise, since $b < \infty$ and since the interval $I$ is maximal, we must have $h_2(b) \leq h_1(b)$ as well. Overall, $h_2$ is dominated at the endpoints of $I$ by an affine curve $h_1$, so by convexity $h_2$ must be dominated by $h_1$ over all of $I \subseteq X(h_1, h_2)$, which is impossible.
\end{proof}

We now can prove our main result of the section, which implies \Cref{lemma:commutativity_implies_total_ordering}. We note the small caveat that the queries must be non-degenerate, namely their PLDs must be supported at at least two finite points. This is not a serious practical limitation.

\begin{theorem}
    \label{thm:hockey_link}
    Let $h_1, h_2 : \bR^\times \to \bR$ be non-degenerate hockey-stick curves that are not totally-ordered. Then one of $\supp(h_1), \supp(h_2)$ links $X(h_1, h_2)$ and $X(h_2, h_1)$.
\end{theorem}

\begin{proof}
    We begin with the easy case: Assume that one of $\partial X(h_1, h_2), \partial X(h_2, h_1)$ is not fully contained in $C(h_1, h_2)$. Without loss of generality, suppose $\partial X(h_1, h_2) \nsubseteq C(h_1, h_2)$. By Proposition~\ref{prop:boundary_support}, $\supp(h_2)$ meets $\partial X(h_1, h_2) \subseteq \overline{X(h_1, h_2)}$. Furthermore, since $h_1$ and $h_2$ are not totally-ordered, $X(h_2, h_1) \neq \varnothing, \bR^\times$ and thus $\supp(h_2)$ meets $X(h_2, h_1)$ as well by Proposition~\ref{prop:dominated_support}. By Lemma~\ref{lemma:closure_linkage}, $\supp(h_2)$ must link $X(h_1, h_2)$ and $X(h_2, h_1)$.

    On the other hand, assume that $\partial X(h_1, h_2)$ and $\partial X(h_2, h_1)$ are both contained in $C(h_1, h_2)$ and assume toward a contradiction that neither $\supp(h_1)$ nor $\supp(h_2)$ links $X(h_1, h_2)$ and $X(h_2, h_1)$. We will derive a contradiction in two steps: First, we establish the structure of $E(h_1, h_2)$, $\supp(h_1)$, $\supp(h_2)$ and second, we derive a recurrence for the values of $h_1$ and $h_2$ at the points where they meet. By analyzing the characteristic polynomial of this recurrence, we can show that $h_1$ and $h_2$ must be either unbounded or constant, contradicting our assumptions.

    Now, our first major goal is to argue that $E(h_1, h_2)$ can be expressed of the form $a^\bZ \{b_1, b_2\}$ for some $1 \leq b_1 < b_2 < a$ and, exchanging the roles of $h_1$ and $h_2$ as needed, that $\supp(h_1) = a^\bZ s_1$ and $\supp(h_2) = a^\bZ s_2$ for some $s_1 \in (b_1, b_2)$ and $s_2 \in (b_2, ab_1)$.

    To this end, we first argue that $E(h_1, h_2)$ lives in both $\overline{X(h_1, h_2)}$ and $\overline{X(h_2, h_1)}$. Indeed, suppose we could find some $x \in \partial \textrm{int}(E(h_1, h_2))$. In this case,
    \begin{align*}
        x & \in \partial E(h_1, h_2) \\
            & = \partial (X(h_1, h_2) \cup X(h_2, h_1))^C \\
            & = \partial (X(h_1, h_2) \cup X(h_2, h_1)) \\
            & \subseteq \partial X(h_1, h_2) \cup \partial X(h_2, h_1) \\
            & \subseteq C(h_1, h_2),
    \end{align*}
    so, by construction of $C(h_1, h_2)$, there is some punctured neighbourhood $(x r^{-1}, x r) \setminus \{x\}$ of $x$ lying outside of $E(h_1, h_2) \supseteq \textrm{int}(E(h_1, h_2))$, which would contradict $x \in \partial \textrm{int}(E(h_1, h_2))$. Thus $\textrm{int}(E(h_1, h_2))$ has empty boundary, so it must be both closed and open in $\bR^\times$. But $h_1$ and $h_2$ cross, so we cannot have $\textrm{int}(E(h_1, h_2)) = \bR^\times$. By connectivity of $\bR^\times$, the only possibility is that $\textrm{int}(E(h_1, h_2)) = \varnothing$. In particular, by construction of the simple crossing points, we have
    \begin{align*}
        E(h_1, h_2)
            = \partial E(h_1, h_2)
            \subseteq C(h_1, h_2)
            \subseteq \partial X(h_1, h_2) \cap \partial X(h_2, h_1)
            \subseteq \overline{X(h_1, h_2)} \cap \overline{X(h_2, h_1)}.
    \end{align*}

    We can now deduce the structure of $E(h_1, h_2)$. First, notice that $\supp(h_i)$ divides $E(h_1, h_2)$. If not, then $\supp(h_i)$ must link $E(h_1, h_2)$ with $\bR^\times \setminus E(h_1, h_2) = X(h_1, h_2) \cup X(h_2, h_1)$, so it must either link $E(h_1, h_2) \subseteq \overline{X(h_2, h_1)}$ with $X(h_1, h_2)$ or it must link $E(h_1, h_2) \subseteq \overline{X(h_1, h_2)}$ with $X(h_2, h_1)$. In either case, by Lemma~\ref{lemma:closure_linkage}, $\supp(h_i)$ links $X(h_1, h_2)$ with $X(h_2, h_1)$, whereas we assumed otherwise. Therefore, $\supp(h_i)$ divides $E(h_1, h_2)$, so, by Proposition~\ref{prop:dense_subgroup}, there must be $a > 1$ and $B \subseteq [1, a)$ such that $E(h_1, h_2) = a^\bZ B$ and such that the elements of $\supp(h_1)$ and of $\supp(h_2)$ are multiplicatively separated by at least a factor of $a$. We claim that $B$ contains no more than two points. Indeed, suppose we could find $b_1 < b_2 < b_3 \in B$. By construction of $C(h_1, h_2) \ni b_1, b_2, b_3$, $(b_1, b_2) \cup (b_2, b_3) \cup (b_3, ab_1)$ contains at least three connected components of $X(h_1, h_2) \cup X(h_2, h_1)$, so in particular, it must contain at least two connected components of $X(h_1, h_2)$ or of $X(h_2, h_1)$. By Proposition~\ref{prop:dominated_support}, $(b_1, ab_1)$ contains at least two points from $\supp(h_1)$ or from $\supp(h_2)$, which is impossible because these points were multiplicatively separated by $a$. Thus $|B| \leq 2$. On the other hand, if $B = \{b\}$, then we have $E(h_1, h_2) = a^\bZ\{b\} = a^{2\bZ}\{b\} \cup a^{2\bZ + 1}\{b\} = (a^2)^\bZ\{b, ab\}$, so in all cases we can write $E(h_1, h_2) = a^\bZ\{b_1, b_2\}$ for some $1 \leq b_1 < b_2 < a$.

    Finally, we determine the structure of $\supp(h_1)$ and $\supp(h_2)$. For $k \in \bZ$, consider the alternating intervals $I_k := (a^k b_1, a^k b_2)$ and $J_k := (a^k b_2, a^{k+1} b_1)$. Then, for every $k \in \bZ$, $I_k$ does not meet $E(h_1, h_2)$, so, by the intermediate value theorem, either $I_k \subseteq X(h_1, h_2)$ or $I_k \subseteq X(h_2, h_1)$ and, similarly, either $J_k \subseteq X(h_1, h_2)$ or $J_k \subseteq X(h_2, h_1)$. Now, assume without loss of generality that $I_0 \subseteq X(h_1, h_2)$. By induction and by construction of the simple crossings $C(h_1, h_2) \supseteq E(h_1, h_2)$, this forces $I_k \subseteq X(h_1, h_2)$ and $J_k \subseteq X(h_2, h_1)$ for every $k \in \bZ$. By Proposition~\ref{prop:dominated_support}, we can find $s_1 \in (b_1, b_2) \cap \supp(h_1)$ and $s_2 \in (b_2, ab_1) \cap \supp(h_2)$. In particular, $b_1/s_1 \cdot \supp(h_1) \ni b_1$ meets $E(h_1, h_2) = a^{\bZ}\{b_1, b_2\} \ni b_1$, so, since $\supp(h_1) \mid E(h_1, h_2)$, we must have that $b_1/s_1 \cdot \supp(h_1) \subseteq a^{\bZ} \{b_1, b_2\}$, i.e. $\supp(h_1) \subseteq a^{\bZ} \{s_1, b_2/b_1 \cdot s_1\}$. We now eliminate $b_2/b_1 \cdot s_1$ from the residue set. To that end, suppose we could find $s'_1 := a^k b_2/b_1 \cdot s_1 \in \supp(h_1)$. Then $\frac{b_1^2}{s_1 b_2} \cdot s'_1 = a^k b_1 \in E(h_1, h_1)$, so, since $\supp(h_1) \mid E(h_1, h_2)$, we must have $\frac{b_1^2}{b_2} = \frac{b_1^2}{s_1 b_2} \cdot s_1 \in E(h_1, h_2) = a^\bZ \{b_1, b_2\}$ as well. But, $b_1/b_2 \in (a^{-1}, 1)$, so $\frac{b_1^2}{b_2} = a^k b_1$ is impossible, whereas $\frac{b_1^2}{b_2} = a^k b_2$ forces $b_1/b_2 = a^{-1/2}$. Moreover, $s_1 \in (b_1, b_2)$, so we have
    \begin{align*}
        s'_1
            = a^k b_2/b_1 \cdot s_1
            \in (a^k b_2/b_1 \cdot b_1, a^k b_2/b_1 \cdot b_2)
            = (a^k b_2, a^k b_1 \cdot (b_2/b_1)^2)
            = (a^k b_2, a^{k + 1} b_1)
            = J_k,
    \end{align*}
    which contradicts our assumption that $\supp(h_1)$ fails to link $X(h_1, h_2) \supseteq I_0$ with $X(h_2, h_1) \supseteq J_k$. Therefore $\supp(h_1) \subseteq a^\bZ s_1$ which, noting that $a^\bZ s_1 \cap I_k = a^k s_1$, forces $\supp(h_1) = a^\bZ s_1$ due to Proposition~\ref{prop:dominated_support}. By a symmetric argument, $\supp(h_2) = a^\bZ s_2$ as well.

    Now that we understand the structure of the crossings and of the support, we can calculate the trajectory of $h_1$ and $h_2$ over the crossing points. To this end, consider $y_k := h_1(a^k b_1) = h_2(a^k b_1)$ and $y'_k := h_1(a^k b_2) = h_2(a^k b_2)$. Now, for any $k \in \bZ$, the interval $[a^k b_1, a^{k+1} b_2] = \overline{I_k} \cup \overline{J_k} \cup \overline{I_{k + 1}}$ meets $\supp(h_2)$ at $a^k s_2$ and nowhere else. Thus $h_2$ is affine over $[a^k b_1, a^k s_2]$ and $[a^k s_2, a^{k+1}b_2]$ separately. By extrapolation in the former interval, we get
    \begin{align*}
        h_2(a^k s_2)
            = \frac{a^k b_2 - a^k s_2}{a^k b_2 - a^k b_1}h_1(a^k b_1) +
              \frac{a^k s_2 - a^k b_1}{a^k b_2 - a^k b_1}h_1(a^k b_2)
            = \frac{b_2 - s_2}{b_2 - b_1}y_k + \frac{s_2 - b_1}{b_2 - b_1}y'_k
    \end{align*}
    and, furthermore, extrapolation over the latter interval yields
    \begin{align*}
        y'_{k + 1}
            & = h_2(a^{k+1}b_2) \\
            & = \frac{a^{k+1}b_1 - a^{k+1}b_2}{a^{k+1}b_1 - a^k s_2}h_1(a^k s_2) +
                \frac{a^{k+1}b_2 - a^k s_2}{a^{k+1}b_1 - a^k s_2}h_2(a^{k+1}b_1) \\
            & = \frac{a(b_1 - b_2)}{ab_1 - s_2}\left(
                    \frac{b_2 - s_2}{b_2 - b_1}y_k + \frac{s_2 - b_1}{b_2 - b_1}y'_k
                \right) +
                \frac{ab_2 - s_2}{ab_1 - s_2}y_{k + 1} \\
            & = \underbrace{\frac{a(s_2 - b_2)}{ab_1 - s_2}}_{=: U'}y_k + \underbrace{-\frac{a(s_2 - b_1)}{ab_1 - s_2}}_{=: V'}y'_k + \underbrace{\frac{ab_2 - s_2}{ab_1 - s_2}}_{=: W'}y_{k + 1}.
    \end{align*}
    By an analogous argument, we have that
    \begin{align*}
    y_{k + 2}
        = \underbrace{\frac{a(s_1 - b_1)}{b_2 - s_1}}_{=: U}y'_k + \underbrace{-\frac{as_1 - b_2}{b_2 - s_1}}_{=: V}y_{k + 1} + \underbrace{\frac{ab_1 - s_1}{b_2 - s_1}}_{=: W}y'_{k + 1}
    \end{align*}
    as well, which together yields the system
    \begin{align*}
        \left\{
        \begin{aligned}
            y'_{k + 1} & = U' y_k + V' y'_k + W' y_{k + 1} \\
            y_{k + 2} & = U y'_k + V y_{k+1} + W y'_{k+1} \\
            y'_{k + 2} & = U' y_{k+1} + V' y'_{k+1} + W' y_{k+2} \\
            y_{k + 3} & = U y'_{k+1} + V y_{k+2} + W y'_{k+2}
        \end{aligned}
        \right.
    \end{align*}
    for $k \in \bZ$. By isolating $y'_{k+1}$ in terms of $y_k, y_{k+1}, y_{k+2}$ from the first two equations, substituting into the equation for $y'_{k+2}$, and finally substituting this into the equation for $y_{k + 3}$, we obtain the symmetrized recurrence
    \begin{align*}
        y_{k + 3} = U U' y_k + (U W' + U' W - V V') y_{k + 1} + (V + V' + W W') y_{k + 2}
    \end{align*}
    for $k \in \bZ$. Thus $(y_k)_{k \in \bZ}$ satisfies a linear homogeneous recurrence with characteristic polynomial
    \begin{align*}
        f(t) = t^3 - (V + V' + W W')t^2 - (U W' + U' W - V V') t - U U',
    \end{align*}
    which after applying some computational effort factors as
    \begin{align*}
        f(t) = (t - 1)(t - a)\left( t - a\left( \frac{s_1 - b_1}{b_2 - s_1} \right)\left( \frac{s_2 - b_2}{ab_1 - s_2} \right) \right).
    \end{align*}
    In particular, every root of $f$ lies in $\bR^\times$. On the other hand, a homogeneous linear recurrence over the integers admits a bounded non-constant solution if and only if its characteristic polynomial has a non-unit root lying on the complex unit circle (see e.g. Section 2.3 of \cite{elaydi2005introduction}). Therefore either $y_k = h_1(a^k b_1) = h_2(a^k b_1)$ is unbounded, which is impossible because hockey-stick curves are bounded, or it must be constant, which is also impossible because hockey-stick curves are monotone, so this would force $h_1$ and $h_2$ to be constant and thus totally-ordered.
\end{proof}

\Cref{lemma:commutativity_implies_total_ordering} now follows immediately by combining Propositions~\ref{prop:pld_gap_amplification} and~\ref{prop:support_conversion} with Theorem~\ref{thm:hockey_link} by taking $L := \sup\{L_1, L_2\}$.

%% file: sections/conclusion.tex
\section{Conclusion and Future Directions}

In this work, we study privacy filters, which allow flexible adaptive interaction with sensitive data. Prior work has developed valid filters based on inexact scalar composition forms of privacy accounting. We consider the natural filter, which purports to bring exact privacy accounting into the setting of full adaptivity.

To this end, we develop general-purpose theoretical tools for analyzing the privacy cost of adaptive interaction under arbitrary privacy constraints. We apply these tools to develop the residue filters, which unify known privacy filters in addition to the natural filter.

By applying our new analytical tools for adaptive interaction to the natural filter, we show that the natural filter does not generally ensure its claimed privacy guarantee. In particular, we show that a strong and unrealistic total-ordering condition is required for the natural filter to achieve its intended privacy cost.

We also present a new analysis of the natural filter via discrete approximation showing that, although the natural filter fails, its true privacy cost is at most logarithmically larger than its intended privacy cost, growing with the number of rounds of interaction. It is an interesting question for future work whether the privacy blowup incurred by the natural filter can be reduced from logarithmic down to a constant. Likewise, a complementary direction is to show a quantitative lower bound on this blowup matching the upper bound.

Finally, while we have shown that full adaptivity necessarily brings some privacy overhead, we show a path toward ensuring this overhead tends to zero asymptotically. In particular, we construct a new privacy filter, the GDP residue filter, which offers greater utility than the original GDP filter. For pure DP queries, we show that the GDP residue filter permits asymptotically as many fully adaptive queries as in the non-adaptive setting. We propose that an important direction for future research is to construct new kinds of residue filters with asymptotically vanishing overhead for practical applications, including private machine learning.

%% file: sections/app_prelims.tex
\section{Proofs from Preliminaries}
\label{app:prelim_proofs}

\begin{proof}[Proof of \Cref{prop:pld_to_hs}]
    For any pair $(P, Q)$ defined on $\Omega$ with likelihood ratio $\ell$, we have that
    \begin{align*}
        H_x(P \parallel Q)
            & = \sup_{E \subseteq \Omega} P(E) - xQ(E) \\
            & = \sup_{E \subseteq \Omega} \int_E \frac{d(P - xQ)}{dP}(\omega) \, dP(\omega) \\
            & = \sup_{E \subseteq \Omega} \int_E 1 - x/\frac{dP}{dQ}(\omega) \, dP(\omega) \\
            & = \int_\Omega \left(1 - x/\frac{dP}{dQ}(\omega)\right)_+ \, dP(\omega) \\
            & = \int_\Omega (1 - xe^{-\log(\ell(\omega))})_+ \, dP(\omega) \\
            & = \E_{Z \sim \pld(P \parallel Q)}[(1 - xe^{-Z})_+].
    \end{align*}
\end{proof}

\begin{proof}[Proof of \Cref{prop:hockey_stick_composition}]
    Let $\ell$ denote the likelihood ratio for $M_1 \otimes M_2$ and let $\ell^2_{y_1}$ denote the likelihood ratio for $M_2(\cdot; y_1)$. By Bayes' rule, we have $\ell(y_1, y_2) = \ell^1(y_1) \cdot \ell^2_{y_1}(y_2)$, so we have
    \begin{align*}
        \pld(M_1 \otimes M_2) \equiv \log(\ell^1(Y_1)) + \log(\ell^2_{Y_1}(Y_2)), Y_1 \sim M_1(D_1), Y_2 \sim M_2(D_1; Y_1).
    \end{align*}
    By \Cref{prop:pld_to_hs} and the law of total expectation, it follows that
    \begin{align*}
        h_{M_1 \otimes M_2}(x)
            & = \E_{Z \sim \pld(M_1 \otimes M_2)}[
                        (1 - xe^{-Z})_+
                    ] \\
            & = \E_{Y_1 \sim M_1(D_1)}[\E_{Y_2 \sim M_2(D_1; Y_1)}[
                        (1 - xe^{-(\log(\ell^1(Y_1)) + \log(\ell^2_{Y_1}(Y_2)))})_+
                    ]] \\
            & = \E_{Y_1 \sim M_1(D_1)}[\E_{Y_2 \sim M_2(D_1; Y_1)}[
                        (1 - x/\ell^1(Y_1) \cdot e^{-\log(\ell^2_{Y_1}(Y_2))})_+
                    ]] \\
            & = \E_{Y_1 \sim M_1(D_1)}[\E_{Z_2 \sim \pld(M_2(\cdot; Y_1)}[
                        (1 - x/\ell^1(Y_1) \cdot e^{-Z_2})_+
                    ]] \\
            & = \E_{Y_1 \sim M_1(D_1)}[h_{M_2(\cdot; Y_1)}(x/\ell^1(Y_1))].
    \end{align*}
\end{proof}

The well-known Neyman--Pearson lemma states that the optimal testing rule thresholds the PLD. For convenience, given $z^* \in \overline{\bR}$, we call $\psi : \overline{\bR} \to [0, 1]$ a $z^*$-threshold when $\psi(z) = 0$ for $z < z^*$ and $\psi(z) = 1$ for $z > z^*$.

\begin{proposition}
    \label{prop:pld_to_tradeoff}
    Let $(P, Q)$ be a pair of distributions and let $F$ be the CDF of $\pld(P \parallel Q)$. Then, for every $\alpha \in [0, 1]$, setting $z_\alpha := \min\{z \in \overline{\bR} : F(z) \geq \alpha\}$, there is a $z_\alpha$-threshold $\psi_\alpha$ such that
    \begin{align*}
        \tau(\alpha) = \E_{Z \sim \pld(P \parallel Q)}[\psi_\alpha(Z) e^{-Z}]
    \end{align*}
    and such that $\E_{Z \sim \pld(P \parallel Q)}[\psi_\alpha(Z)] = 1 - \alpha$. Moreover, if $\alpha = F(z^*)$, then, for $Z \sim \pld(P \parallel Q)$, $\psi_\alpha(Z) = 1(Z > z^*)$ almost surely.
\end{proposition}

\begin{proof}[Proof of \Cref{prop:tradeoff_to_hs}]
    First, let $L$ denote the PLD of $(P, Q)$ and let $F$ be its CDF. For $\alpha \in [0, 1]$, let $\psi_\alpha$ denote the threshold as in \Cref{prop:pld_to_tradeoff}. For any $x \in \bR^\times$, \Cref{prop:pld_to_hs} yields
    \begin{align*}
        h(x)
            & = \E_{Z \sim L}[1(Z > \log{x}) \cdot (1 - x e^{-Z})] \\
            & = \sup_{\alpha \in [0, 1]} \E_{Z \sim L}[\psi_\alpha(Z) \cdot (1 - x e^{-Z})] \tag{$\alpha = F(\log{x})$} \\
            & = \sup_{\alpha \in [0, 1]} 1 - \alpha - x\tau(\alpha) \\
            & = \sup_{\beta \in [0, 1]} 1 - \tau^{-1}(\beta) - x \beta \tag{$\alpha = \tau^{-1}(\beta), \beta = \tau(\alpha)$} \\
            & = 1 + (\tau^{-1})^*(-x).
    \end{align*}
\end{proof}

%% file: sections/app_constrained_adversaries.tex
\section{Proofs for \Cref{sec:adversary_theory}}
\label{app:adversary_theory}

We prove \Cref{lemma:restricted_adversary} by passing to hockey-stick curves because the pointwise supremum is easier to analyze than the lower convex envelope. We will essentially show that a worst-case adversary first plays an optimal move with corresponding PLD $L$, observes the outcome $Y_1$, and adversarially chooses a followup in $\Adv(\Gamma|_L)$ to ensure that the overall process ``traces out'' the hockey-stick curve of $L \oplus \Adv(\Gamma|_L)$.

\begin{proof}[Proof of \Cref{lemma:restricted_adversary}]
    The case where $k = 1$ is immediate, so we assume $k > 1$. We first show that
    \begin{align*}
        \pld(\Gamma) \preceq \sup\{L \oplus \pld(\Gamma|_L) : L \in \Gamma(\lambda)\}.
    \end{align*}
    To that end, consider any $\Gamma$-adversary $M = M_1 \otimes M_2$ and let $L := \pld(M_1) \in \Gamma(\lambda)$. Then, clearly, for every fixed $y_1$, $M_2(\cdot; y_1)$ is a $\Gamma|_L$-adversary and, in particular, $\pld(M_2(\cdot; y_1)) \preceq \pld(\Gamma|_L)$. By \Cref{prop:composition_convolution}, $\pld(M) \preceq L \oplus \pld(\Gamma|_L)$ and thus the desired inequality follows by taking suprema.

    More surprising is the reverse inequality, which we prove constructively using hockey-stick curves. Let $L \in \Gamma(\lambda)$, $x \in \bR^\times$, and $\gamma > 0$. Choose any mechanism $M_1$ with PLD $L$ and let $\ell_1 := \frac{dM_1(D_1)}{dM_1(D_2)}$ denote its likelihood function. Moreover, for any $x' \in \R^\times$ we can choose by \Cref{prop:sup-conv}, a $\Gamma|_L$-adversary $M^{x'}$ such that $h_{M^{x'}}(x') \geq h_{M_L^\star}(x') - \gamma$ where $\pld(M_L^\star) = \pld(\Gamma|_L)$.
    Now consider the $\Gamma$-adversary $M := M_1 \otimes M_2$ where $M_2(D; y_1) := M^{x/\ell_1(y_1)}(D)$. By \Cref{prop:hockey_stick_composition}, we get
    \begin{align*}
        h_M(x)
            & = \E_{Y_1 \sim M_1(D_1)}[h_{M_{x/\ell_1(Y_1)}}(x/\ell_1(Y_1))] \\
            & \geq \E_{Y_1 \sim M_1(D_1)}[h_{M_L^\star}(x/\ell_1(Y_1)) - \gamma] \\
            & = h_{M_1 \otimes M_L^\star}(x) - \gamma.
    \end{align*}
    Taking suprema, we get
    $\sup_{M \in \Adv(\Gamma)} h_M \succeq h_{M_1 \otimes M_L^\star}$
    and thus $\pld(\Gamma) \succeq L \oplus \pld(\Gamma|_L)$ by \Cref{prop:sup-conv,prop:composition_convolution}, which completes the proof since $L$ was arbitrary.
\end{proof}

\begin{proof}[Proof of \Cref{lemma:infinite_rule_limit}]
    It is immediate that $\pld(\Gamma^k) \preceq \pld(\Gamma)$ for all $k \in \bZ^+$ because any mechanism $M \in \Adv(\Gamma^k)$ can be extended to a mechanism in $\Adv(\Gamma)$ with the same privacy loss by simply repeatedly issuing the $\Iddist$ query after the first $k$ rounds.

    As for the reverse inequality, we pass again to the hockey-stick curve. To this end, let $M \in \Adv(\Gamma)$, let $\ell$ denote the likelihood ratio for $M$, let $M^k$ denote the restriction of $M$ to the first $k$ outputs, and let $\ell^k$ denote the likelihood ratio of $M^k$. Denoting $y^k := (y_1, \dots, y_k)$ for $y = (y_1, y_2, \dots)$, it is clear that $\ell(y) = \lim_{k \to \infty} \ell^k(y^k)$. By \Cref{prop:pld_to_hs} and the dominated convergence theorem, we have
    \begin{align*}
        h_M(x)
            & = \E_{Y \sim M(D)}[(1 - x/\ell(Y))_+] \\
            & = \E_{Y \sim M(D)}[\lim_{k \to \infty}(1 - x/\ell^k(Y^k))_+] \\
            & = \lim_{k \to \infty}\E_{Y \sim M(D)}[(1 - x/\ell^k(Y^k))_+] \\
            & = \lim_{k \to \infty}\E_{Y \sim M^k(D)}[(1 - x/\ell^k(Y))_+] \\
            & = \lim_{k \to \infty} h_{M^k}(x) \\
            & \leq \sup\{h_{M^k}(x) : k \in \bZ^+\}.
    \end{align*}
    Since clearly $M^k \in \Adv(\Gamma^k)$, the desired inequality follows from \Cref{prop:sup-conv}.
\end{proof}

%% file: sections/app_gdp_residue.tex
\section{Analysis of GDP Residue Filter for $\eps$-DP Queries}
\label{app:gdp_residue}

For a PLD $L$, let $\mu(L)$ be its GDP parameter, i.e. the minimal $\nu \geq 0$ for which $L \preceq G_\nu$. Throughout we will denote by $\phi$, $\Phi$, and $\overline{\Phi}$ the density, CDF, and survival function of $N(0, 1)$. We also denote by $\sigma(t) := \frac{1}{1 + e^{-t}}$.

Now, our first goal is to estimate $\mu(R_\eps \oplus G_\nu)$.

\begin{lemma}
    \label{lemma:pure_gdp_exact}
    For $\eps, \nu > 0$, let $\nu_\eps := \mu(R_\eps \oplus G_\nu)$. Then
    \begin{align*}
        \Phi(\nu_\eps/2) = \sigma(\eps) \overline{\Phi}\left( -\frac{\nu}{2} - \frac{\eps}{\nu} \right) + \sigma(-\eps) \overline{\Phi}\left( -\frac{\nu}{2} + \frac{\eps}{\nu} \right).
    \end{align*}
\end{lemma}

\begin{proof}
    Let $L := R_\eps \oplus G_\nu$, let $\overline{F}_L$ be its survival function, and let $f_L$ be its PDF. As $L$ is symmetric, the $f$-DP characterization of GDP \cite{DongRS19} yields
    \begin{align*}
        \nu_\eps
            = \sup_{t \in \R} \Phi^{-1}\left( \mP_{Z \sim L}(Z > t) \right) - \Phi^{-1}\left( \mP_{Z \sim L}(Z < -t) \right)
            = \sup_{t \in \R} \underbrace{\Phi^{-1}(\overline{F}_L(t)) + \Phi^{-1}(\overline{F}_L(-t))}_{=: g(t)}.
    \end{align*}
    We claim that the supremum is attained at $t = 0$ as, in this case,
    \begin{align*}
        \nu_\eps
            = g(0)
            = 2\Phi^{-1}(\overline{F}_{R_\eps \oplus G_\nu}(0))
            = 2\Phi^{-1}\left( \sigma(\eps)\overline{\Phi}\left( -\frac{\nu}{2} - \frac{\eps}{\nu} \right) + \sigma(-\eps)\overline{\Phi}\left( -\frac{\nu}{2} + \frac{\eps}{\nu} \right) \right)
    \end{align*}
    and we are done.

    To this end, first notice that $g(t) = g(-t)$, so it is enough show that $g'(t) \leq 0$ for $t > 0$. Indeed, as $L$ is symmetric, this implies that $f_L(-t) = e^{-t}f_L(t)$, so
    \begin{align*}
        g'(t)
            = \frac{-f_L(t)}{\phi(\Phi^{-1}(\overline{F}_L(t)))} + \frac{f_L(-t)}{\phi(\Phi^{-1}(\overline{F}_L(-t)))}
            = f_L(t) \cdot \left( \frac{e^{-t}}{\phi(\Phi^{-1}(\overline{F}_L(-t)))} - \frac{1}{\phi(\Phi^{-1}(\overline{F}_L(t)))} \right).
    \end{align*}
    That is, we need to show that
    \begin{align*}
        r(t) := \ln(\phi(\Phi^{-1}(\overline{F}_L(t)))) - \ln(\phi(\Phi^{-1}(\overline{F}_L(-t)))) - t \leq 0
    \end{align*}
    for $t > 0$.

    By the intermediate value theorem, it is enough to argue that for any $t \geq 0$ with $r(t) = 0$, we have $r'(t) < 0$ as this will force $r(t) \leq 0$ for $t \geq 0$. To this end, let $t \geq 0$ and assume that $r(t) = 0$. For short, let $u := \Phi^{-1}(\overline{F}_L(t))$, $v := \Phi^{-1}(\overline{F}_L(-t))$, and let $m := g(t) = u + v$. As $r(t) = 0$, we have $e^t = \frac{\phi(u)}{\phi(v)}$. In particular, $t = \frac{v^2 - u^2}{2} = \frac{m(v - u)}{2}$, so we get
    \begin{align*}
        \begin{cases}
            v - u = \frac{2t}{m} \\
            v + u = m
        \end{cases}
        \implies
        \begin{cases}
            u = \frac{m}{2} - \frac{t}{m} \\
            v = \frac{m}{2} + \frac{t}{m}
        \end{cases}.
    \end{align*}
    Therefore,
    \begin{align*}
        \phi(u)
            = \phi\left( \frac{m}{2} - \frac{t}{m} \right)
            = m \cdot \frac{1}{m}\phi\left( \frac{t}{m} - \frac{m}{2} \right)
            = mf_{G_m}(t)
    \end{align*}
    where $f_{G_m}$ is the PDF of $G_m = N(m^2/2, m^2)$. In particular, since $\frac{d}{dt}[\Phi^{-1}(\overline{F}_L(t))] = -\frac{f_L(t)}{\phi(\Phi^{-1}(\overline{F}_L(t)))}$, we have
    \begin{align*}
        r'(t)
            & = \frac{uf_L(t)}{\phi(u)} + \frac{vf_L(-t)}{\phi(v)} - 1 \\
            & = \frac{uf_L(t)}{\phi(u)} + \frac{ve^tf_L(-t)}{e^t\phi(v)} - 1 \\
            & = \frac{uf_L(t)}{\phi(u)} + \frac{vf_L(t)}{\phi(u)} - 1 \tag{Symmetry of $L$, $e^t = \frac{\phi(u)}{\phi(v)}$} \\
            & = \frac{mf_L(t)}{\phi(u)} - 1 \\
            & = \frac{f_L(t)}{f_{G_m}(t)} - 1.
    \end{align*}
    Now, let $h(s) := f_L(s)/f_{G_m}(s)$ and assume toward a contradiction that $h(t) \geq 1$.

    We claim that $0 = \int_t^\infty (h(s) - 1) f_{G_m}(s) \, ds = \int_t^\infty (h(s) - 1)e^{-s} f_{G_m}(s) \, ds$. Indeed, $\overline{F}_L(t) = \Phi(u) = \Phi(\frac{m}{2} - \frac{t}{m}) = \overline{F}_{G_m}(t)$, so
    \begin{align*}
        \int_t^\infty (h(s) - 1) f_{G_m}(s) \, ds
            = \int_t^\infty f_L(s) \, ds - \int_t^\infty f_{G_m}(s) \, ds
            = \overline{F}_L(t) - \overline{F}_{G_m}(t)
            = 0.
    \end{align*}
    Analogously, we have $\overline{F}_L(-t) = \Phi(v) = \Phi(\frac{m}{2} + \frac{t}{m}) = \overline{F}_{G_m}(-t)$, so by symmetry of $L$ and $G_m$, we have $h(s) = h(-s)$ as well as
    \begin{align*}
        \int_t^\infty (h(s) - 1)e^{-s} f_{G_m}(s) \, ds
            & = \int_t^\infty (h(-s) - 1) f_{G_m}(-s) \, ds \\
            & = \int_t^\infty f_L(-s) \, ds - \int_t^\infty f_{G_m}(-s) \, ds \\
            & = F_L(t) - F_{G_m}(t) \\
            & = 0.
    \end{align*}

    On the other hand, recalling that $L = R_\eps \oplus G_\nu$, we have
    \begin{align*}
        \ln(h(s))
            = \ln\left( \frac{f_L(s)}{f_{G_m}(s)} \right)
            = C + \left( \frac{1}{2m^2} - \frac{1}{2\nu^2} \right)s^2 + \ln\left( \cosh\left( \frac{\eps s}{\nu^2} \right) \right)
    \end{align*}
    for a constant $C$, so it follows that
    \begin{align*}
        \frac{d}{ds}[\ln(h(s))]
            = \left( \frac{1}{m^2} - \frac{1}{\nu^2} \right)s + \frac{\eps}{\nu^2}\tanh(\eps s / \nu^2)
            = s\left( \frac{1}{m^2} - \frac{1}{\nu^2} + \frac{\eps}{s\nu^2}\tanh\left( \eps s / \nu^2 \right) \right).
    \end{align*}
    As $\tanh(x)/x$ is strictly decreasing in $x > 0$, it follows that $\frac{d}{ds}[\ln(h(s))]$ either has no roots in $(0, \infty)$ and is strictly positive or negative over that region, or it has exactly one root in $(0, \infty)$ before which $\frac{d}{ds}[\ln(h(s))] > 0$ and after which $\frac{d}{ds}[\ln(h(s))] < 0$. In either case, as $h(t) \geq 1$ by assumption, either $h(s) > 1$ for all $s > t$ or there is $w > t$ such that $h(s) \geq 1$ for $t \leq s \leq w$ and $h(s) \leq 1$ for $s \geq w$. In the former case, if $h(s) > 1$ for all $s > t$, then we could not have $\int_t^\infty (h(s) - 1) f_{G_m}(s) \, ds = 0$, so it must be the case that $h|_{[t, w]} \geq 1$ and $h|_{[w, \infty)} \leq 1$ for some $w \geq t$. In this case, however, since $h$ is cannot be identically $1$, we have
    \begin{align*}
        \int_t^\infty (h(s) - 1)e^{-s} f_{G_m}(s) \, ds
            & = \int_t^w (\underbrace{h(s) - 1}_{\geq 0})e^{-s} f_{G_m}(s) \, ds - \int_w^\infty (\underbrace{1 - h(s)}_{\geq 0})e^{-s} f_{G_m}(s) \, ds \\
            & > e^{-w}\int_t^w (h(s) - 1) f_{G_m}(s) \, ds - e^{-w}\int_w^\infty (1 - h(s)) f_{G_m}(s) \, ds \\
            & = e^{-w}\int_t^\infty (h(s) - 1) f_{G_m}(s) \, ds \\
            & = 0,
    \end{align*}
    which is also a contradiction. Hence $h(t) < 1$ and so $r'(t) < 0$ as desired.
\end{proof}

\begin{lemma}
    \label{lemma:pure_gdp_estimate}
    For any $\nu, \eps > 0$, we have
    \begin{align*}
        \mu(R_\eps \oplus G_\nu)^2 \leq \nu^2 + 8\ln(\cosh(\eps/2)) + \frac{4\eps^2}{\nu^2}.
    \end{align*}
\end{lemma}

\begin{proof}
    We will rely on two inequalities from survival analysis. Let $h(t) := \phi(t)/\overline{\Phi}(t)$ be the hazard function of the standard normal. Gordon's inequality asserts that $h(t) \geq t$ for all $t \in \R$ and Sampford's inequality asserts that $h'(t) \leq 1$ for all $t \in \R$ (see~\cite{sampford1953some}).

    In particular, noticing that $\frac{d}{dt}[\ln(\overline{\Phi}(t))] = -h(t)$, we have by Sampford's inequality that $\frac{d^2}{dt^2}[\ln(\overline{\Phi}(t))] \geq -1$, so Taylor's theorem yields
    \begin{align*}
        \ln(\overline{\Phi}(t + u)) \geq \ln(\overline{\Phi}(t)) - h(t)u - \frac{u^2}{2}
    \end{align*}
    for all $t, u \in \R$.

    Applying this to \Cref{lemma:pure_gdp_exact}, we get for $\nu_\eps := \mu(R_\eps \oplus G_\nu)$ that
    \begin{align*}
        \overline{\Phi}(\nu_\eps/2)
            & = \sigma(\eps)\overline{\Phi}(\frac{\nu}{2} + \frac{\eps}{\nu}) + \sigma(-\eps)\overline{\Phi}(\frac{\nu}{2} - \frac{\eps}{\nu}) \\
            & = \frac{1}{2\cosh(\eps/2)}\left( e^{\eps/2}\overline{\Phi}(\frac{\nu}{2} + \frac{\eps}{\nu}) + e^{-\eps/2}\overline{\Phi}(\frac{\nu}{2} - \frac{\eps}{\nu}) \right) \\
            & \geq \frac{1}{2\cosh(\eps/2)}\left( e^{\eps/2}\overline{\Phi}(\nu/2)e^{-h(\nu/2)\frac{\eps}{\nu} - \frac{\eps^2}{2\nu^2}} + e^{-\eps/2}\overline{\Phi}(\nu/2)e^{h(\nu/2)\frac{\eps}{\nu} - \frac{\eps^2}{2\nu^2}} \right) \\
            & = \frac{\overline{\Phi}(\nu/2)e^{-\frac{\eps^2}{2\nu^2}}\cosh\left( \eps/2 - h(\nu/2)\frac{\eps}{\nu} \right)}{\cosh(\eps/2)} \\
            & \geq \frac{\overline{\Phi}(\nu/2)e^{-\frac{\eps^2}{2\nu^2}}}{\cosh(\eps/2)}.
    \end{align*}
    In particular, we have
    \begin{align*}
        \ln\left( \frac{\overline{\Phi}(\nu/2)}{\overline{\Phi}(\nu_\eps/2)} \right)
            \leq \ln(\cosh(\eps/2)) + \frac{\eps^2}{2\nu^2}.
    \end{align*}

    On the other hand, recalling that $\frac{d}{dt}[\ln(\overline{\Phi}(t))] = -h(t)$, Gordon's inequality yields
    \begin{align*}
        \ln\left( \frac{\overline{\Phi}(x)}{\overline{\Phi}(y)} \right)
            = \int_x^y h(t) \, dt
            \geq \int_x^y t \, dt
            = \frac{y^2 - x^2}{2}.
    \end{align*}
    for any $y \geq x$. Altogether, we have
    \begin{align*}
        \frac{(\nu_\eps/2)^2 - (\nu/2)^2}{2} \leq \ln(\cosh(\eps/2)) + \frac{\eps^2}{2\nu^2}
    \end{align*}
    and the desired estimate follows.
\end{proof}

These estimates on $\mu(R_\eps \oplus G_\nu)$ allow us to upper bound the per-step budget cost of the na\"ive and residue GDP filters. For lower bounding these costs, we rely on Bhattacharyya coefficients (also called Hellinger affinity). See~\cite{raskhodnikova2026local} for other applications of this tool.

\begin{proposition}
    \label{prop:bhattacharyya}
    For a PLD $L$, let $\BC(L) := \E_{Z \sim L}[e^{-Z/2}]$. Then, for any PLDs $L$ and $L'$, we have
    \begin{enumerate}
        \item $L \preceq L' \implies \BC(L) \geq \BC(L')$; and
        \item $\BC(L \oplus L') = \BC(L)\BC(L')$.
    \end{enumerate}
\end{proposition}

The first property follows by noticing that the Hellinger affinity is an $F$-divergence, which are well known to respect the Blackwell order (see e.g. Appendix B of~\cite{DongRS19}). The second property follows immediately from properties of expectations.

\begin{proof}[Proof of \Cref{thm:gdp_residue_pure_query}]
    For the first part, simply notice that $\mu^2 - \mu_{\textrm{na\"ive}}^2 = \mu(R_\eps \oplus G_0)^2$ and apply \Cref{lemma:pure_gdp_exact}.

    As for the second part, we have that $\mu = \mu(R_\eps \oplus G_{\mu_{\textrm{res}}})$, so \Cref{lemma:pure_gdp_estimate} yields
    \begin{align*}
        \mu^2 - \mu_{\textrm{res}}^2 \leq 8\ln(\cosh(\eps/2)) + \frac{4\eps^2}{\mu_{\textrm{res}^2}}.
    \end{align*}

    Finally, as for the reverse inequality, notice that
    \begin{align*}
        \BC(R_\eps)
            = \sigma(\eps)e^{-\eps/2} + \sigma(-\eps)e^{\eps/2}
            = \frac{2}{e^{\eps/2} + e^{-\eps/2}}
            = \frac{1}{\cosh(\eps/2)}.
    \end{align*}
    Moreover, since $R_\eps \oplus G_{\mu_{\textrm{res}}} \preceq G_\mu$, by \Cref{prop:bhattacharyya} and by appealing to the moment generating function of $G_\nu = N(\nu^2/2, \nu^2)$, we get
    \begin{align*}
        e^{-\mu^2/8}
            = \BC(G_\mu)
            \leq \BC(R_\eps \oplus G_{\mu_{\textrm{res}}})
            = \BC(R_\eps)\BC(G_{\mu_{\textrm{res}}})
            = \frac{e^{-\mu_{\textrm{res}}^2/8}}{\cosh(\eps/2)}.
    \end{align*}
    Rearranging yields $\mu^2 - \mu_{\textrm{res}}^2 \geq 8\ln(\cosh(\eps/2))$ as desired.
\end{proof}

\begin{proof}[Proof of \Cref{thm:gdp_residue_many_pure_queries}]
    By \Cref{thm:gdp_residue_pure_query}, the na\"ive GDP filter deducts exactly $C_{\textrm{na\"ive}}(\eps_i)$ from an initial square budget of $\mu^2$ for every $\eps_i$-DP query it receives, so we immediately get that $N_{\textrm{na\"ive}}(\mu) = \max\{n : \sum_{i = 1}^n C_{\textrm{na\"ive}}(\eps_i) \leq \mu^2\}$.

    Moreover, if an adversary non-adaptively issues the queries $\eps_1, \eps_2, \dots, \eps_{N_{\textrm{res}}(\mu)}$ to the GDP residue filter, it passes by construction, so, because the GDP residue filter is free, the privacy cost of this non-adaptive adversary is
    \begin{align*}
        R_{\eps_1} \oplus \dots \oplus R_{\eps_{N_{\textrm{res}}(\mu)}} \preceq G_\mu
    \end{align*}
    and hence $N_{\textrm{res}}(\mu) \leq N_{\textrm{max}}(\mu)$.

    Next, we bound $N_{\textrm{max}}(\mu)$ by way of Bhattacharyya coefficients. Indeed, if $R_{\eps_1} \oplus \dots \oplus R_{\eps_n} \preceq G_\mu$, then by \Cref{prop:bhattacharyya}, we must have
    \begin{align*}
        e^{-\mu^2/8}
            = \BC(G_\mu)
            \leq \BC(R_{\eps_1} \oplus \dots \oplus R_{\eps_n})
            = \prod_{i = 1}^n \BC(R_{\eps_i})
            = \frac{1}{\prod_{i = 1}^n \cosh(\eps_i/2)}.
    \end{align*}
    Taking logs, we get $N_{\textrm{max}}(\mu) \leq \max\{n : \sum_{i = 1}^n C_{\textrm{res}}(\eps_i) \leq \mu^2 \}$ as desired.

    Finally, it remains to lower bound $N_{\textrm{res}}(\mu)$. To this end, let $\mu_i^2$ denote the square budget remaining after $i$ steps of applying the GDP residue filter ($\mu_0 = \mu$). Furthermore, let
    \begin{align*}
        K_E := \sup_{0 < \eps \leq E} \frac{4\eps^2}{8\ln(\cosh(\eps/2))}.
    \end{align*}
    By a straightforward power series analysis, we have that $K_E = O(E)$. Now, let $N$ be maximal such that $\mu_0^2 \geq \dots \geq \mu_N^2 \geq \mu$. In particular, $N_{\textrm{res}}(\mu) \geq N$, so we will need to argue that $N$ is large. Indeed, for each $0 \leq i < N$, we have by \Cref{thm:gdp_residue_pure_query} that
    \begin{align*}
        \mu_i^2 - \mu_{i + 1}^2
            \leq C_{\textrm{res}}(\eps_i) + \frac{4\eps_i^2}{\mu_{i + 1}^2}
            \leq C_{\textrm{res}}(\eps_i) + \frac{K_E C_{\textrm{res}}(\eps_i)}{\mu}
            = C_{\textrm{res}}(\eps_i)\left( 1 + \frac{K_E}{\mu} \right).
    \end{align*}
    On the other hand, $\mu_0^2 = \mu^2$ and, by maximality of $N$, $\mu_{N + 1}^2 < \mu$, so we must have
    \begin{align*}
        \sum_{i = 1}^{N + 1} C_{\textrm{res}}(\eps_i)
            > \frac{\mu^2 - \mu}{1 + \frac{K_E}{\mu}}
            = \mu^2 - \left( \frac{1 + K_E}{1 + K_E/\mu} \right) \mu
            = \mu^2 - O_E(\mu),
    \end{align*}
    so $N_{\textrm{res}}(\mu) \geq N \geq \max\{n : \sum_{i = 1}^n C_{\textrm{res}}(\eps_i) \leq \mu^2 - O_E(\mu)\}$ as desired.

    We now turn to the average case analysis of queries $\eps_i \sim \textrm{Uniform}([0, E])$. To this end, recall Wald's identity (see e.g. Theorem 3.3.2 of~\cite{ross1995stochastic}), which asserts that, for i.i.d. random variables $0 \leq X_1, X_2, \dots \leq M$ and $N := \max\{n : \sum_{i = 1}^n X_i \leq C\}$, $N + 1$ is a stopping time such that
    \begin{align*}
        C < (\E[N] + 1)\E[X_1] \leq C + M.
    \end{align*}
    In particular, if $M = O(1)$ and $C = \omega(1)$, we have $\E[N] \sim \frac{C}{\E[X_1]}$.

    Now, by power series analysis, we have $0 \leq C_{\textrm{res}}(\eps_i) = O(E)$ and $A_{\textrm{res}} := \E[C_{\textrm{res}}(\eps_1)] = O(E)$, so, applying Wald's identity to $\E[N_{\textrm{res}}(\mu)]$ and $\E[N_{\textrm{max}}(\mu)]$, we get
    \begin{align*}
        1
            \leq \frac{\E[N_{\textrm{max}}(\mu)]}{\E[N_{\textrm{res}}(\mu)]}
            \leq \frac{\E[\max\{n : \sum_{i = 1}^n C_{\textrm{res}}(\eps_i) \leq \mu^2\}]}{\E[\max\{n : \sum_{i = 1}^n C_{\textrm{res}}(\eps_i) \leq \mu^2 - O_E(\mu)\}]}
            \sim \frac{\frac{\mu^2}{A_{\textrm{res}}}}{\frac{\mu^2 - O_E(\mu)}{A_{\textrm{res}}}}
            = \frac{\mu^2}{\mu^2 - O_E(\mu)}
            \sim 1
    \end{align*}
    in the regime $\mu = \omega(1)$ and $E = O(1)$.

    Likewise, we have $0 \leq C_{\textrm{na\"ive}}(\eps_i) = O(E)$ and $A_{\textrm{na\"ive}} := \E[C_{\textrm{na\"ive}}(\eps_1)]$. By Wald's identity followed by L'H\^opital and FTC part I, we have
    \begin{align*}
        \frac{\E[N_{\textrm{res}}(\mu)]}{\E[N_{\textrm{na\"ive}}(\mu)]}
            \sim \frac{\E[\max\{n : \sum_{i = 1}^n C_{\textrm{res}}(\eps_i) \leq \mu^2\}]}{\E[\max\{n : \sum_{i = 1}^n C_{\textrm{na\"ive}}(\eps_i) \leq \mu^2 \}]}
            \sim \frac{\frac{\mu^2}{A_{\textrm{res}}}}{\frac{\mu^2}{A_{\textrm{na\"ive}}}}
            = \frac{A_{\textrm{na\"ive}}}{A_{\textrm{res}}}
            \sim \frac{C_{\textrm{na\"ive}}(E)}{C_{\textrm{res}}(E)}
            \sim \pi/2
    \end{align*}
    in the regime $\mu = \omega(1)$ and $E = o(1)$.
\end{proof}

%% file: sections/app_implications.tex
\section{Proofs for Section~\ref{sec:implications-special-case}}
\label{app:implications-special-case}

\begin{proof}[Proof of \Cref{prop:piecewise-linear-composition}]
    For each $i \in \{1,2\}$, let $f_i$ be a tradeoff curve with $n_i$ linear segments and $f_i(0) =1 -\delta_i$.  Each linear segment $j \in [n_i]$ has a slope $f'_{i,j} \le 0$ over an interval $[\alpha_{i,j}, \alpha_{i,j+1})$.
    The width of interval $j$ is  denoted by $w_{i,j} $. By convexity, we have $-f'_{i,1} \ge -f'_{i,2 }\ge \cdots \ge -f'_{i,n_i } $.
    Note that we have $\alpha_{i,1} = 0$, $\alpha_{i,n_i +1} = 1$, $f_i(1) = 0$, and $\sum_{j =1}^{n_i}w_{i,j} = 1$.

    For each $i \in \{1,2 \}$, we construct a probability distribution $Q_i$ supported by $[0,1]$ with density function as follows:
    \begin{equation}
         q_i(x) = \begin{cases}
            -f'_{i,1}, & \text{if } x \in [\alpha_{i,1}, \alpha_{i,2}),\\
             -f'_{i,2}, & \text{if }  x \in [\alpha_{i,2}, \alpha_{i,3}), \\
            \vdots & \\
             -f'_{i,j}, & \text{if }  x \in [\alpha_{i,j}, \alpha_{i,j+1}), \\
            \vdots & \\
             -f'_{i,n_i}, & \text{if }  x \in [\alpha_{i,n_i}, \alpha_{i,n_i+1}),\\
             (1-f_i(0)) \cdot \delta(x-1),  & \text{if }  x = \alpha_{i,n_i+1}.
        \end{cases} 
    \end{equation}
    Let $P$ be the uniform distribution over $[0,1]$. Now, we show that $T(P,Q_i) = f_i$.

    When $\alpha \in [\alpha_{i,1}, \alpha_{i,2})$, the optimal rejection rule is
    \begin{equation*}
         \phi(x) = \begin{cases}
            \alpha/w_{i,1}, & \text{if }  x \in [\alpha_{i,1}, \alpha_{i,2}),\\
            1, & \text{if } x = 1, \\
            0, & \text{otherwise}.
        \end{cases} 
    \end{equation*}
    Note that the type I error is $\mathbb{E}_P[\phi] = \alpha = c \cdot w_{i,1}$. The type II error is 
    \begin{equation*}
        \begin{array}{l}
    1- \mathbb{E}_{Q_i}[\phi] =  1 -c \cdot w_{i,1} \cdot (-f'_{i,1}) -  (1-f_i(0)) = f_i(0) +\alpha f'_{i,1} = f_i(\alpha)= T(P,Q_i)(\alpha).
        \end{array}
    \end{equation*}

    When $\alpha \in [\alpha_{i,2}, \alpha_{i,3})$, the optimal rejection rule is 
    \begin{equation*}
         \phi(x) = \begin{cases}
            1, & \text{if }  x \in [\alpha_{i,1}, \alpha_{i,2}) \bigcup \{1\},\\
             c, & \text{if }  x \in [\alpha_{i,2}, \alpha_{i,3}),\\
            0, & \text{otherwise}.
        \end{cases} 
    \end{equation*}
    Note that the type I error is $\mathbb{E}_P[\phi] = \alpha = w_{i,1} + c \cdot w_{i,2}$.  The type II error is 
    \begin{equation*}
        \begin{array}{l}
    1- \mathbb{E}_{Q_i}[\phi] =  f_i(0) -     w_{i,1} \cdot (-f'_{i,1}) - c \cdot w_{i,2} \cdot (-f'_{i,2}) = f_i(0)+w_{i,1} \cdot  f'_{i,1} + (\alpha - w_{i,1})  f'_{i,2}  = f_i(\alpha)=T(P,Q_i)(\alpha).
        \end{array}
    \end{equation*}

    More generally, when $\alpha \in [\alpha_{i,j}, \alpha_{i,j+1})$, the optimal rejection rule is 
    \begin{equation*}
         \phi(x) = \begin{cases}
            1, & \text{if }  x \in [\alpha_{i,1}, \alpha_{i,j-1}) \bigcup \{1 \},\\
             c, & \text{if }  x \in [\alpha_{i,j}, \alpha_{i,j+1}),\\
            0, & \text{otherwise}.
        \end{cases} 
    \end{equation*}
    Note that the type I error is $\mathbb{E}_P[\phi] = \alpha = \sum_{s=1}^{j-1} w_{i,s} + c \cdot w_{i,j}$. The type II error is 
    \begin{equation*}
        \begin{array}{lll}
    1- \mathbb{E}_{Q_i}[\phi] &= & f_i(0) -     \sum_{s=1}^{j-1}w_{i,s} \cdot (-f'_{i,s}) - c \cdot w_{i,j} \cdot (-f'_{i,j})\\
    & = & f_i(0)+\sum_{s=1}^{j-1}w_{i,s} \cdot  f'_{i,s} + (\alpha -\sum_{s=1}^{j-1} w_{i,s})  f'_{i,j}  = f_i(\alpha)=T(P,Q_i)(\alpha).
        \end{array}
    \end{equation*}

    Now, we construct a testing problem $P \times P \text{ vs } Q_1 \times Q_2 $. 
    We first compute $(-f'_{1,k_1}) \cdot (-f'_{2,k_2})$ for all $k_1 \in [n_1]$ and $k_2 \in [n_2]$. Then, we sort all these $n_1 \cdot n_2$ values  in a non-increasing order, using $k \in [n_1 \cdot n_2]$ as the index. For a fixed $k$, we know  its respective $k_1, k_2$ and then $w_{1, k_1}, w_{2, k_2}$ and $f'_{1, k_1}, f'_{2, k_2}$. 
    Let $\tilde{w}_k := w_{1, k_1} \cdot w_{2, k_2}$ and  $\tilde{f}'_k := f'_{1, k_1} \cdot f'_{2, k_2}$.




    \begin{enumerate}
        \item When $\alpha \in \left[0, \ \tilde{w}_1 \right)$, we have the optimal rejection rule is 
        \begin{equation*}
         \phi(x) = \begin{cases}
            1, & \text{if }  (x_1, x_2) \in [0, 1) \times \{1\} \bigcup \{1\} \times [0, 1) \bigcup \{1\} \times \{1\} \\
        
             c, & \text{if }  (x_1, x_2) \in [\alpha_{1,1}, \alpha_{1,2}) \times [\alpha_{2,1}, \alpha_{2,2})  \\
            0, & \text{otherwise}.
        \end{cases} 
    \end{equation*}
    The type I error is $\mathbb{E}_{P \times P}[\phi] = \alpha = c \cdot \tilde{w}_{1}$. The type II error is
    \begin{equation*}
         \mathbb{E}_{Q_1 \times Q_2}[\phi] = \delta_1 + \delta_2 + \delta_1 \cdot  \delta_2 +
            c \cdot \tilde{w}_{1} \cdot \tilde{f}'_{1}   =  \delta_1 - \delta_2 - \delta_1   \delta_2 -\alpha \cdot \tilde{f}'_{1}.
    \end{equation*}
    Therefore, we have
    \begin{equation*}
        \begin{array}{ll}
        
            &  T(P \times P, Q_1 \times Q_2) (\alpha) = 1- \mathbb{E}_{Q_1 \times Q_2}[\phi] =  1 - \delta_1 - \delta_2 - \delta_1   \delta_2 -\alpha \cdot \tilde{f}'_{1}.
                \end{array}
        \end{equation*}

        \item When $\alpha \in \left[\tilde{w}_{1} , \ \tilde{w}_1 + \tilde{w}_{2} \right)$, we have
        \begin{equation*}
            \begin{array}{l}
                 \alpha = \tilde{w}_{1}  + c \cdot  \tilde{w}_{2} \quad \Rightarrow c \cdot  \tilde{w}_{2}  = \alpha - \tilde{w}_1,  \\
                 T(P \times P, Q_1 \times Q_2) (\alpha)  =  1 - \delta_1 - \delta_2 - \delta_1   \delta_2- \tilde{w}_1 \cdot \tilde{f}'_1- c \cdot \tilde{w}_2 \cdot \tilde{f}'_2 \\
                 =  1 - \delta_1 - \delta_2 - \delta_1   \delta_2- \tilde{w}_1 \cdot \tilde{f}'_1- (\alpha - \tilde{w}_1) \cdot \tilde{f}'_2.
            \end{array}
        \end{equation*}
        \item When $\alpha \in \left[ \tilde{w}_1 + \tilde{w}_2, \  \tilde{w}_1 + \tilde{w}_2 + \tilde{w}_3  \right)$, we have
        \begin{equation*}
            \begin{array}{l}
                 \alpha = \sum_{s = 1}^{2}\tilde{w}_s + c \cdot \tilde{w}_3 \quad \Rightarrow c \cdot \tilde{w}_3 =  \alpha - \sum_{s = 1}^{2}\tilde{w}_s, \\
                 T(P \times P, Q_1 \times Q_2) (\alpha) = 1 - \delta_1 - \delta_2 - \delta_1   \delta_2- \sum_{s =1}^{2} \tilde{w}_s \cdot \tilde{f}'_{s} - c \cdot \tilde{w}_3 \cdot \tilde{f}'_3 \\ =  1 - \delta_1 - \delta_2 - \delta_1   \delta_2- \sum_{s =1}^{2} \tilde{w}_s \cdot \tilde{f}'_s - (\alpha - \sum_{s=1}^2 \tilde{w}_s) \cdot   \tilde{f}'_3. 
            \end{array}
        \end{equation*}

        \item When $\alpha \in \left[\sum_{s=1}^{j}  \tilde{w}_s , \ \sum_{s=1}^{j+1} \tilde{w}_s \right)$,  we have
        \begin{equation*}
        \begin{array}{l}
            \alpha = \sum_{s = 1}^{j}\tilde{w}_s + c \cdot \tilde{w}_{j+1} \quad \Rightarrow c \cdot \tilde{w}_{j+1} =  \alpha - \sum_{s = 1}^{j}\tilde{w}_s, \\
         T(P \times P, Q_1 \times Q_2) (\alpha) = 1- \delta_1 - \delta_2 - \delta_1   \delta_2- \sum_{s =1}^{j} \tilde{w}_s \cdot \tilde{f}'_{s} - c \cdot \tilde{w}_{j+1} \cdot \tilde{f}'_{j+1} \\ =  1- \delta_1 - \delta_2 - \delta_1   \delta_2- \sum_{s =1}^{j} \tilde{w}_s \cdot \tilde{f}'_s - (\alpha - \sum_{s=1}^j \tilde{w}_s) \cdot   \tilde{f}'_{j+1}.   \end{array}
        \end{equation*}
    \end{enumerate}
\end{proof}

%% file: sections/app_symmetry.tex
\section{On Symmetric PLDs and Dominating Pairs}
\label{app:symmetry}


For the purpose of accurate privacy accounting, we often need to identify the ``worst'' pairs of datasets and its corresponding privacy loss. This concept can be formalized via the notion of \emph{(worst-case) dominating pairs} as follows:

\begin{definition}[Dominating Pair~\cite{ZhuDW22}]
%
$(P, Q)$ is a \emph{dominating pair} for a mechanism $\mech$ if 
$(P, Q) \succeq (\mech(D), \mech(D'))$ for all neighbouring datasets $D, D'$. Furthermore, it is said to be a \emph{worst-case dominating pair} if additionally $P = \mech(D), Q = \mech(D')$ for some neighbouring datasets $D, D'$. 
\end{definition}

The privacy characteristics of a mechanism may not be determined by a single pair of datasets. For the purposes of constructing filters, this can be handled by considering multiple dataset pairs that characterize some regime of the worst-case privacy characteristics of a mechanism. One can also smooth over the choices of datasets by considering dominating pairs, though this is lossy. See~\cite{ZhuDW22} for more discussion.

In practice, however, a typical privacy accounting is done by taking each $L_i$ to be the PLD of a worst-case dominating pair of the mechanism $\mech_i$. We observe below that such PLD is symmetric. Thus, the assumption in \Cref{sec:natural-apx} is quite a natural one that is already satisfied in practice.

\begin{lemma} \label{lem:worst-case-symmetric}
If $(P, Q)$ is a worst-case dominating pair of $\mech$, then $(P, Q)$ is symmetric.
\end{lemma}

\begin{proof}
Since $(P, Q)$ is a worst-case dominating pair, we have $P = \mech(D)$ and $Q = \mech(D')$ for some neighbouring datasets $D, D'$ and $(P, Q) \succeq (\mech(D'), \mech(D)) = (Q, P)$. By definition, there exists a Markov kernel $\kappa$ such that $\kappa P = Q$ and $\kappa Q = P$, but this also implies that $(Q, P) \succeq (P, Q)$.
\end{proof}

We will next prove \Cref{lem:sym-check-only-large-alpha}. To do so, let us start by stating a relationship between the hockey-stick divergence at $x$ and $x^{-1}$, which will be convenient for our proof.

\begin{lemma}[\cite{ZhuDW22}] \label{lem:hockey-inverse-alpha}
For any pair $(P, Q)$ and $x \in \bR^\times$, we have $H_x(Q \parallel P) = x \cdot H_{x^{-1}}(P \parallel Q) - x + 1$.
\end{lemma}

\begin{proof}[Proof of \Cref{lem:sym-check-only-large-alpha}]
Let $L = \pld(P, Q), L' = \pld(P',Q')$ where $(P, Q), (P',Q')$ are symmetric pairs.
The forward direction is obvious (from definition of dominating). For the backward direction, if $0 < x < 1$, then \Cref{lem:hockey-inverse-alpha} implies that
\begin{align*}
H_x(P \parallel Q) \geq H_x(P' \parallel Q') \Leftrightarrow  H_{x^{-1}}(Q \parallel P) \geq H_{x^{-1}}(Q' \parallel P') \Leftrightarrow H_{x^{-1}}(P \parallel Q) \geq H_{x^{-1}}(P' \parallel Q'),
\end{align*}
where the second equivalence is due to the symmetry assumption.
\end{proof}